\documentclass[lettersize,onecolumn]{IEEEtran}
\usepackage[utf8]{inputenc}

\def\PrintMode{0}

\usepackage{tikz}
\usetikzlibrary{positioning}
\usepackage{pgffor}
\usetikzlibrary{decorations.pathreplacing}
\usepackage{expl3}
\usepackage[T1]{fontenc}
\usepackage{graphicx}
\usepackage{float}
\usepackage[lined,ruled,linesnumbered]{algorithm2e}
\usepackage{xspace}
\usepackage[utf8]{inputenc}
\usepackage[noadjust]{cite}
\usepackage{amsmath, amssymb, amsfonts, amsthm}
\usepackage{enumerate}
\usepackage{xcolor}
\usepackage[all]{xy}
\usepackage[linktocpage=true,pagebackref=true,colorlinks,linkcolor=magenta,citecolor=blue,bookmarks,bookmarksopen,bookmarksnumbered]{hyperref}
\usepackage{bm}
\usepackage{bbm}
\usepackage[normalem]{ulem}
\usepackage{paralist}
\usepackage{gensymb}
\usepackage{thmtools}
\usepackage{rotating}
\usepackage{mdframed}
\usepackage{multirow}
\usepackage{tabularx}
\usepackage{verbatim}
\usepackage{mathrsfs}
\usepackage{indentfirst}
\usepackage{mleftright}
\usepackage[nottoc]{tocbibind}
\usepackage{complexity}
\usepackage{tablefootnote}
\usepackage{epigraph}
\usepackage{autobreak}
\usepackage{todonotes}

\ifnum\PrintMode=1

\usepackage{savetrees}
\else
\fi

\graphicspath{{./figs/}}

\usepackage[capitalize]{cleveref}
\newtheorem{theorem}{Theorem}[section]
\newtheorem{lemma}[theorem]{Lemma}
\newtheorem{proposition}[theorem]{Proposition}
\newtheorem{corollary}[theorem]{Corollary}
\newtheorem{claim}[theorem]{Claim}

\theoremstyle{definition}
\newtheorem{definition}[theorem]{Definition}
\newtheorem{problem}[theorem]{Problem}

\newtheorem{remark}[theorem]{Remark}
\newtheorem*{remark*}{Remark}

\renewcommand*\backref[1]{\ifx#1\relax \else (cit.~on p.~#1) \fi}

\makeatletter
\def\moverlay{\mathpalette\mov@rlay}
\def\mov@rlay#1#2{\leavevmode\vtop{%
		\baselineskip\z@skip \lineskiplimit-\maxdimen
		\ialign{\hfil$\m@th#1##$\hfil\cr#2\crcr}}}
\newcommand{\charfusion}[3][\mathord]{
	#1{\ifx#1\mathop\vphantom{#2}\fi
		\mathpalette\mov@rlay{#2\cr#3}
	}
	\ifx#1\mathop\expandafter\displaylimits\fi}
\makeatother

\newcommand\numberthis{\addtocounter{equation}{1}\tag{\theequation}}

\renewcommand{\poly}{\mathrm{poly}}

\renewcommand{\emptyset}{\varnothing}

\newcommand{\dist}{\operatorname{dist}}

\newlang{\MCSP}{MCSP}
\newlang{\MFSP}{MFSP}
\newlang{\MKtP}{MKtP}
\newlang{\MKTP}{MKTP}
\newlang{\itrMCSP}{itrMCSP}
\newlang{\itrMKTP}{itrMKTP}
\newlang{\itrMINKT}{itrMINKT}
\newlang{\MINKT}{MINKT}
\newlang{\MINK}{MINK}
\newlang{\MINcKT}{MINcKT}
\newlang{\CMD}{CMD}
\newlang{\DCMD}{DCMD}
\newlang{\CGL}{CGL}
\newlang{\PARITY}{PARITY}
\newlang{\Empty}{\textsc{Empty}}
\newlang{\Avoid}{\textsc{Avoid}}
\newlang{\Sparsification}{\textsc{Sparsification}}
\newlang{\HamEst}{\mathsf{HammingEst}}
\newlang{\HamHit}{\mathsf{HammingHit}}
\newlang{\CktEval}{\textsc{Circuit-Eval}}
\newlang{\Hard}{\textsc{Hard}}
\newlang{\cHard}{\textsc{cHard}}
\newlang{\CAPP}{CAPP}
\newlang{\GapUNSAT}{GapUNSAT}
\newlang{\OV}{OV}
\renewlang{\PCP}{PCP}
\newlang{\PCPP}{PCPP}
\newclass{\Avg}{Avg}
\newclass{\ZPEXP}{ZPEXP}
\newclass{\DLOGTIME}{DLOGTIME}
\newclass{\ALOGTIME}{ALOGTIME}
\newclass{\ATIME}{ATIME}%
\newclass{\SZKA}{SZKA}
\newclass{\Laconic}{Laconic\text{-}}
\newclass{\APEPP}{APEPP}
\newclass{\SAPEPP}{SAPEPP}

\newclass{\TFSigma}{TF\Sigma}
\newclass{\NTIMEGUESS}{NTIMEGUESS}

\newlang{\Formula}{Formula}
\newlang{\THR}{THR}
\newlang{\EMAJ}{EMAJ}
\newlang{\MAJ}{MAJ}
\newlang{\SYM}{SYM}
\newlang{\DOR}{DOR}
\newlang{\ETHR}{ETHR}
\newlang{\Midbit}{Midbit}
\newlang{\LCS}{LCS}
\newlang{\TAUT}{TAUT}
\newlang{\Poly}{\text{-}Poly}

\newcommand{\supp}{\operatorname{supp}}

\newcommand{\Z}{\mathbb{Z}}
\newcommand{\F}{\mathbb{F}}

\newcommand{\be}{\overline{e}}
\newcommand{\row}{\mathsf{row}}
\newcommand{\col}{\mathsf{col}}

\newcommand{\Eval}{\mathsf{\mathbf{Eval}}}

\newcommand{\bB}{\mathbf{B}}

\newcommand{\bD}{\mathbf{D}}

\newcommand{\bF}{\mathbf{F}}

\newcommand{\bI}{\mathbf{I}}

\newcommand{\bL}{\mathbf{L}}
\newcommand{\bM}{\mathbf{M}}
\newcommand{\bN}{\mathbf{N}}

\newcommand{\bP}{\mathbf{P}}

\newcommand{\bY}{\mathbf{Y}}
\newcommand{\bZ}{\mathbf{Z}}

\newcommand{\RS}{\mathsf{RS}}
\newcommand{\IRS}{\mathsf{IRS}}

\renewcommand{\epsilon}{\varepsilon}
\newcommand{\eps}{\epsilon}

\newcommand{\Fail}{\mathsf{Fail}}
\newcommand{\Lagrange}{\mathsf{Lagrange}}
\newcommand{\Hermite}{\mathsf{Hermite}}
\newcommand{\Diag}{\mathsf{\mathbf{Diag}}}
\newcommand{\Lower}{\mathsf{\mathbf{Lower}}}
\newcommand{\BlockDiag}{\mathsf{\mathbf{BlockDiag}}}
\newcommand{\Vand}{\mathsf{\mathbf{Vand}}}

\newcommand{\cB}{\mathcal B}
\newcommand{\cC}{\mathcal C}

\usepackage[draft, inline,marginclue,index]{fixme} 
\fxsetup{theme=color,mode=multiuser}

\definecolor{color1}{RGB}{46,134,193}

\definecolor{color7}{RGB}{128,0,128}
\FXRegisterAuthor{yeyuan}{ayeyuan}{\color{red}Yeyuan}

\definecolor{color3}{RGB}{255,128,0}
\FXRegisterAuthor{zihan}{azihan}{\color{blue}Zihan}

\definecolor{color5}{RGB}{128,128,128}
\FXRegisterAuthor{task}{atask}{\colorbox{color5}{\color{white}Task}}

\newcommand{\change}[1]{{\color{blue}#1}}
\renewcommand{\change}[1]{{#1}}

\newif\ifmynotes
\mynotestrue

\title{Unique Decoding of Reed–Solomon and Related Codes for Semi-Adversarial Errors\thanks{An extended abstract of this work \change{appeared} in the proceedings of ICALP 2026. This full version includes proofs of all results.}}
\date{}

\author{\IEEEauthorblockN{Joshua Brakensiek\IEEEauthorrefmark{1}, Yeyuan Chen\IEEEauthorrefmark{2}, Manik Dhar\IEEEauthorrefmark{3}, and Zihan Zhang\IEEEauthorrefmark{4}}\\
\IEEEauthorblockA{\IEEEauthorrefmark{1}Department of Electrical Engineering and Computer Sciences,\\University of California, Berkeley.\\Email: \texttt{josh.brakensiek@berkeley.edu}.}\\
\IEEEauthorblockA{\IEEEauthorrefmark{2}Department of EECS, University of Michigan, Ann Arbor.\\Email: \texttt{yeyuanch@umich.edu}}\\
\IEEEauthorblockA{\IEEEauthorrefmark{3}Department of Mathematics,\\Massachusetts Institute of Technology.\\ Email: \texttt{dmanik@mit.edu}.}\\
\IEEEauthorblockA{\IEEEauthorrefmark{4}Department of Computer Science and Engineering,\\The Ohio State University.\\Email: \texttt{zhang.13691@osu.edu}}}
\date{}

\begin{document}
\maketitle
\thispagestyle{empty}
\begin{abstract}
Motivated by recent developments in coding theory, particular in list-decoding, we introduce a new error model which we call \textbf{semi-adversarial errors}. This error model bridges between fully random errors and fully adversarial errors by allowing some symbols of a message to be corrupted by an adversary while others are replaced with uniformly random symbols. 

\change{We} seek to understand optimal efficient \textbf{unique} decoding algorithms in the semi-adversarial model. For interleaved Reed--Solomon (IRS),  folded Reed--Solomon (FRS) and univariate multiplicity codes, we design  decoding algorithms running in near-linear time for most mixtures of random and adversarial errors. Our analysis matches the information-theoretic optimum for semi-adversarial errors. 

Our algorithm for interleaved Reed--Solomon codes is \change{based on} the decoding algorithm by Bleichenbacher--Kiayias--Yung (BKY) for fully random errors. We use \change{a} monomial-tracking technique to analyze its performance in this new semi-adversarial errors. Inspired by the BKY algorithm, we use \change{interpolations} to extend our approach to the settings of folded Reed--Solomon and multiplicity codes, resulting in fast  algorithms for unique decoding against semi-adversarial errors. Our new decoders for FRS and multiplicity codes replace the sophisticated root-finding step in traditional algorithms, such as the Guruswami--Wang algorithm, with a straightforward polynomial long division. Analysis of these algorithms requires more robust monomial-tracking arguments than IRS codes.
\end{abstract}

\newpage
\thispagestyle{empty}
\setcounter{tocdepth}{3}
\tableofcontents
\newpage
\section{Introduction}
The primary goal of this investigation is to better understand the error-correcting properties of Reed--Solomon (RS) codes \cite{RS-codes}, the most fundamental class of algebraic error-correcting codes. To define a Reed--Solomon code, we consider a finite field $\F_q$ as well as distinct evaluation points $\alpha_1, \hdots, \alpha_n \in \F_q$. An $[n,k]$ RS code over $\F_q$ is then
\[
\mathsf{RS}_{n,k,q}(\alpha_1, \dots, \alpha_n):=\left\{ \big(f(\alpha_1), \dots, f(\alpha_n) \big)\,:\, f(X) \in \F_q[X], \ \deg(f) < k \right\}\subseteq \F_q^n.
\]

\change{An important question} one can ask about Reed--Solomon codes is their decodability. In particular, consider a protocol in which a sender picks $\vec{c} \in \RS_{n,k,q}(\alpha_1,\hdots, \alpha_n)$ and transmits $\vec{c}$ across a channel, which is received as a corrupted message $\vec{y} \in \F^n_q$. We seek to understand under what circumstances $\vec{c}$ can be recovered from $\vec{y}$. For instance, the Berlekamp--Welch decoder~\cite{welch1986error} can efficiently recover $\vec{c}$ if and only if the Hamming distance (i.e., the number of distinct symbols) between $\vec{c}$ and $\vec{y}$ is at most $(n-k)/2$. This algorithm is information-theoretically optimal because $\RS_{n,k,q}$ has distance $n-k+1$, i.e., it attains the \emph{Singleton bound}~\cite{singleton1964maximum} and thus is an \emph{MDS code}.

However, there are many other questions one can ask about the structure of Reed--Solomon codes. For example, what if the corruptions are sampled uniformly at random? Or what if the number of corruptions exceeds the unique-decoding radius of $(n-k)/2$, and we can only decode to a list? For decoding problems like these, a central open problem is to find \emph{efficient} decoding algorithms which match the information-theoretic optimum for these problems. For instance, the Guruswami-Sudan \cite{guruswami1998improved} can efficiently list-decode any Reed--Solomon code up to the \emph{Johnson radius} of $1 - \sqrt{R}$, where $R := k/n$ is the rate of the Reed--Solomon code. However, a recent line of work \change{has established} that most Reed--Solomon codes can be list-decoded up to the generalized Singleton bound, which asymptotically tends to the channel capacity of $1-R$ \cite{ST23,brakensiek2023generic,guo2023randomly,alrabiah2024randomly}.

Although closing the gap between such combinatorial and algorithmic results currently seems out of reach,\footnote{\change{After completion of this work, \cite{BCPZZ26} showed that many Reed--Solomon codes can be efficiently list-decoded up to capacity, including Reed--Solomon codes over sufficiently large prime fields. See the survey \cite{HKS26} for further information and an overview of subsequent developments. See also \cite{AGG26} for a recent explicit construction of Reed--Solomon codes attaining the generalized Singleton bound.}} we note that an analogous result has been known for decades in the \emph{random} error setting. More precisely, for a special class of Reed--Solomon codes known as \emph{subfield evaluation} RS codes or \emph{interleaved} RS codes, the result of Bleichenbacher, Kiayias, and Yung \cite{bleichenbacher2003decoding} proved that for an $s$-interleaved RS code one can efficiently uniquely decode $\frac{s}{s+1}(n-k)$ random errors. In particular, as the interleaving grows, one can efficiently uniquely decode random errors up to a radius approaching capacity.\footnote{Of note, a contemporaneous algorithm by Coppersmith and Sudan~\cite{coppersmith2003reconstructing} is \emph{believed} to also attain capacity, but its current analysis has a decoding radius of only $1-2R$ in the limit. We discuss this open problem in Section~\ref{sec:concl}.} Given the success of the BKY algorithm in the random error model, is there any hope of generalizing it to improve our understanding of the adversarial model?

The \change{main goal of this paper is to attempt to translate} results \change{in} the random error model back into a setting with adversarial errors. As a primary contribution, we introduce a new error model, which we call the \emph{semi-adversarial} error model which hybridizes the fully random and fully adversarial error regimes. Studying such a model allows us to integrate approaches used in the decoding literature for fully adversarial and fully random errors, with the benefit that we are able to understand \emph{tight} tradeoffs between information-theoretic bounds and fast (i.e., near-linear time) decoding algorithms.

\subsection{The Semi-Adversarial Error Model}

We now proceed to define the error models (i.e., channels) we examine in this paper. Formally, given a code $C \subseteq \Sigma^n$, we define a \emph{channel} to be a function $\Psi$ which takes as input a codeword $\vec{c} \in C$ and returns as output a probability distribution $\Psi(\vec{c})$ over messages $\vec{y} \in \Sigma^n$. We let $\change{\Phi}(\Sigma^n)$ denote the set of probability distributions over $\Sigma^n$, so a channel can be viewed as a function $\Psi : C \to \change{\Phi}(\Sigma^n)$. We then define an \emph{error model} $\mathcal E$ to be a set of channels. We say that an algorithm $\mathcal A$ can decode $\mathcal E$ with probability $p$, if for every $\vec{c} \in C,\Psi\in\mathcal{E}$, the algorithm $\mathcal A$ when given input a sample $\vec{y} \sim \Psi(\vec{c})$ outputs $\vec{c}$ with probability at least $p$.  

For each $i \in [n]$, we let $y_i$ denote the $i$-th symbol of $\vec{y}$. Given $I \subseteq [n]$, we let $\Psi(\vec{c})|_{I}$ denote the marginal distribution where we restrict to the coordinates of the output string indexed by $I$.  We start by defining the well-studied adversarial and random error models.

\begin{definition}[Adversarial Errors]
We say that a channel $\Psi : C \to \change{\Phi}(\Sigma^n)$ has at most $e$ adversarial errors if for every transmitted string $\vec{c} \in C$, every $\vec{y} \in \supp \Psi(\vec{c})$ (i.e., the strings with nonzero probability) have Hamming distance at most $e$ from $\vec{c}$.
\end{definition}

\begin{definition}[Random Errors]
We say that a channel $\Psi : C \to \change{\Phi}(\Sigma^n)$ has $e$ random errors if for every $\vec{c}$, there exists $I \subseteq [n]$ with $|I| \ge n-e$ such that $\Psi(\vec{c})$ has the following properties.
\begin{itemize}
\item[(1)] For every $\vec{y} \in \supp \Psi(\vec{c})$, $c_i = y_i$ for all $i \in I$.
\item[(2)] $\Psi(\vec{c})|_{[n] \setminus I}$ is the uniform distribution on $\Sigma^{[n] \setminus I}$.
\end{itemize}
\end{definition}

\begin{remark}
The ability for an algorithm to decode a family of channels with success probability $p$ is equivalent to the algorithm being able to decoding with probability $p$ any convex combination of these channels (i.e., sample a random channel then sample the output of that channel). 
\end{remark}

\begin{remark}
In our random error model, we assume that the set $I$ of non-errors can be adversarially chosen. Note this coincides with the random error model of Coppersmith--Sudan \cite{coppersmith2003reconstructing}, but deviates from Bleichenbacher et al.~\cite{bleichenbacher2003decoding} where the set $I$ is chosen uniformly at random. However, as we shall see, the decoding algorithm considered in \cite{bleichenbacher2003decoding} applies equally well in this stronger error model.
\end{remark}

\begin{remark}
To simplify the notation later in the paper, we often keep the choice of $\Psi$ implicit and exclusively refer to the transmitted codeword $\vec{c}$ and the sampled codeword $\vec{y} \sim \Psi(\vec{c})$. For example, we will say that $\vec{y}$ is received with at most $e$ random errors if the implicit channel $\Psi$ comes from the at most $e$ random errors model.
\end{remark}

Note the random error channel is weaker than the corresponding adversarial error channel in the sense that for the same number of errors, the random error channel can be viewed as a special case of the adversarial error channel where the adversary chooses to corrupt each symbol uniformly at random.
However, we have a tighter understanding of the random error model---the unique decoding radius in the random error model (approximately) coincidences with the capacity of the channel. This was proven combinatorially by Rudra and Uurtamo \cite{rudra2010two} who showed that any MDS code (i.e., code attaining the Singleton bound) is uniquely decodable up to $(1-R-\varepsilon)n$ random errors.\footnote{This is not restricted to MDS codes. Generally, for a code with relative distance $\delta$, it is possible to uniquely decode it from $(\delta-\eps)n$ random errors. For simplicity we state the result according to MDS codes here.}

For any alphabet $\Sigma$, radius $e\ge 0$ and length-$n$ string $\vec{y}\in\Sigma^n$, we define the Hamming ball $\mathcal{B}(\vec{y},e)$ with center $\vec{y}$ as the set of all length-$n$ strings that have Hamming distance at most $e$ from $\vec{y}$. That is, they differ with $\vec{y}$ on at most $e$ positions.

\begin{proposition}[{\cite[Theorem 1]{rudra2010two}}]
    Consider parameters $0<\eps,R<1$, $n>\Omega(1/\eps)$, any alphabet $\Sigma$ with size at least $2^{\Omega(1/\eps)}$, and any MDS code $\mathcal{C}\subseteq \Sigma^n$ with rate $R$. For any codeword $\vec{c}\in\mathcal{C}$, let $\vec{y}$ denote the received word after applying $e=(1-R-\eps)n$ random errors to $\vec{c}$. Then, with probability at least $1-|\Sigma|^{-\Omega(\eps n)}$, we have $\mathcal{B}(\vec{y},e)\cap\mathcal{C}=\{\vec{c}\}$.
\end{proposition}

However, if we inspect the proof of Rudra and Uurtamo 
\cite{rudra2010two}, we can observe that, in fact, unique-decoding is still possible even when the errors are a mixture of adversarial errors and random errors, which implies similar unique-decoding bounds (and as we shall soon discuss, corresponding algorithms) could be further generalized to stronger error models. This motivates us to consider a partial-adversarial/partial-random error model, which we call the semi-adversarial error model.

\paragraph{Semi-Adversarial Errors}
We now present the primary error model we study in this paper: the semi-adversarial error model. It particular, it can be viewed as a \emph{composition} of the channels defined in the adversarial and random error models. \change{See \cref{fig:semi-adv} for an illustration.}

\begin{figure}
    \centering
\begin{tikzpicture}[scale=0.7]
\fill[red!50!white] (6,-0.5)--(6,-1.5)--(9,-1.5)--(9,-0.5)--cycle;
\fill[red!50!white] (6,-2)--(6,-3)--(9,-3)--(9,-2)--cycle;
\fill[green!50!white] (3,-2)--(3,-3)--(6,-3)--(6,-2)--cycle;

\draw (0,0)--(9,0)--(9,1)--(0,1)--cycle;
\draw (1,0)--(1,1);
\draw (2,0)--(2,1);
\draw (3,0)--(3,1);
\draw (4,0)--(4,1);
\draw (5,0)--(5,1);
\draw (6,0)--(6,1);
\draw (7,0)--(7,1);
\draw (8,0)--(8,1);

\node at (-2, 0.5) {$\vec{c}$};

\node at (0.5,0.5) {$0$};
\node at (1.5,0.5) {$1$};
\node at (2.5,0.5) {$2$};
\node at (3.5,0.5) {$0$};
\node at (4.5,0.5) {$1$};
\node at (5.5,0.5) {$2$};
\node at (6.5,0.5) {$0$};
\node at (7.5,0.5) {$1$};
\node at (8.5,0.5) {$2$};

\draw (0,-0.5)--(9,-0.5)--(9,-1.5)--(0,-1.5)--cycle;
\draw (1,-0.5)--(1,-1.5);
\draw (2,-0.5)--(2,-1.5);
\draw (3,-0.5)--(3,-1.5);
\draw (4,-0.5)--(4,-1.5);
\draw (5,-0.5)--(5,-1.5);
\draw (6,-0.5)--(6,-1.5);
\draw (7,-0.5)--(7,-1.5);
\draw (8,-0.5)--(8,-1.5);

\node at (-2, -1) {$\vec{z}$};

\node at (6.5,-1) {$1$};
\node at (7.5,-1) {$2$};
\node at (8.5,-1) {$0$};

\draw (0,-2)--(9,-2)--(9,-3)--(0,-3)--cycle;
\draw (1,-2)--(1,-3);
\draw (2,-2)--(2,-3);
\draw (3,-2)--(3,-3);
\draw (4,-2)--(4,-3);
\draw (5,-2)--(5,-3);
\draw (6,-2)--(6,-3);
\draw (7,-2)--(7,-3);
\draw (8,-2)--(8,-3);

\node at (-2, -2.5) {$\vec{y} \in \supp \Psi(\vec{c})$};

\node at (0.5,-2.5) {$0$};
\node at (1.5,-2.5) {$1$};
\node at (2.5,-2.5) {$2$};
\node at (3.5,-2.5) {$*$};
\node at (4.5,-2.5) {$*$};
\node at (5.5,-2.5) {$*$};
\node at (6.5,-2.5) {$1$};
\node at (7.5,-2.5) {$2$};
\node at (8.5,-2.5) {$0$};

\end{tikzpicture}
\caption{\change{An illustration of $(e_0,e) = (3,6)$ semi-adversarial errors being applied to a word $\vec{c} \in \{0,1,2\}^9$. In this illustration $J = \{1,2,3\}$ and $I = \{1,2,3,4,5,6\}$. The red grid squares correspond to adversarial corruptions, while the green grid squares correspond to uniformly random corruptions. The symbol $*$ is used to depict a uniformly random symbol.}}
\label{fig:semi-adv}
\end{figure}

\begin{definition}[Semi-Adversarial Errors]\label{def:semi-adv}
We say that a channel $\Psi : C \to \change{\Phi}(\Sigma^n)$ has $(e_0, e)$-semi-adversarial errors (where $0 \le e_0 \le e$) if for every $\vec{c} \in \Sigma^n$ there exist a set of non-adversarial indices  $I \subseteq [n]$ with $|I| = n-e_0$, a set of agreement indices $J \subseteq I$ with $|J| \ge n-e$, and an adversarial corruption string $\vec{z}\in \Sigma^{n-|I|}$ with the following properties. 
\begin{itemize}
\item[(Agreement indices uncorrupted):] For every $\vec{y} \in \supp \Psi(\vec{c})$, $\vec{c}|_{J} = \vec{y}|_{J}$.
\item[(Adversarial corruptions applied)] For every $\vec{y} \in \supp \Psi(\vec{c})$, $\vec{y}|_{[n]\setminus I} = \vec{z}$. 
\item[(Random errors after adversarial corruptions):]The distribution $\Psi(\vec{c})|_{I \setminus J}$ is uniform over $\Sigma^{|I \setminus J|}$.
\end{itemize}
\end{definition}

\begin{remark}
In our analysis of decoding algorithms for semi-adversarial errors, we make use of some more fine-grained notation. In particular, for the given $I\subseteq[n]$ with $|I|=n-e_0$ we assume that the ``true'' adversarial errors occur in some subset $K\subseteq [n]\setminus I$ with $e_A := |K|\leq e_0$. Formally, we define $K$ such that for any position $j\in[n]\setminus I$, $j\in K$ iff $y_j \neq c_j$. Note that it may be the case $K \neq [n] \setminus I$, as some of the ``corruptions'' $\vec{z}$ can equal the corresponding coordinates in the original coordinates $\vec{c}$. We also often use the notation $K' := I \setminus J$ with $e_R := |K'|$ to denote the coordinates for which random errors are applied. We then let $\be :=e_A+e_R\leq e$ denote the number of positions which are truly corrupted during transmission.
\end{remark}

\begin{remark}
    We note that $(e,e)$-semi-adversarial errors are equivalent to exactly $e$ adversarial errors, while $(0,e)$-semi-adversarial errors are equivalent to $e$ random errors.
\end{remark}

\begin{remark}[Comparison to recent and concurrent work.] A recent paper \cite{guerrini2023simultaneous} has studied hybrid error models in connection to the closely related problem of rational function decoding. Some follow-up papers \cite{abbondati2024decoding,abbondati2025simultaneous} deal with the same problem over number fields as opposed to function fields. 
The stated results in \cite{guerrini2023simultaneous} are for semi-adversarial errors where the adversarial errors are for \emph{valuation errors} which do not exist in the IRS case so the statement does not directly imply results in our setting. The results for number fields (like $\mathbb{Q}$) in \cite{abbondati2025simultaneous} do work for any kind of adversarial errors, but this work was done concurrently and independently of us. A later work \cite{abbondati2025communication} generalizes techniques of \cite{guerrini2023simultaneous,abbondati2025simultaneous} and gives results for IRS codes in the semi-adversarial model. Although \cite{abbondati2025communication} was posted later than our work, it is developed independently. We compare their results with ours.
\begin{compactitem}
\item For decoding $(e_0,e)$ semi-adversarial errors for IRS codes, \cite{abbondati2025communication} requires $e_0\leq n-k-e-(e-e_0)/s$, while our bound $e_0\leq n-k-e$ is better and actually optimal. This difference could be ignored if one allows $s$ being arbitrarily large, with the cost of larger alphabet.
\item Unlike our work, \cite{abbondati2025communication} does not include results for folded RS codes and multiplicity codes, but they can handle a family of interleaved rational function codes that we do not consider.
\item The success probability of decoding algorithms in \cite{abbondati2025communication} is larger than ours, but both two papers achieve $1-o_n(1)$ success probability when the field size $q=\omega(n)$. 
\end{compactitem}
\end{remark}

For any number of total errors $e$, when the number $e_0$ of adversarial errors gets larger, the task of unique-decoding (with high success probability) gets more difficult. As our first nontrivial observation about the semi-adversarial error model, we show that in any MDS code the number $e_0$ of adversarial errors must be bounded by
\[
e_0\leq\min(e,(1-R)n-e).
\]
In particular, to uniquely decode $e$ errors in the semi-adversarial model, the maximum number of adversarial errors that could occur $e_0=\min(e,n-k-e)$. Our formal result is as follows.

\begin{proposition}\label{prop:optimal}
  For any \change{code} $\mathcal{C}\subseteq \Sigma^n$ over alphabet $\Sigma$ with rate $R$, block length $n$, $e_0>(1-R)n-e$, there exists some word $\vec{z}\in\Sigma^n,|\mathcal{B}(\vec{z},e_0)\cap \mathcal{C}|\ge 1$ and $K\subseteq [n], |K|=e-e_0$ such that any received word $\vec{y}$ that differs with $\vec{z}$ only on positions in $K$ satisfies $|\mathcal{B}(\vec{y},e)\cap \mathcal{C}|\ge 2$. 
\end{proposition}
We should regard $\vec{y}$ as the real received word by adding $e-e_0$ random errors upon the adversarial received word $\vec{z}$. From \cref{prop:optimal} we can see that if $e_0>(1-R)n-e$, then no matter the content of random errors the unique-decoding is impossible. \change{We prove \cref{prop:optimal} as} a special case of \cref{prop:list_optimal} \change{in \cref{sec:combo}}.

On the other hand, by slightly adapting the proof of Rudra and Uurtamo \cite{rudra2010two}, we can actually get the following positive combinatorial bound that states with high probability it is possible to uniquely decode from any $e_0\leq \min\left(e,(1-R-\varepsilon)n-e\right)$ adversarial errors, which approaches the optimal bound $(1-R)n$ for diminishing constant $\varepsilon$. Formally, we have the following combinatorial unique-decodability.

\begin{proposition}\label{prop:RU_semi}
    Consider parameters $0<\eps,R<1$, $n>\Omega(1/\eps)$, any alphabet $\Sigma$ with size at least $2^{\Omega(1/\eps)}$, and any MDS code $\mathcal{C}\subseteq \Sigma^n$ with rate $R$.  For any codeword $\vec{c}\in\mathcal{C}$, let $\vec{y}$ denote the received word after applying $(e_0,e)$-semi-adversarial  errors to $\vec{c}$ where $e_0\leq \min(e,(1-R-\eps)n-e),e\leq (1-R)n$. Then, with probability at least $1-|\Sigma|^{-\Omega(\eps n)}$, we have $\mathcal{B}(\vec{y},e)\cap\mathcal{C}=\{\vec{c}\}$.
\end{proposition}

We prove \cref{prop:RU_semi} in \change{\cref{sec:combo}.} Motivated by these observations, a very natural question to ask is how to design unique-decoding algorithms in semi-adversarial error model that match the combinatorial bound $(1-R-\varepsilon)n-e$ from \cref{prop:RU_semi} or even achieve the optimal bound $(1-R)n-e$. In this paper, we provide unique-decoding algorithms from semi-adversarial errors for (a subclass of) Reed--Solomon and related codes that match the counting-based or even the optimal bound. In particular, our algorithm for interleaved Reed--Solomon codes builds upon the approach of Bleichenbacher, Kiayias, and Yung \cite{bleichenbacher2003decoding} and Alekhnovich \cite{alekhnovich2005linear} and our algorithms for folded Reed--Solomon and multiplicity codes are entirely new.

\subsection{Our Results}

We now proceed to discuss our results on the unique decoding of variants of Reed--Solomon codes in the semi-adversarial error model. All of our results are algorithmically efficient, running in $\tilde{O}(n)$ time. Moreover, for most tradeoff regimes between adversarial and random errors, our results are either information-theoretically optimal or near-optimal.
In particular, our results pertain to three variant of Reed--Solomon codes: \emph{interleaved} RS codes (equivalently subfield evalaution RS codes), \emph{folded} RS codes, and \emph{multiplicity} codes. For each result, we briefly define the variant of RS codes we are considering along with the guarantees of our algorithm.

\paragraph{Interleaved Reed--Solomon Codes}
Interleaving is a common technique in coding theory that was initially used to design codes and decoders with enhanced burst error-correction capabilities. By dispersing these errors, interleaving transforms them into random errors that are easier to correct. An important subfamily of Reed--Solomon codes is that of interleaved Reed--Solomon (IRS) codes (e.g., \cite{krachkovsky1997decoding, coppersmith2003reconstructing,bleichenbacher2003decoding,schmidt2009collaborative}).

Given polynomials $f_1, \hdots, f_s \in \F_q[x]$, we define the $s$-interleaved encoding map $\mathcal{I}(f_1,f_2,\hdots,f_s)$ to be $(\cC^f_1,\cC^f_2,\hdots,\cC^f_n)\in\left(\mathbb{F}_q^s\right)^n,$ where $\cC^f_i:=\Bigl(f_1(\alpha_i),f_2(\alpha_i),\dots,f_s(\alpha_i)\Bigl)$. With this encoder, we define a $[n,k]$ $s$-interleaved RS code for any distinct $\alpha_1, \hdots \alpha_n \in \F_q$ as follows.
\[\mathsf{IRS}_{n,k,q,s}(\alpha_1,\hdots, \alpha_n):=\bigg\{\mathcal{I}(f_1,f_2,\hdots,f_s):f_i\in\mathbb{F}_q[x],\text{ }\deg f_i<k,\text{ and }i\in[s]\bigg\}\subseteq\left(\mathbb{F}_q^s\right)^n,\]
where $[s] := \{1,2,\hdots, s\}$. Although an IRS code is technically not an RS code, we show in Appendix~\ref{app:IRS-RS} that IRS codes have decoding properties identical to those of a class of \emph{subfield} Reed--Solomon codes. That is, if we let $\F_{q^s}$ be the $s$-degree extension of $\F_q$, then $\RS_{n,k,q^s}(\alpha_1,\dots,\alpha_n) \subseteq \F_{q^s}^n$ has identical decoding properties in the semi-adversarial model.

We now state our main decoding theorem for IRS codes in the semi-adversarial model. See Figure~\ref{fig:mainresult} for a visual summary of our result.
\begin{theorem}[Informal Version of \cref{thm:algoirs_analysis}]\label{thm:main_RS}
   Consider integers $s,n,k\ge1$.
   For any interleaved Reed--Solomon code $\mathsf{RS}_{n,k,q^s}(\alpha_1, \dots, \alpha_n)$ with \change{distinct} $\alpha_1, \dots, \alpha_n\in\F_q\change{\setminus \{0\}},$ there exists an efficient unique decoding algorithm 
   that, with \change{probability at least $1-{e}/{q}$}, can correct $(e_0,e)${-semi-adversarial} errors for any 
   \[e \leq \frac{s}{s+1}(n-k)\quad \text{and}\quad e_0\leq\min(e,n-k-e).\] The running time of the algorithm is at most $O\left(s^{O(1)}n\log^2(n)\log\log (n)\right)$. We refer the reader to our Algorithm~\ref{algo:BKY+ALEC} for more details.
\end{theorem}

In particular, note that we smoothly interpolate between the algorithms of Berlekamp--Welch for adversarial errors and the algorithm of BKY for random errors. Furthermore, note that as long as $e \leq \frac{s}{s+1}(n-k)$, our decoding bound \change{matches} the combinatorial bound of \Cref{prop:optimal} and  \Cref{prop:RU_semi}. Furthermore, for $e > n - \sqrt{nk}$ we ``break'' the Johnson bound\footnote{In this paper ``Johnson bound'' refers to the error/rate tradeoff of $e \le n - \sqrt{nk}$. We do \emph{not} improve on Guruswami-Sudan algorithm in the adversarial model.} in the sense that our result is not implied by a combination of Guruswami-Sudan with a combinatorial bound on the list size.

\begin{figure}
    \centering
\definecolor{cyan}{HTML}{00FFCC}
    
\begin{tikzpicture}
\usetikzlibrary{patterns,patterns.meta}
\usetikzlibrary{patterns,patterns.meta}

    \coordinate (A) at (0,0);
    \coordinate (B) at (3,3);
       \coordinate (F) at (3,0);
    \coordinate (C) at (4.5,4.5);
      \coordinate (G) at (4.5,0);
      \coordinate (H) at (4.5,1.5);
    \coordinate (D) at (6,0);
        \coordinate (E) at (5,5);
           \coordinate (I) at (5.5,0);
           \coordinate (L) at (5.5,0.5);
            \coordinate (X) at (0,1.5);
\coordinate (Y) at (0,0.5);
\coordinate (Z) at (0,4.5);
    \fill[cyan,opacity=0.5] (A) -- (B) -- (F) -- cycle;
    \fill[blue!100,opacity=0.5] (B) -- (F) -- (G) -- (C) -- cycle;
    \fill[violet,opacity=0.5] (G) -- (H) -- (L) -- (I) -- cycle;
    \draw[thick] (B) --++ (0,-3);
    \draw[line width =2 pt] (C) --++ (0,-4.5); %
            \draw[dashed] (X) -- (H);
    \draw[line width =3.5 pt, color=orange] (A) -- (B);
    \draw[thick] (A) -- (B);
     \draw[thick] (I) -- (L);
         \draw[line width =3.5 pt, color=orange] (B) -- (D);
    \draw[thick] (B) -- (E);
      \draw[dashed] (Z) -- (C);
          \draw[dashed] (Y) -- (L);
       \draw[thick] (A) -- (D);
     \draw[line width =3.5 pt, color=green!60!black] (A) -- (I);

    \draw[thick,->] (0,0) -- (10,0) node[anchor=north] {\# of total errors $(e)$};
     \draw[thick,->] (1.5,0.5) -- (1,-1) node[anchor=north] {\textcolor{cyan!80!black}{Berlekamp--Welch (BW) Algorithm}};
     \draw[thick,->] (4.5,3.2) -- (6.3,3.5) node[anchor=south] {{Johnson Bound}};
     \draw[thick,->] (5,0) -- (9.5,1) node[anchor=south] {\textcolor{green!60!black}{Bleichenbacher–Kiayias–Yung (BKY) Algorithm}};
     \draw[thick,->] (4,3.5) -- (6,5) node[anchor=south] {\textcolor{blue!100}{Guruswami--Sudan (GS) Algorithm}};
     \draw[thick,->] (4.6,1) -- (6,2.5) node[anchor=south]{\textcolor{violet}{Our Result}};
       \draw[thick,->] (5.5,0) -- (5.5,-0.9) node[anchor=north] {$\frac{s}{s+1}(n-k)$};
    \draw[thick,->] (0,0) -- (0,5.2) node[anchor=south] {\# of adversarial errors $(e_0)$};

    \draw[dashed] (B) --++ (-3,0) node[left] {\(\frac{n-k}{2}\)};

    \node[below] at (3,0) {\(\frac{n-k}{2}\)};
    \node[below] at (4.5,0) {\(n - \sqrt{nk}\)};

    \node[below] at (6.1,0) {\(n-k\)};
    \node[left] at (0,0.5) {\(\frac{1}{s+1}(n-k)\)};
    \node[left] at (0,4.5) {\(n-\sqrt{nk}\)};
     \node[left] at (0,1.5) {\({\sqrt{nk}-k}\)};
      \node[left] at (0,0) {\(0\)};
\end{tikzpicture}
 \caption{This figure is a comparison between our result with BW Algorithm, GS Algorithm, and the original BKY Algorithm for {$[n,k]$ {\bf Reed--Solomon codes} over alphabet $\F_{q^s}$ with $n$ evaluation points chosen from $\F_q$}. The horizontal axis depicts the total number of errors $e$ allowed in the semi-adversarial model, and the vertical axis depicts the number of adversarial errors $e_0$. The region below the two orange lines is where unique-decoding is combinatorially possible by \Cref{prop:RU_semi}. It is impossible to achieve unique decoding above the two orange lines by \cref{prop:optimal}. \change{ChatGPT was used to help make this figure.}
 }
    \label{fig:mainresult}
\end{figure}

\paragraph{Folded Reed--Solomon (RS) Codes}

Another common variant of RS codes which we study is that of \emph{folded} Reed--Solomon (FRS) codes~\cite{Guru06}. In contrast to IRS codes which bundle the evaluations of the same point $\alpha \in \F_q$ at $s$ different polynomials, FRS codes bundle $s$ correlated evaluations of the same polynomial. More precisely, let $\gamma$ be a multiplicative generator of $\F_q$ and consider $s \ge 1$. We say that $\alpha_1, \hdots, \alpha_n \in \F_{q} \change{\setminus \{0\}}$ are \emph{appropriate} if the $sn$ values $\{\gamma^{j-1}\alpha_i : i \in [n], j \in [s]\}$ are distinct \cite{cz24}. Given $f \in \F_{q}[x]$, we define our evaluation map to be $\mathcal{C}(f):=(\cC^f_1,\cC^f_2,\dots,\cC^f_n)\in\left(\mathbb{F}_q^s\right)^n$, with $\cC^f_i:=\Bigl(f(\alpha_i),f(\gamma\alpha_i),\dots,f(\gamma^{s-1}\alpha_i)\Bigl).$
Then, for parameters $n,  k$ and appropriate evaluation points $\alpha_1, \hdots, \alpha_n \in \F_q \change{\setminus \{0\}}$, we define our $(s,\gamma)$-folded RS code to be
\[\mathsf{FRS}^{s,\gamma}_{n,k,q}(\alpha_1,\alpha_2,\dots,\alpha_n):=\bigg\{\mathcal{C}(f):f\in\mathbb{F}_q[x],\text{ }\deg f<k\bigg\}\subseteq\left(\mathbb{F}_q^s\right)^n.
\]

We now state our results which apply to folded Reed–Solomon codes.
\begin{theorem}[Informal Version, See \cref{thm:algofrs_analysis}]\label{thm:FRS}
   Consider integers $s,L,n,k\ge1$ and a generator $\gamma$ of $\F_q^{\times}$ where $s\ge L$. 
   For any appropriate $(s,\gamma)$-folded Reed--Solomon code $\mathsf{FRS}_{n,k}^{s,\gamma}(\alpha_1, \dots, \alpha_n)$ over the alphabet $\F_{q}^s$, there exists an efficient unique decoding algorithm that, with \change{probability at least $1-e(s-L+1)/q$}, can correct {$(e_0,e)$-semi-adversarial} errors for any 
   \[e \leq \frac{L}{L+1}\left(n-\frac{k}{s-L+1}\right)\quad and\quad e_0\leq\min\left(e,n-e-\frac{k}{s-L+1}\right).\] The running time of the algorithm is at most $O\left(s^{O(1)}n\log^2(n)\log\log (n)\right)$. We refer the reader to our Algorithm~\ref{algo:main} for more details.
\end{theorem}

Compared to the well-known Guruswami--Wang algorithm \cite{guruswami2011linear,guruswami2013linear,kopparty2018improved,goyal2024fast}, a key advantage of \change{our} algorithm (see Algorithm~\ref{algo:main}) is that it replaces the less-efficient root-finding step with a straightforward polynomial long division, resulting in \change{a modest} efficiency boost. However, this is in tradeoff to the fact that the Guruswami--Wang algorithm can handle fully adversarial errors, whereas our approach is limited to semi-adversarial errors.

\paragraph{Multiplicity Codes} Our final main result is for another popular variant of RS codes known as \emph{(univariate) multiplicity codes} (e.g., \cite{guruswami2013linear,kopparty2014high,kopparty2015list}). 
Again, multiplicity codes are similar to IRS codes except this time instead of evaluating a point $\alpha$ at $s$ different polynomials, multiplicity codes evaluate $s$ Hasse derivatives of the same polynomial at $\alpha$. In particular, given $f \in \F_q[x]$ and symbolic variable $z$, we define $f^{(i)}(x)$ (the $i$-th Hasse derivative) to be the coefficient of $z^i$ in the expansion of $f(x+z)$. From this, we can define our evaluation function to be $\mathcal{M}(f):=(\cC^f_1,\cC^f_2,\dots,\cC^f_n)\in\left(\mathbb{F}_q^s\right)^n$, with $\cC^f_i:=\Bigl(f(\alpha_i),f^{(1)}(\alpha_i),\dots,f^{(s-1)}(\alpha_i)\Bigl)$. Then, for parameters $n, k, s$ and distinct evaluation points $\alpha_1, \hdots, \alpha_n \in \F_q$, our multiplicity code is defined to be 
\[\mathsf{MULT}_{n,k}^{s}(\alpha_1,\alpha_2,\dots,\alpha_n):=\bigg\{\mathcal{M}(f):f\in\mathbb{F}_q[x],\text{ }\deg f<k\bigg\}\subseteq\left(\mathbb{F}_q^s\right)^n.\]

\begin{remark}
RS, IRS, FRS, and MULT codes all coincide when $s=1$.
\end{remark}

\begin{theorem}[Informal Version, See \cref{thm-multAlgAna}]\label{thm:MULT}
   Consider integers $s,L,n,k\ge1$ and $n$ distinct evaluation points $\alpha_1,\dots,\alpha_n\in\F_q\change{\setminus \{0\}}$.
   If $\operatorname{Char}(\F_q)>s$, then for any multiplicity code $\mathsf{MULT}_{n,k}^s(\alpha_1,\dots,\alpha_n)$  over the alphabet $\F_{q}^s$, there exists an efficient unique decoding algorithm that, with \change{probability at least $1-e(s-L+1)/q$}, can correct {$(e_0,e)$-semi-adversarial} errors for any 
   \[e \leq \frac{L}{L+1}\left(n-\frac{k}{s-L+1}-1\right)\quad and\quad e_0\leq\min\left(e,n-e-\frac{k}{s-L+1}\right).\] The running time of the algorithm is at most $O\left(s^{O(1)}n\log^2(n)\log\log (n)\right)$. We refer the reader to our Algorithm \ref{algo:Multmain} for more details.
\end{theorem}

As in Theorem~\ref{thm:FRS}, Theorem~\ref{thm:MULT} has the advantage of the work of Guruswami--Wang algorithm \cite{guruswami2011linear,guruswami2013linear,kopparty2018improved,goyal2024fast} of replacing the less-efficient root-finding step with a straightforward polynomial long division, resulting in \change{an} efficiency boost. 

\change{
\paragraph{Comparison to prior work} To compare \cref{thm:FRS} and \cref{thm:MULT} to previous works, we consider the regime in which the total number of errors $e$ equals $n(1 - R - \eps)$ where $R = \frac{k}{sn}$ is the rate of the code and $\eps > 0$ measures proximity to capacity. In \cref{thm:FRS} and \cref{thm:MULT}, this corresponds to setting $L = \Theta(1/\eps)$ and $s = \Theta(1/\eps^2)$, where our algorithms gives a runtime guarantee of $n \log^{2+o(1)} n \cdot \poly(1/\eps)$.

For both FRS and MULT codes, the work of Goyal, Harsha, Kumar, and Shankar~\cite{goyal2024fast} gives a randomized algorithm for decoding $e$ adversarial errors, in $n\cdot \log(n)^{O(1)} \cdot 2^{O(1/\eps)}$ time.  Thus, \cref{thm:FRS} and \cref{thm:MULT} prove a tradeoff that if the error model is changed to semi-adversarial errors, then the runtime dependence can be made polynomial in $1/\eps$, albeit in many applications $\eps > 0$ is considered to be constant.

More recently,\footnote{\change{Subsequent to the initial version of our paper.}} Ashvinkumar, Habib, and Srivastava~\cite{ashvinkumar2025algorithmic} presented new randomized algorithm for decoding folded Reed-Solomon codes in the adversarial error model, with a polynomial dependence on $1/\eps$.
}

\subsection{Literature Overview}
In this subsection, we provide an overview of previous work on both unique and list decoding of variants of Reed–Solomon (RS) codes, including RS codes themselves. In particular, we discuss the state-of-the-art decoding algorithms in both the adversarial and random error models.

\paragraph{Decoding Reed--Solomon Codes}
Since the 1960s, efficiently decoding Reed--Solomon (RS) codes has become a central subject of study in coding theory, due to its importance in both practical error-correction applications and fundamental theoretical research in algorithms, complexity, and cryptography (e.g., \cite{ben2018fast,ben2020deep,ben2023proximity,arnon2024stir}). For an $[n,k]$ Reed--Solomon code, it can uniquely correct up to $(n-k)/2$ adversarial errors. During the past half century, numerous algorithms have been developed to uniquely decode Reed--Solomon codes up to this bound in the adversarial error model (e.g., \cite{peterson1960encoding,gorenstein1961class,welch1986error,berlekamp1996bounded,massey1969shift,gao2003new}). Among these many algorithms, two of the most popular are the Berlekamp--Welch algorithm \cite{welch1986error} and the Berlekamp--Massey algorithm \cite{massey1969shift}. 

The Berlekamp--Welch algorithm, with a running time complexity of $O(n^3)$, employs a clever algebraic technique to decode Reed--Solomon codes, using polynomial interpolation to identify and correct errors. More concretely, \cite{welch1986error} approaches the decoding problem by transforming it into a task of finding two polynomials: an error locator polynomial $E(x)$ and a numerator polynomial $Q(x)$. For a received word represented as points $(\alpha_i, y_i)$, where $\alpha_i$ are distinct evaluation points and $y_i$ are the received (corrupted) values, the algorithm sets up a system of linear equations based on the condition $Q(\alpha_i) = y_i E(\alpha_i)$ for all $i$. Here, $E(x)$ has degree equal to the number of errors $e$, and $Q(x)$ has degree less than $k + e$, where $k$ is the dimension of the code. When the number of errors $e$ is at most $(n-k)/2$, this system has a unique solution up to scaling. Solving it allows the decoder to determine $E(x)$, whose roots indicate the error locations, and subsequently recover the original message polynomial as $Q(x)/E(x)$. 

Unlike an interpolation-based framework, the Berlekamp--Massey algorithm \cite{massey1969shift}, with a better running time complexity of $O(n^2)$, takes a different approach by directly processing the syndrome sequence to find the error locator polynomial, taking advantage of shift registers for multiplication. Their method is recognized as a standard syndrome-based decoding framework. By building on these two algorithms, more efficient algorithms have been developed with a running time complexity of $O(n\cdot\poly(\log n))$ (e.g., \cite{reed1978fast,lin2016fft,tang2022new}).

Beyond unique decoding, to correct more than $(n-k)/2$ errors in the adversarial error model, list decoding algorithms are involved (e.g., \cite{sudan1997decoding,guruswami1998improved,roth2000efficient,alekhnovich2005linear,wu2008new}). In 1997, inspired by the Berlekamp--Welch algorithm, Sudan \cite{sudan1997decoding} provided the first list decoding algorithm of Reed--Solomon codes capable of correcting up to $n-\sqrt{2nk}$ errors.
Later, by introducing the notion of multiplicity, Guruswami and Sudan \cite{guruswami1998improved} designed an improved list decoding algorithm that can correct up to $n-\sqrt{nk}$ errors, matching the Johnson bound \cite{johnson2003new}. Two follow-up works by Roth--Ruckenstein \cite{roth2000efficient} and Alekhnovich \cite{alekhnovich2005linear} showed some fast implementations of the Guruswami--Sudan algorithm in time $O(n^2)$ and $O(n\cdot \poly(\log n))$ respectively.
Most recently, a long line of works \cite{rudra2014every,ST23,brakensiek2023generic,guo2023randomly,alrabiah2024randomly} have shown that Reed--Solomon codes can achieve list decodability beyond the Johnson bound, reaching up to list decoding capacity. However, it remains an open problem to design a list decoding algorithm for any subclasses of Reed--Solomon codes capable of correcting adversarial errors beyond the Johnson bound. We emphasize that our result on interleaved Reed--Solomon codes (see \cref{thm:main_RS}) addresses this question in the semi-adversarial error model.

\paragraph{Decoding Interleaved Reed--Solomon Codes} Over the past two decades, considerable attention has been devoted to develop (unique) decoder (e.g., \cite{krachkovsky1997decoding,coppersmith2003reconstructing,bleichenbacher2003decoding,schmidt2009collaborative,puchinger2017decoding,yu2018simultaneous}) for interleaved RS codes --- an important subclass of RS codes --- especially under the random error model. In particular, two independent seminal works --- one by Coppersmith and Sudan \cite{coppersmith2003reconstructing} and the other by Bleichenbacher, Kiayias, and Yung \cite{bleichenbacher2003decoding} --- proposed two remarkable (probabilistic) unique decoding algorithms that demonstrate surprisingly large error correction capabilities under the random error model.\footnote{Although BKY's result is only stated in the weaker model in which the errors locations are randomly located, the proof can be adapted to also handle adversarial error locations. This fact was recently observed by \cite{guerrini2023simultaneous,abbondati2024decoding}, where they adapted the BKY algorithm for rational evaluation codes.
}

Coppersmith--Sudan is in some sense the dual of the multivariate interpolation decoding algorithm developed by Parvaresh Vardy~\cite{Parvaresh2004MultivariateID} (although Coppersmith Sudan predates the Parvaresh Vardy algorithm). 

The well-known Guruswami-Sudan algorithm does bivariate interpolation (i.e., a polynomial in $(X,Y)$) on the space of points of the form $\{(\alpha_i, y_i) : i \in [n]\}$, where $\alpha_i$ is the $i$-th evaluation point and $y_i$ is the $i$-th received symbol. They use this interpolation to list decode Reed Solomon codes up to an error threshold of $1-\sqrt{R}$. Multivariate interpolation decoding extends this idea to interleaved Reed-Solomon codes where one interpolates a $s+1$ variate polynomial over $(X,Y_1,\hdots,Y_s)$, where $(Y_1,\hdots,Y_s)$ corresponds to the $s$ components of each received symbol. If the error threshold is $1-R^{s/(s+1)}$ then one can show that the interpolated polynomials will lie in the ideal $\langle Y_1-f_1,\hdots,Y_s-f_s\rangle$. In general, this is not enough for decoding as $Y_i-f_i(x),i\in [n]$ may have many algebraic dependencies, but experimentally that occurs with negligible probability. Parvaresh and Vardy conjecture that for the setting of random errors multivariate interpolation decoding will succeed up to the error rate $1-R^{s/(s+1)}$ with high probability (around $1-n^{-cn}$). The only theoretical result studying this is a follow up work of, Parvaresh, Taghavi, and Vardy~\cite{MIDAnalysis} which shows that this algorithm succeeds for $s=2$ with error rate $1-(6R)^{2/3}-O(R^{5/3})$ and with failure probability at most $n^{-cn}$.

Interpolating over $X,Y_1,\hdots,Y_s$ for $n$ evaluation points $\alpha_1,\hdots,\alpha_n$ without multiplicity can be set up as finding an element in the row-kernel of a matrix where the columns are indexed by $i\in [n]$ and the rows are indexed by monomials $m$ in $X,Y_1,\hdots,Y_s$ with the entries being $m$  evaluated at $X=\alpha_i$ and $(Y_1,\hdots,Y_s)$ equals the $i$-th letter in the received word. Coppersmith-Sudan \cite{coppersmith2003reconstructing} picks an element in the {\em column kernel} of this matrix and uses its support (the coordinates $j\in [n]$ which are non-zero in the chosen vector) to help decode the codeword. Using multiplicities, they get further improvements. In the end, they obtain a (unique) decoder capable of correcting $1-R^{\frac{s}{s+1}}-R$ fraction of random errors with a failure probability at most $O(n^s/q)$. 

In contrast, Bleichenbacher, Kiayias, and Yung \cite{bleichenbacher2003decoding} draw inspiration from the original Berlekamp--Welch unique decoder \cite{welch1986error}, providing a (unique) decoder capable of correcting $\frac{s}{s+1}(1-R)$ fraction of random errors with a failure probability at most $O(n/q).$ The failure probability of \cite{bleichenbacher2003decoding} was further reduced to $O(1/q)$ by  Brown, Minder, and Shokrollahi \cite{brown2004probabilistic}.

Later, inspired by another famous unique decoder of Berlekamp and Massey \cite{berlekamp2015algebraic,massey1969shift}, Schmidt, Sidorenko, and Bossert \cite{schmidt2009collaborative} proposed a faster syndrome-based unique decoder\footnote{The IRS codes considered in \cite{schmidt2009collaborative} require a cyclic restriction on their evaluation points.} capable of correcting $\frac{s}{s+1}(1-R)$ fraction of random errors with a failure probability at most $O(1/q)$. The running time of their algorithm is $O_s(n^2)$ (i.e., $n^2$ times some function of $s$). More recently, using power decoding techniques, Puchinger and Nielsen \cite{puchinger2017decoding} proposed an improved algorithm with a conjectured error correction capability that can correct up to $1-R^{\frac{s}{s+1}}$ fraction of random errors. However, although \cite{puchinger2017decoding} has demonstrated numerous simulations for a wide variety of parameters that support the heuristic correctness of their approach, a formal proof of the theoretical correctness of their proposed algorithm remains open. Therefore, from a theoretical perspective, for $R>1/3$ and $s>2$, the relative decoding radius $\frac{s}{s+1}(1-R)$ for $s$-interleaved RS codes of rate $R$ remains state-of-the-art, as demonstrated in \cite{krachkovsky1997decoding,bleichenbacher2003decoding,schmidt2009collaborative}. 

\paragraph{Decoding Folded Reed--Solomon and Multiplicity Codes} Folded Reed--Solomon (RS) codes and multiplicity codes are two prominent variants of RS codes, with list decoders (e.g., \cite{Guru06,kopparty2015list,guruswami2011linear,guruswami2013linear,kopparty2018improved,goyal2024fast}) that achieve the list decoding capacity even in the presence of adversarial errors. The first folded RS code list decoder was introduced in the seminal work of Guruswami and Rudra \cite{Guru06}, building on the ideas of the Guruswami--Sudan list decoder. The first list decoder for multiplicity codes was introduced independently by Guruswami--Wang \cite{guruswami2013linear} and Kopparty \cite{kopparty2015list}, and it remains the most widely used framework in this domain. 

In fact, Guruswami and Wang \cite{guruswami2011linear,guruswami2013linear} proposed a unified linear algebraic list decoding algorithm for folded RS and multiplicity codes capable of correcting adversarial errors up to capacity. However, for codes of rate $R$ and block length $n$, the list size guarantee from \cite{guruswami2013linear} is $O\bigl(n^{1/\epsilon}\bigr)$ when decoding up to radius $1 - R - \epsilon$. A follow-up work by Kopparty, Ron-Zewi, Saraf, and Wootters~\cite{kopparty2018improved} introduced a randomized algorithm for the root-finding step under the Guruswami--Wang decoding framework, \change{improving} its efficiency and reducing the list size guarantee to $\bigl(1/\epsilon\bigr)^{O(1/\epsilon)}$. Another follow-up work by Goyal, Harsha, Kumar, and Shankar~\cite{goyal2024fast} adapted the approach of Alekhnovich~\cite{alekhnovich2005linear} into the Guruswami--Wang decoding framework, improving the running time to quasi-linear, which is $\widetilde{O}\bigl(2^{\mathrm{poly}(1/\epsilon)} \cdot n\bigr)$.

As a side remark, very recently, two concurrent and independent works by Srivastava~\cite{srivastava2025improved} and Chen--Zhang~\cite{cz24} \change{provided improvements}, reducing the list size from $\bigl(1/\epsilon\bigr)^{O(1/\epsilon)}$ to $O(1/\epsilon^2)$ and $O\bigl(1/\epsilon\bigr)$, respectively.
Moreover, the Chen--Zhang bound almost matches the generalized Singleton bound \cite{ST23} up to a constant
multiplicative factor. It is worth noting that both proofs of Srivastava and Chen--Zhang are purely combinatorial, and turning them into efficient algorithms remains an open problem. For further details, the reader is referred to the most recent comprehensive survey by Garg, Harsha, Kumar, Saptharishi, and Shankar \cite{garg2025exposition}.

Meanwhile, current efficient decoding algorithms for folded Reed--Solomon and multiplicity codes are heavily based on the Guruswami--Wang decoder framework~\cite{guruswami2011linear,guruswami2013linear}. It is an interesting open problem to design alternative, more efficient decoding algorithms for folded RS and multiplicity codes that do not rely on the Guruswami--Wang framework, and a primary contribution of this paper is making progress toward this exact question.

As previously mentioned, we define folded Reed--Solomon codes in the sense of Guruswami--Rudra~\cite{Guru06}. However, the term ``folded Reed–Solomon codes'' was first coined in a paper by Krachkovsky \cite{krachkovsky2003reed}. The codes defined in both papers are constructed by bundling $s$ classical Reed--Solomon codes, but they choose evaluation points in different ways. Concretely, recall that for the folded Reed--Solomon codes defined in \cite{Guru06}, for each $i\in[n]$ and message polynomimal $f$, the $i$-th symbol in the encoding of $f$ is $\left(f(\alpha_i),f(\gamma\alpha_i)\dots,f(\gamma^{s-1}\alpha_i)\right)\in\F^s_q$, where $\alpha_i\in\F_q$ and $\gamma$ is a generator of $\F^{\times}_q$. On the other hand, \cite{krachkovsky2003reed} defined it as $\left(f(\alpha^i),f(\alpha^{i+n})\dots,f(\alpha^{i+(s-1)n})\right)\in\F^s_q$, where $\alpha$ has order exactly $sn$ in $\F^{\times}_q$. This special choice on evaluation points enables \cite{krachkovsky2003reed} to perform FFT-related algorithms and obtain a syndrome-based decoding algorithm for random ``phased burst'' errors.

\subsection{Proof Overview}

We now give an overview of the main technical contributions of our paper.

\paragraph{Interleaved Reed--Solomon Codes} We begin by discussing our new algorithm for interleaved Reed--Solomon codes (\cref{thm:main_RS}). Our algorithm and analysis builds on the techniques used by Bleichenbacher, Kiayias, and Yung \cite{bleichenbacher2003decoding}, but we deviate from their methods in a \change{few ways}. In this IRS decoding algorithm, our main technical contribution is that we use \change{a} monomial-tracking argument (see \cref{subsec:reduction_matrix}) to analyze the performance of this algorithm in the new semi-adversarial error setting.

To describe these improvements, we first discuss the interpolation algorithm of \cite{bleichenbacher2003decoding} (henceforth the ``BKY algorithm''). Consider the $s$-interleaved Reed--Solomon code $\cC := \mathsf{IRS}_{n,k,q,s}(\alpha_1, \dots, \alpha_n)$ with distinct $\alpha_1, \dots, \alpha_n\in\F_q\setminus\{0\}$. We choose polynomials $f_1(X), \hdots, f_s(X)$ of degree at most $k-1$ as our message, encoding them according to $\cC$ and transmit them across the channel. Assume that our received word $\vec{y}:=(\vec{y}_1,\vec{y}_2,\hdots,\vec{y}_n)\in(\mathbb{F}^s_q)^n$, where $\vec{y}_i:=({y}_{i,1},{y}_{i,2},\hdots,{y}_{i,s})\in\F_q^s$. 

The BKY algorithm begins by assuming $\bar{e}$ of the $n$ coordinates of $\vec{y}$ are corrupted. To identify these $\bar{e}$ coordinates, the BKY algorithm adapts the Berlekamp-Welch algorithm by solving the following interpolation problems. Let $E(X)$ be a polynomial of degree at most $\bar{e}$ with unknown coefficients and let $A_1(X), \hdots, A_s(X)$ be polynomials of degree at most $k-1+\bar{e}$ with unknown coefficients. We seek to solve the following interpolation equations:
\begin{align}
A_h(\alpha_i) = y_{i,h}E(\alpha_i) \qquad \forall i \in \{1,2,\hdots, n\}, \forall h \in \{1,2,\hdots, s\}.\label{eq:BKY-interp}
\end{align}
To analyze this system of equations, we consider the block matrix
\begin{align}
    \bB := \begin{pmatrix}
    \bM & & & & -\bN_1\\
    & \bM & & & -\bN_2\\
    & & \ddots & & \vdots\\
    & & & \bM & -\bN_s
    \end{pmatrix},\label{eq:B}
\end{align}
where $\bM \in \F_q^{n \times (k+\be)}$ such that $\bM_{i,j} = \alpha_i^{j-1}$, and for each $h \in [s]$ we have that $\bN_{h} \in \F_q^{n \times \be}$ with $(\bN_h)_{i,j} = y_{i,h}\alpha_i^j$. For the purposes of this overview, assume for simplicity that $\bB$ is a square matrix. That is, $s(k+\be)+\be = sn$. In this scenario, one can prove that the interpolation (\ref{eq:BKY-interp}) succeeds if and only if, we have that $\det \bB \neq 0$. That is, it suffices to bound the probability that $\det \bB \neq 0$ over the random errors in the channel. By the Schwartz-Zippel lemma, it suffices to replace the random errors $y_{i,h}$ with corresponding symbolic variables $Y_{i,h}$ to create a symbolic matrix $\bB^{\SYM}$ and argue that $\det \bB^{\SYM} \neq 0$.

Bleichenbacher, Kiayias, and Yung \cite{bleichenbacher2003decoding} proceeded to argue this by substituting fixed values for some of the $Y_{i,h}$'s and then using an ad-hoc analysis of the determinant to argue that it is nonzero. 
We improve on the BKY algorithm and analysis in \change{two ways}. First, we improve the runtime of the algorithm from roughly $O(n^4)$ time to near-linear time by using the interpolation framework of Alekhnovich~\cite{alekhnovich2005linear} to solve (\ref{eq:BKY-interp}) efficiently.\footnote{After the initial posting of our work, we also learned that the use of Alekhnovich~\cite{alekhnovich2005linear} and similar techniques for fast polynomial interpolation have appeared in related works~\cite{zappatore2020simultaneous,guerrini2021polynomial}.} In particular, the interpolation framework automatically finds the optimal choice of $\bar{e}$ rather than doing a brute-force search like in the BKY algorithm.

Second, we streamline the analysis of the determinant of (\ref{eq:B}). By applying a series of column operations to $\bB$, we can transform the rightmost blocks $\bN_1, \hdots, \bN_s$ into \emph{diagonal} matrices $\bD_1, \hdots, \bD_s$, where each $\bD_h$ corresponds precisely to the errors which occur in the transmission of the codeword. The exact sequence of column operations needed to reveal these errors depends on the original polynomials $f_1, \hdots, f_s$, but we can assume polynomials are known for the purposes of analyzing the rank of $\bB$.

By having the errors be presented transparently as diagonal matrices, arguing that $\det \bB^{\SYM} \neq 0$ can be done in a much more straightforward manner. In particular, it suffices to identify an explicit monomial $Y_{i_1,h_1} \cdots Y_{i_\ell, h_\ell}$ with the following properties:
\begin{itemize}
\item[\textbf{P1.}] There is exactly one expansion of the last $\be$ columns of $\bB^{\SYM}$ which produces that monomial.
\item[\textbf{P2.}] The minor induced by that monomial choice has nonzero determinant.
\end{itemize}
In the purely random error model, finding such a choice of monomials is straightforward--any choice of monomials from distinct columns for which an equal number of variables appear in each row block will work. However, the true power of this perspective is that it gives us a foothold for analyzing a mixture of adversarial and random errors, i.e., the semi-adversarial error model. More precisely, depending on the exact choices for the adversarial errors, the set of monomials which satisfy both P1 and P2 can change \change{substantially}. In particular, there is no ``oblivious'' choice for the monomial. Instead, we give a procedure which constructs a monomial based on the pattern of the adversarial errors--which we may assume we know for the purposes of the analysis. At a high level this procedure is a greedy algorithm, we iterated through the row blocks of $\bB$ and in each block we pick the maximal number of monomials possible such that P1 and P2 can still hold.

\paragraph{Folded Reed--Solomon Codes} Next, we describe the algorithm and analysis leading to \cref{thm:FRS}. The new decoding algorithm is inspired by the previous IRS algorithm plus the Guruswami-Wang \cite{guruswami2013linear} interpolation. The analysis for semi-adversarial errors follows the similar framework as in the IRS case. However, the monomial-tracking is more difficult for FRS codes and we therefore use a new ``greedy algorithm'' to find the special monomial (see \cref{subsec:frs-monomial}). This is the main challenge of extending the analysis for IRS to FRS. 

Let $\mathcal{C}:=\mathsf{FRS}^{s,\gamma}_{n,k,q}(\alpha_1,\alpha_2,\dots,\alpha_n)\subseteq(\mathbb{F}^s_q)^n$ and consider a polynomial $f(X)$ of degree at most $k-1$. After encoding $f$ according to $\cC$ and transmitting it through the $(e_0, \be)$-semi-adversarial channel, let $\vec{y}:=(\vec{y}_1,\vec{y}_2,\hdots,\vec{y}_n)\in(\mathbb{F}^s_q)^n$, where $\vec{y}_i:=({y}_{i,1},{y}_{i,2},\hdots,{y}_{i,s})\in\F_q^s$ is the received \change{word}. Pick parameters $\be \in [n]$ and $L \in [s]$ and consider the following \emph{$L$-length} interpolation equations:
\begin{align}
A_{h}(\gamma^{i-1}\alpha_j)=y_{j,i+h-1}E\left(\gamma^{i-1}\alpha_j\right)\quad i\in [s-L+1],\ j\in [n],\ h \in [L].\label{eq:FRS-interp}
\end{align}
Like for BKY interpolation, these interpolation equations can be solved in near-linear time using the techniques of Alekhnovich~\cite{alekhnovich2005linear}. To the best of our knowledge, our interpolation \change{has not appeared before in the literature}, blending characteristics of both BKY and Guruswami--Wang \cite{guruswami2013linear} interpolation. In particular, compared to other list-decoding algortihms for folded Reed--Solomon codes (such as Guruswami--Wang), the lower-order terms in the runtime for our approach are much smaller. More concretely, our running time is $\tilde{O}(\poly(s)n)$, while the current best implementation of the Guruswami--Wang algorithm takes time $\tilde{O}(2^{\poly(s)}n)$ \cite{goyal2024fast}. A key advantage of our algorithm is that it eliminates the need for the root-finding step required in the Guruswami--Wang algorithm. Instead, after our new interpolation, only a simple polynomial long division is necessary.

To analyze these interpolation equations, we construct a block matrix similar to that of (\ref{eq:B}). Now, the number of rows in each block is $(s-L+1)n$, which we index by a pair $(i,j) \in [s-L+1] \times [n]$. Then, $\bM \in \F_q^{(s-L+1)n \times (k+\be)}$ has the structure $\bM_{(i,j),\ell} = (\gamma^{i-1}\alpha_j)^{\ell-1}$ and for each $h \in [s]$, we have that $(\bN_h)_{(i,j),\ell} = y_{j,i+h-1}(\gamma^{i-1}\alpha_j)^{\ell}$.

Using the techniques of Bleichenbacher, Kiayias, and Yung \cite{bleichenbacher2003decoding}, this matrix would be very difficult to analyze directly, but by applying analogous simplifications to the ones we used for interleaved Reed--Solomon codes, we can again transform the last block-column of $\bB$ to consist only of diagonal matrices. However, since our interpolation is $L$-length, we now have the extra challenge that each error can appear in up to $L$ diagonal matrices (since each $y_{j,i+h-1}$ can appear up to $L$ times).

Like in the IRS case, we transform $\bB$ into a symbolic matrix $\bB^{\SYM}$ with symbolic variables as substitutes for the random variables. We then seek to find a monomial in $\bB^{\SYM}$ satisfying the aforementioned properties P1 and P2. With the combination of repeated symbolic variables as well as the presence of adversarial errors in the semi-adversarial model, the procedure of identifying a valid monomial becomes \change{much} more complex. It is still fundamentally a greedy algorithm like in the IRS case, but much more bookkeeping is needed to ensure that P1 and P2 remain satisfied. See \cref{subsec:frs-monomial} for a description of the monomial identification algorithm and its analysis.

\paragraph{Multiplicity Codes} Next, we describe the adaptations of \cref{thm:FRS} which led to \cref{thm:MULT}. Fundamentally, the interpolation is similar to that of (\ref{eq:FRS-interp}), except we now replace multiplication by $\gamma$ with taking a derivative of the message polynomial $f$. That is, our length-$L$ interpolation equations are of the form,
\begin{align}
A_{h}^{(i-1)}(\alpha_j)=\sum_{\ell=0}^{i-1}\binom{i+h-2-\ell}{h-1} y_{j,i+h-1-\ell}E^{(\ell)}\left(\alpha_j\right)\quad i\in [s-L+1],\ j\in [n],\ h \in [L].\label{eq:MULT-interp}
\end{align}
where $E$ again has degree at most $\be$. Again, this interpolation can be done in near-linear time by using the interpolation algorithm of Alekhnovich~\cite{alekhnovich2005linear}.
Note that due to the chain rule for derivatives, the error term in (\ref{eq:MULT-interp}) turns into a sum. As a result, if we build an interpolation matrix $\bB$ similar to that of (\ref{eq:B}) to analyze (\ref{eq:MULT-interp}), we cannot use column operations to simplify the final column block to only consist of diagonal matrices. Instead, we will have block-diagonal matrices for which each block is a \emph{lower triangular} matrix. The corresponding symbolic matrix $\bB^{\SYM}$ will have each variable appear numerous times, with repetition coming both from the ``length-$L$'' of the interpolation as well as the repeated terms arising from the chain rule.

Again, we seek to identify a symbolic monomial in the expansion of $\det \bB^{\SYM}$ which satisfies properties P1 and P2. The key observation is that within one of these lower triangular matrices, variable appearing in the lower left corner appears exactly once. Thus, by using a greedy algorithm which starts by taking monomials from these lower left corners, we have much more control over the uniqueness of the monomial selected. See Claim~\ref{clm-multMono} for more details about the monomial-picking procedure.

\subsection*{Outline}

In Section~\ref{sec:notation}, we give an overview of the interpolation procedures we use throughout the paper as well as our notation for constructing block matrices. In Section~\ref{sec:combo}, we prove some combinatorial results on decodability against semi-adversarial errors. In Section~\ref{sec:BKY}, we give an improved analysis of the BKY algorithm for interleaved Reed-Solomon codes, proving Theorem~\ref{thm:main_RS}. In Section~\ref{sec:FRS}, we present \change{a} decoder for folded Reed--Solomon codes and prove \cref{thm:FRS}. In Section~\ref{sec:MULT}, we adapt the FRS decoder for multiplicity codes and prove \cref{thm:MULT}. In Section~\ref{sec:concl}, we look toward future work in the semi-adversarial error model, including understanding connections to list decoding. In Appendix~\ref{app:IRS-RS}, we prove an equivalence between decoding interleaved Reed--Solomon codes and subfield Reed--Solomon codes in the semi-adversarial model\change{.}

\subsection*{Acknowledgements}

\change{We thank anonymous reviewers for many helpful comments which have improved the writing of this manuscript.} We thank Sivakanth Gopi and Venkatesan Guruswami for many helpful discussions and much encouragement while writing this paper. We thank Matteo Abbondati, Eleonora Guerrini, and Romain Lebreton for taking the time to explain to us the contributions of \cite{guerrini2023simultaneous,abbondati2024decoding,abbondati2025simultaneous,abbondati2025communication}.  Zihan Zhang also thanks Gabrielle Beck for some helpful discussions. We thank Rishabh Kothary for some valuable corrections to the paper. Joshua Brakensiek was partially supported by Venkatesan Guruswami's Simons Investigator award and National Science Foundation grants No. CCF-2211972 and No. DMS-2503280.. Yeyuan Chen was partially supported by the National Science Foundation grant No. CCF-2236931. Zihan Zhang was partially supported by the National Science Foundation grant No. CCF-2440926.

\section{Notations and Preliminaries}\label{sec:notation}

We now state a number of fundamental facts about computations involving polynomials. For many of these facts, we cite relevant theorems in the textbook of Von Zur Gathen and Gerhard~\cite{von2003modern}.
In what follows, we denote $\mathsf{M}(n)$
by the minimal number of arithmetic operations necessary to multiply two polynomials of degree $n.$ If the ground field
supports fast Fourier transform then $\mathsf{M}(n)= O(n\log n)$, otherwise $\mathsf{M}(n)= O(n\log n\log\log n).$ (Theorem 8.18 and Theorem 8.22 in \cite{von2003modern}). Note that univariate polynomial division can likewise be done in $O(\mathsf{M}(n))$ time, where $n$ is the maximum degree of any of the involved polynomials (Theorem~9.6~\cite{von2003modern}).

\paragraph{Lagrange Interpolation}

Given a field $\F_q$ and $n$ pairs $(\alpha_i, \beta_i) \in \F_q^2$, we define the \emph{Lagrange interpolation} $\Lagrange(X; (\alpha_1,\beta_1), \hdots, (\alpha_n, \beta_n))$ to be the minimal-degree polynomial which passes through these points:
\[
    \Lagrange(X; (\alpha_1,\beta_1), \hdots, (\alpha_n, \beta_n)) := \sum_{i=1}^n \beta_i \prod_{j \in [n] \setminus \{i\}} \frac{X - \alpha_j}{\alpha_i - \alpha_j}.
\]
We note that Von Zur Gathen and Gerhard~\cite{von2003modern} presented a fast implementation of Lagrange interpolation.
\begin{theorem}[Algorithm 10.11 and Corollary 10.12~\cite{von2003modern}]\label{thm:fast-lagrange}
$\Lagrange(X; (\alpha_1,\beta_1), \hdots, (\alpha_n,\beta_n))$ can be computed in time $O(\mathsf{M}(n)\log n).$ 
\end{theorem}

We will also need \emph{Hermite interpolation} when we also want to interpolate Hasse derivatives. Given a field $\F_q$ and $n$ tuples $(\alpha_1,\beta_{1,1},\hdots,\beta_{1,s})\in \F_q^{s+1}$, we define $\Hermite(X; (\alpha_i,\beta_{i,1},\hdots,\beta_{i,s})_{i=1}^n)$ as the minimal degree polynomial $f$ which satisfies $$f^{(i-1)}(\alpha_j)=\beta_{j,i},$$
for $i\in \{1,\hdots,s\}$ and $j\in \{1,\hdots,n\}.$

\begin{theorem}[Algorithm~10.22 and Corollary~10.23 \cite{von2003modern}]\label{thm-fast-Hermite}
    $\Hermite(X; (\alpha_i,\beta_{i,1},\hdots,\beta_{i,s})_{i=1}^n)$ can be computed in time $O(\mathsf{M}(sn)\log(sn))$.
\end{theorem}

In our algorithms, once we find a polynomial $f$ as the potential solution, we need to evaluate it on the chosen evaluation points and check its Hamming distance from the received word. We will need the following folklore fast algorithm for multipoint evaluation.
\begin{theorem}[Algorithm 10.7 and Corollary 10.8 \cite{von2003modern}]\label{thm:evaluation} For any finite field $\F_q$, given $\alpha_1,\dots,\alpha_n\in\F_q$ and a polynomial $f(X)\in\F_q[X]$ with degree $\change{k < n}$ the evaluations $f(\alpha_1),\dots,f(\alpha_n)$ can be computed in time $O(\mathsf{M}(n)\log{n})$.

\end{theorem}

\paragraph{Interpolation Based on Alekhnovich Algorithm}\text{}
A result of Alekhnovich \cite{alekhnovich2005linear} is provided below to solve the following optimization problem.
\begin{problem}[Special Case of \text{\cite[Problem 1]{alekhnovich2005linear}}]\label{prob:alec} Consider an optimization problem stated as follows. Given an input of $h+1$ polynomials $Q_0, Q_1, \dots,Q_h\in\F_q[X]$, we want to output a tuple of polynomials 
$E(X),C_1(X),C_2(X),\dots,C_h(X)\in\F_q[X]$ and integers
$d_0,d_1,d_2,\dots,d_h$ such that
\change{
\begin{align*}
\deg\bigl(X^{k-1}E(X)\bigr)&=d_0\\
\deg\bigl(Q_1(X)\,E(X)+Q_0(X)\,C_1(X)\bigr)&=d_1\\
\deg\bigl(Q_2(X)\,E(X)+Q_0(X)\,C_2(X)\bigr)&=d_2\\
\vdots \quad\quad\quad\quad\quad\quad\quad&\\
\deg\bigl(Q_h(X)\,E(X)+Q_0(X)\,C_h(X)\bigr)&=d_h
\end{align*}
such that $\max(d_0, \hdots, d_h)$ is the smallest positive integer possible.
}
\end{problem}
Alekhnovich \cite{alekhnovich2005linear} provided a quasi-linear algorithm for the above \cref{prob:alec}. In this paper, we will make use of the algorithm as a pure black box, and the theorem of Alekhnovich is stated below.

\begin{theorem}[\text{See \cite[Theorem 1.1]{alekhnovich2005linear}}]\label{thm:alek}
Given any $h+1$ polynomials $Q_0, Q_1, \dots,Q_h\in\F_q[X]$, each of degree at most $n$, there exists an efficient algorithm $\mathsf{Interpolate}\left(Q_0, Q_1, \dots,Q_h\right)$ that outputs a tuple of polynomials $E(X),C_1(X),C_2(X),\dots,$ $C_h(X)\in\F_q[X]$ solving \cref{prob:alec} in time at most $h^{O(1)}\mathsf{M}(n)\log n$.
\end{theorem}

\paragraph{Hasse Derivatives} 

We will use some simple properties of Hasse derivatives (for more details see Proposition~4 in \cite{Dvir2009ExtensionsTT}).

\begin{lemma}\label{lem:HasseProp}
Let $f,g\in \F[x]$ be polynomials and $i,j \in \Z_{\ge 0}$. Then:
\begin{enumerate}
    \item $f^{(i)}(x)+g^{(i)}(x)=(f+g)^{(i)}(x),$
    \item $(fg)^{(i)}(x)=\sum\limits_{\ell=0}^i f^{(i-\ell)}(x) g^{(\ell)}(x),$
    \item $(f^{(i)})^{(j)}(x) = \binom{i+j}{i} f^{(i+j)}(x).$
\end{enumerate}

\end{lemma}

\paragraph{Special Matrices} To ascertain the decoding properties of various algorithms, we bound the ranks of various block matrices. We present some notation to describe the most common components of these block matrices.
Given $\vec{y} \in \F^n_q$, we indicate the diagonal matrix corresponding to $\vec{y}$ to be
\[
    \Diag(\vec{y}) := \begin{pmatrix}
    y_1 & & &\\
    & y_2 & &\\
    & & \ddots &\\
    & & & y_n
    \end{pmatrix}.
\]
In particular, we let $\bI_n := \Diag(1,\hdots, 1)$ be the identity matrix. We also let,
$$\Lower(\vec{y}):=\begin{pmatrix}
y_1  & &  \\
\vdots &\ddots & \\
y_{n} &\hdots & y_{1}
\end{pmatrix}, $$
be the lower triangular matrix associated with $\vec{y}$. Given square matrices $A_1,\hdots,A_t$ of sizes $n_1,\hdots,n_t$ we define $\BlockDiag(A_1,\hdots,A_t)$ as the  block diagonal matrix of size $\sum_{i=1}^t n_t\times \sum_{i=1}^t n_t$ with $A_i$ on the diagonal. 

We also consider \emph{Vandermonde} matrices. Given $\vec{\alpha} \in \F_q^n$ and integers $d_1 \le d_2$, we define the $n \times (d_2 - d_1 + 1)$ Vandermonde matrix to be
\[
    \Vand_{[d_1,d_2]}(\vec{\alpha}) := \begin{pmatrix}
    \alpha_1^{d_1} & \alpha_1^{d_1+1} & \alpha_1^{d_1+2} & \cdots & \alpha_1^{d_2}\\
    \alpha_2^{d_1} & \alpha_2^{d_1+1} & \alpha_2^{d_1+2} & \cdots & \alpha_2^{d_2}\\
    \vdots & \vdots & \vdots & \ddots & \vdots\\
    \alpha_n^{d_1} & \alpha_n^{d_1+1} & \alpha_n^{d_1+2} & \cdots & \alpha_n^{d_2}
    \end{pmatrix}.
\]

To work with multiplicity codes we also need evaluation matrices which encode evaluation of a monomial and its derivatives till a given order. For convenience we let $|w|=\sum_{i=1}^n w_i$ for $w\in \Z^n$ and $m_\ell(x)=x^\ell$ for $\ell \ge 0$. Given $\vec{\alpha} \in \F_q^n, \vec{w}\in \Z_{\ge 0}^n$, integers $d_1\le d_2$, we define the $|w|\times (d_2-d_1+1)$ evaluation matrix to be,

$$\Eval_{[d_1,d_2]}(\vec{\alpha},\vec{w})=\begin{pmatrix}
m_{d_1}(\alpha_1)& m_{d_1+1}(\alpha_1)&\cdots& m_{d_2}(\alpha_1)\\
\vdots & \vdots &\ddots&\vdots\\
m^{(w_1-1)}_{d_1}(\alpha_1)& m^{(w_1-1)}_{d_1+1}(\alpha_1)&\cdots& m^{(w_1-1)}_{d_2}(\alpha_1)\\
\vdots & \vdots &\ddots&\vdots\\
m_{d_1}(\alpha_n)& m_{d_1+1}(\alpha_n)&\cdots& m_{d_2}(\alpha_n)\\
\vdots & \vdots &\ddots&\vdots\\
m^{(w_n-1)}_{d_1}(\alpha_n)& m^{(w_n-1)}_{d_1+1}(\alpha_n)&\cdots& m^{(w_n-1)}_{d_2}(\alpha_n)
\end{pmatrix}.$$

\paragraph{Matrix slices} Given a matrix $\bM \in \F_q^{m \times n}$ and sets $I \subseteq [m]$ and $J \subseteq [n]$, we let $\bM_{I,J}$ denote the $|I| \times |J|$ matrix induces by the rows indexed by $I$ and the rows indexed by $J$.
A common choice of $I$ is $[a,b] := \{a,a+1, \hdots, b\}$.

To simplify common use cases, we let $\bM_I^{\row} := \bM_{I,[n]}$, and $\bM_{J}^{\col} := \bM_{[m],J}$. Furthermore, if either $I$ or $J$ is a singleton, we can replace the set with its single element. That is, $\bM_{i,J} = \bM_{\{i\},J}$, etc.

\change{
\section{Combinatorial Bounds for Semi-adversarial Errors}\label{sec:combo}

In this section, we prove two combinatorial results on the decodability of codes against semi-adversarial errors.

\subsection{A Generalized Singleton Bound}

We first seek to explore the \emph{combinatorial} limits of unique-decoding and list-decoding in the semi-adversarial model.

Recall that a code $\cC \subseteq \Sigma^n$ is $(e, L)$-list-decodable (in the adversarial model) if for any message $\vec{y} \in \Sigma^n$ we have that at most $|\cC \cap \cB(\vec{y}, e)| \le L$. The generalized Singleton bound of Shangguan and Tamo~\cite{ST23} says that for any $L \ge 1$, if a code $\cC \subseteq \Sigma^n$ is $(e,L)$-list decodable then $e \le \frac{L}{L+1}(n-k)$. 

This bound can be extended to the semi-adversarial model. In particular, given $0 \le e_0 \le e \le n$, we say that a code $\cC$ is $(e_0, e, L)$-list-decodable if for any $\vec{c} \in \cC$, when sending $\vec{c}$ through the $(e_0, e)$-semi-adversarial channel, the output message $\vec{y}$ satisfies $|\cB(\vec{y}, e)| \le L$ with high probability, no matter what choices the adversary makes in the channel. We state the generalized semi-adversarial Singleton bound as follows.

\begin{proposition}[GSSB: Generalized Semi-adversarial Singleton Bound]\label{prop:list_optimal}
      Consider any code $\cC \subseteq \Sigma^n$ with $R := \frac{\log |\cC|}{n\log |\Sigma|}$ such that $k := Rn$ is integral. For any $(e_0, e, L)$ with $L + 1 \le |\Sigma|$ and  $\lfloor \frac{e_0}{L} \rfloor>(1-R)n-e$, there exists some word $\vec{z}\in \Sigma^n$ such that $|\mathcal{B}(\vec{z},e_0)\cap \mathcal{C}|\ge 1$ as well as $K\subseteq [n], |K|=e-e_0$ such that any received word $\vec{y} \in \Sigma^n$ that differs with $\vec{z}$ only on positions in $K$ satisfies $|\mathcal{B}(\vec{y},e)\cap \mathcal{C}|\ge L+1$.
\end{proposition}

In other words, there exists $\vec{c} \in C$ such that sending $\vec{c}$ through the $(e_0, e)$-semi-adversarial channel results a message $\vec{y}$ for which there are at least $L$ other codewords of $C$ which are within the minimal ball centered at $\vec{y}$ containing $\vec{c}$. For MDS codes, observe that \cref{prop:optimal} is a special case of \cref{prop:list_optimal} when $L=1$.}

\begin{proof}
Since $\frac{|\cC|}{|\Sigma|^{k-1}} \ge |\Sigma| \ge L+1$, there exists codewords $\vec{c}_0, \hdots, \vec{c}_L \in \cC$ which all agree on their first $k-1$ coordinates. Define $m := n - (k - 1) - (e - e_0)$. If $m$ is negative, decrease $e$ to $e' := m+e = n - (k-1)+ e_0$. Note that
\[
    (1 - R)n - e' = -e_0-1 < \left\lfloor\frac{e_0}{L}\right\rfloor,
\]
so the hypotheses of \Cref{prop:list_optimal} are still satisfied. Furthermore, the desired construction $\vec{z}$ and $K$ is monotone in $e$, so it is sufficient to prove \Cref{prop:list_optimal} for $(e_0, e', L)$. In other words, we may without loss of generality assume that $m$ is nonnegative.

Let $B_0 \sqcup B_1 \sqcup \cdots \sqcup B_L$ be a partition of $\{k, \hdots, k + m - 1 = n - (e-e_0)\}$ such that each part has size $\lfloor m / (L+1)\rfloor$ except for the largest part $B_0$. \change{If $m=0$, then $B_0 = \cdots = B_L = \emptyset$.} We now construct $\vec{z} \in \Sigma^n$ as follows. For all $i \in [n]$,
\[
    z_i = \begin{cases}
    c_{j,i} & \exists j \in [L], i \in B_j,\\
    c_{0,i} & \text{otherwise}.
    \end{cases}
\]
By definition, we have that the Hamming distance between $\vec{z}$ and $\vec{c}_0$ is \change{at most} $|B_1| + \cdots + |B_L| = L\lfloor\frac{m}{L+1}\rfloor$. Since $\lfloor \frac{e_0}{L}\rfloor \ge n - k - e + 1$, we have that $m = n - (k -1) - (e - e_0) \le e_0 + \lfloor\frac{e_0}{L}\rfloor \le \frac{(L+1)e_0}{L}.$ Thus, for each $i \in [L]$, we have that $|B_i| = \lfloor \frac{m}{L+1}\rfloor \le \frac{e_0}{L}$. Thus, he Hamming distance between $\vec{z}$ and $\vec{c}_0$ is at most $|B_1| + \cdots + |B_L| \le e_0$, confirming that $\cB(\vec{z}, e_0) \cap \cC \neq \emptyset$.

Now consider any $\vec{y} \in \Sigma^n$ which agrees with $\vec{z}$ in its first $n - (e-e_0)$ coordinates. Since $\vec{c}_0 \in \cB(\vec{z}, e_0)$, we have that $\vec{c}_0 \in \cB(\vec{y}, e)$.

For any $i \in [L]$, we have that the agreement between $\vec{y}$ and $\vec{c}_i$ is at least $k-1 + |B_i|$. We seek to prove that this agreement is at least $n - e$.

To see why, recall that $\frac{(L+1)e_0}{L} \ge m$. Thus, $e_0 \ge m - \frac{m}{L+1}$. Since $e_0$ is integral, this means that $e_0 \ge m - \lfloor\frac{m}{L+1}\rfloor = m - |B_i|$. Therefore,
\[
    k-1 + |B_i| \ge m - e_0 + k-1 = n - e,
\]
as desired. Thus, $\vec{c}_0, \vec{c}_1, \hdots, \vec{c}_L \in \cB(\vec{y}, e)$, as desired.
\end{proof}

\change{ \Cref{prop:list_optimal} naturally interpolates the known combinatorial decoding bounds for fully adversarial and fully random errors. In particular, when $e_0 = e$, we get precisely the $\left(\frac{L}{L+1}(n-k), L\right)$ generalized Singleton bound (GSB) of Shangguan and Tamo~\cite{ST23}. Further, when $e_0 = 0$, we get an upper bound of $n-k$, which is consistent with the fact that combinatorially random errors can be uniquely decoded up to capacity.\footnote{This can be seen by taking the decoding radius of $\frac{s}{s+1}(n-k)$ for $s$-interleaved Reed--Solomon by the BKY algorithm to the limit as $s \to \infty$.} As a result, we can view \Cref{prop:list_optimal} as suggesting that one can \emph{break} the generalized Singleton bound when we are promised that the channel is not entirely adversarial.}

\change{
\subsection{When Unique Decoding Is Possible}
As our second result of this section, we prove the existential result on combinatorial unique-decodability \cref{prop:RU_semi}. Our proof basically follows from \cite[Theorem 1]{rudra2010two}, but generalizes to the new semi-adversarial errors.}
\begin{proposition}[Restatement of \cref{prop:RU_semi}]
    For any $0<\eps,R<1$, $n>\Omega(1/\eps)$, alphabet $\Sigma$ with size at least $2^{\Omega(1/\eps)}$, MDS code $\mathcal{C}\subseteq \Sigma^n$ with rate $R$ and any codeword $\vec{c}\in\mathcal{C}$, let $\vec{y}$ denote the received word after applying $(e_0,e)$-semi-adversarial  errors to $\vec{c}$, $e_0\leq \min(e,(1-R-\eps)n-e),e\leq (1-R)n$. Then, with probability at least $1-|\Sigma|^{-\Omega(\eps n)}$, we have $\mathcal{B}(\vec{y},e)\cap\mathcal{C}=\{\vec{c}\}$.
\end{proposition}
\begin{proof}
Let $\vec{y}$ denote the corrupted word after applying $(e_0,e)$-semi adversarial errors to $\vec{c}$. From the definition of semi-adversarial errors (see \cref{def:semi-adv}), there exists $J$ with size $|J|\ge n-e$ such that $\vec{y}[j]=\vec{c}[j]$ for all $j\in J$. There is also $J\subseteq I\subseteq[n]$ such that $|I|=n-e_0$ and each entry in $I\setminus J$ of $\vec{y}$ is independently and uniformly sampled from $\Sigma$. Let $e_1= n-|J|$ \change{be the true number of corruptions applied to $\vec{c}$.} \change{We} know $e_0\leq e_1\leq e$. When $e_1< (1-R)n/2$, \change{the number of corruptions is within the unique-decoding radius of MDS codes}, so we can assume $e_1\ge (1-R)n/2$ from now on. If we fix the values of all entries of $\vec{y}$ outside $I\setminus J$ (i.e., Fix all positions containing adversarial errors and agreements), the received word $\vec{y}$ only depends on the random errors sampled on $I\setminus J$. Let $\mathcal{Y}\subseteq\Sigma^n$ denote the set of all these possible words $\vec{y}$ by setting different values on $I\setminus J$, we have $|\mathcal{Y}|=|\Sigma|^{|I\setminus J|}=|\Sigma|^{e_1-e_0}$. We need to show for all but at most $|\Sigma|^{-\Omega(\eps n)}$ fraction of $\vec{y}\in\mathcal{Y}$, we have $\mathcal{B}(\vec{y},e)\cap \mathcal{C}=\{\vec{c}\}$. Let $\mathcal{Y}_0\subseteq \mathcal{Y}$ denote the set of all words $\vec{y}\in\mathcal{Y}$ such that $\{\vec{c}\}\subsetneq\mathcal{B}(\vec{y},e)\cap \mathcal{C}$. We need to upper bound $|\mathcal{Y}_0|$. Concretely, we need to show $\frac{|\mathcal{Y}_0|}{|\mathcal{Y}|}\leq |\Sigma|^{-\Omega(\eps n)}$.

For any $\vec{y}\in\mathcal{Y}_0$, there exists some other $\vec{c}_0\in\mathcal{B}(\vec{y},e)\cap \mathcal{C}$ different from $\vec{c}$. This implies there exist some $A\subseteq J, B\subseteq [n]\setminus J$ such that $|A\cup B|\ge n-e$ and $\vec{c}_0[i]=\vec{y}[i]$ for all $i\in A\cup B$. Another restriction is that this $A$ must have size $|A|<Rn$. Suppose not, since $\vec{c}_0[i]=\vec{y}[i]=\vec{c}[i]$ for all $i\in A$, the hamming distance between $\vec{c}$ and $\vec{c_0}$ must be at most $n-|A|\leq (1-R)n$, which contradicts the requirement that $\mathcal{C}$ is a MDS code with distance $n-Rn+1$. This also implies $|B|\ge |A\cup B|-|A|> (1-R)n-e$. Motivated by this, we can classify all ``bad words'' in $\mathcal{Y}_0$ according to $A,B$ and try to upper bound the size of each class. Formally, for any $A\subseteq J,B\subseteq [n]\setminus J$, we can define the following subset $\mathcal{Y}_{A,B}\subseteq \mathcal{Y}_0$.
\[
\mathcal{Y}_{A,B}=\Bigl\{\vec{y}\in\mathcal{Y}_0\colon \exists \vec{c}_0\in(\mathcal{B}(\vec{y},e)\cap\mathcal{C})\setminus\{c\},\forall i\in A\cup B, \vec{c}_0[i]=\vec{y}[i]\Bigl\}.
\]
From the above discussion, we can classify the elements of $\mathcal{Y}_0$. 
\[\mathcal{Y}_0=\bigcup_{A\subseteq J,B\subseteq [n]\setminus J,|B|> (1-R)n-e,|A\cup B|\ge n-e}\mathcal{Y}_{A,B}.
\]
Now fix arbitrary $A\subseteq J, B\subseteq [n]\setminus J, |B|> (1-R)n-e,|A\cup B|\ge n-e$, we can try to upper bound the size of $\mathcal{Y}_{A,B}$. The key idea is that we can arbitrarily choose some subset $C\subseteq B$ such that $|A\cup C|=Rn$. Then for any two $\vec{y}_1,\vec{y}_2\in\mathcal{Y}_{A,B}$, we have the following claim.
\begin{claim}\label{clm:eq}
For any $\vec{y}_1,\vec{y}_2\in\mathcal{Y}_{A,B}$, if for all $i\in C\cup ([n]\setminus B)$, we have $\vec{y}_1[i]=\vec{y}_2[i]$, then $\vec{y}_1=\vec{y}_2$.
\end{claim}
\begin{proof}[Proof of \cref{clm:eq}]
By the definition of $\mathcal{Y}_{A,B}$, there must be $\vec{c}_1,\vec{c}_2\in\mathcal{C}$ such that $\vec{c}_1[i]=\vec{y}_1[i]=\vec{y}_2[i]=\vec{c}_2[i]$ for all $i\in A\cup C$. Since $|A\cup C|=Rn$ and $\mathcal{C}$ is a MDS code, there must be $\vec{c}_1=\vec{c_2}$. Therefore, we have $\vec{y}_1[i]=\vec{c_1}[i]=\vec{c_2}[i]=\vec{y_2}[i]$ for all $i\in B$. Furthermore,since we know from the statement that $\vec{y}_1$ and $\vec{y}_2$ agree on $[n]\setminus B$, there must be $\vec{y}_1=\vec{y}_2$.
\end{proof}
Therefore, by \cref{clm:eq}, it suffices to know the values on $C\cup ([n]\setminus B)$ to identify an element of $\mathcal{Y}_{A,B}$. By the definition of $\mathcal{Y}$, words in $\mathcal{Y}_{A,B}$ can only differ on $I\setminus J$. Therefore, to identify a word in $\mathcal{Y}_{A,B}$ we only need to know its values in $C\cup (I\setminus (J\cup B))$. From the discussions above we know that
\change{
\begin{align*}
|C\cup (I\setminus (J\cup B))|&\leq |C|+|(I\setminus J)\setminus B|\\
&= Rn-|A|+|I\setminus J|-|B\cap (I \setminus J)| & (|A \cup C| = Rn, A \cap C = \emptyset)\\
&= Rn-|A|+|I\setminus J|-|B\cap I| & (B \cap J = \emptyset)\\
&\leq Rn-|A|+|I\setminus J|-(|B|+|I| - n)\\
&= Rn-|A\cup B|+n-|J| & (A \cap B = \emptyset, J \subseteq I)\\
&\leq Rn+e-(n-e_1)=e+e_1-(1-R)n
\end{align*}
}
Since we have $|\Sigma|$ choices on each index in $C\cup (I\setminus(J\cup B))$, we can derive $|\mathcal{Y}_{A,B}|\leq |\Sigma|^{e+e_1-(1-R)n}$. Finally, we can make a union bound for all possible $A,B$ and get the final upper bound on the size of $\mathcal{Y}_0$.
\[
\frac{|\mathcal{Y}_0|}{|\mathcal{Y}|}\leq \frac{3^n|\Sigma|^{e+e_1-(1-R)n}}{|\Sigma|^{e_1-e_0}}=|\Sigma|^{e+e_0-(1-R-\log_{|\Sigma|}{3})n}=|\Sigma|^{(-\eps+\log_{|\Sigma|}3)n}\]
\change{
When $|\Sigma|>2^{4/\eps}$, we get 
\[
\frac{|\mathcal{Y}_0|}{|\mathcal{Y}|}\leq |\Sigma|^{\left(-\eps + \frac{\log 3}{\log (2^{4/\eps})}\right)n} = |\Sigma|^{\left(\frac{\log 3}{4 \log 2} - 1\right) \eps n}< |\Sigma|^{-\eps n / 2},
\] which completes the proof that all but at most $|\Sigma|^{-\Omega(\eps n)}$ fraction of received words are uniquely decodable.}
\end{proof}

\section{Improved Analysis of Bleichenbacher–Kiayias–Yung Algorithm for Interleaved RS Codes}\label{sec:BKY}

In this section, we prove our main results on probabilistically unique decoding a special class of RS codes (interleaved RS codes) well beyond the unique decoding radius for semi-adversarial errors (see \Cref{thm:main_RS}).

\change{
\subsection{An Overview of Algorithm~\ref{algo:BKY+ALEC}}
Before formally presenting our algorithm for IRS codes (Algorithm~\ref{algo:BKY+ALEC}), we give a brief overview of the main algorithmic ideas. The starting point for our algorithm is the BKY interpolation algorithm~\cite{bleichenbacher2003decoding}. Consider the $s$-interleaved Reed--Solomon code $\cC := \mathsf{IRS}_{n,k,q,s}(\alpha_1, \dots, \alpha_n)$ with distinct $\alpha_1, \dots, \alpha_n\in\F_q\setminus\{0\}$. We choose polynomials $f_1(X), \hdots, f_s(X)$ of degree at most $k-1$ as our message, encoding them according to $\cC$ and transmit them across the channel. Assume that our received word $\vec{y}:=(\vec{y}_1,\vec{y}_2,\hdots,\vec{y}_n)\in(\mathbb{F}^s_q)^n$, where $\vec{y}_i:=({y}_{i,1},{y}_{i,2},\hdots,{y}_{i,s})\in\F_q^s$. 

The BKY algorithm begins by assuming $\bar{e}$ of the $n$ coordinates of $\vec{y}$ are corrupted. To identify these $\bar{e}$ coordinates, the BKY algorithm adapts the Berlekamp-Welch algorithm by solving the following interpolation problems. Let $E(X)$ be a polynomial of degree at most $\bar{e}$ with unknown coefficients and let $A_1(X), \hdots, A_s(X)$ be polynomials of degree at most $k-1+\bar{e}$ with unknown coefficients. We seek to solve the following interpolation equations:
\begin{align*}
A_h(\alpha_i) = y_{i,h}E(\alpha_i) \qquad \forall i \in \{1,2,\hdots, n\}, \forall h \in \{1,2,\hdots, s\}.
\end{align*}

If one can find a non-zero solution to this system of linear equations, then one can compute $f_h(x) = \frac{A_h(x)}{E(x)}$ for all $h \in [s]$, to reconstruct the original codeword. The BKY algorithm proceeds to solve this system of linear equations directly by guessing the appropriate degree $\bar{e}$ of $E$, but we can save time by using the interpolation algorithm of Alekhnovich~\cite{alekhnovich2005linear}.

More precisely, we define $Q_0(X) := \prod_{j=1}^{n}\left(X-\alpha_j\right)$, and for all $h \in [s]$ let $Q_h(X)$ be the minimum-degree polynomialsuch that $Q_h(\alpha_i) = y_{i,h}$ for all $i \in [n]$. Using, the interpolation described in \cref{prob:alec}, we can find ``gcd'' polynomials $E, C_1, \hdots, C_s$ such that the polynomials $B_h(X) := Q_h(X)E(X) + Q_0(X)C_h(X)$ for $h \in [s]$ have overall minimum possible degree. Here, $B_h$ serves essentially the same role as $A_h$ above from which we can deduce $f_h(X) = B_h(X)/E(X)$ for all $h \in [s]$.

The heart of the analysis, both in the random error setting \cite{bleichenbacher2003decoding}, and in the semi-adversarial setting here is arguing this set of polynomial equations have a nonzero solution, which we check through a careful rank-counting argument.
}

\subsection{A Modified Implementation of the Algorithm}
\change{We now} present Algorithm~\ref{algo:BKY+ALEC}, which provides a near-linear-time implementation of the IRS decoding algorithm of Bleichenbacher--Kiayias--Yung~\cite{bleichenbacher2003decoding} by employing the interpolation algorithm of Alekhnovich~\cite{alekhnovich2005linear}.
\begin{algorithm}[h]\label{algo:BKY+ALEC}
    \caption{Near-linear time interleaved RS decoding algorithm}
    \SetAlgoLined
    \KwIn{Parameters $n,k,q,s,e$\;
    Evaluation points $(\alpha_1, \hdots, \alpha_n) \in \F_q^n$\;
    Received \change{word} $\vec{\bf{y}} := ((y_{1,1}, \hdots, y_{1,s}), \hdots, (y_{n,1}, \hdots, y_{n,s})) \in (\F_q^s)^n$. }
    \KwOut{$s$ polynomials $\{f_j(X)\}_{j=1}^s$, each of degree at most $k-1$ or $\Fail$}
    {
    $Q_{0}(X):=\prod_{j=1}^{n}\left(X-\alpha_j\right)$\;
    \ForEach{$h=1,\dots,s$}{
        $Q_h(X):= \Lagrange(X; (\alpha_1, y_{1,h}), \hdots, (\alpha_n, y_{n,h}))$\;
    }
    $E,C_1,\hdots, C_s := \mathsf{Interpolate}(Q_0, Q_1, \dots,Q_s)$\;
    \ForEach{$h=1,\dots,s$}{
        $B_h(X):=Q_h(X)E(X)+Q_0(X)C_h(X)$\;
        $f_h(X) := B_h(X) / E(X)$\;
        \If{$f_h \not\in \F_q[X]$ \textbf{\emph{or}} $\deg{f_h}\ge k$}{
        \Return{$\Fail$}\;
        }
    }
    \If{$\Delta(\vec{\mathbf{y}},\{(f_1(\alpha_i), \hdots, f_s(\alpha_i))\}_{i=1}^n) > e$ }{
        \Return{$\Fail$}\;
    }

    \Return{$(f_1, \hdots, f_s)$}\;
    }
\end{algorithm}

\begin{remark}
    Of note, Eldridge et al. \cite{EBGHJ24} proposed similar lattice-based algorithms for IRS codes under different error models while investigating certain real-world security applications. However, we derive our algorithm by directly integrating the work of Alekhnovich~\cite{alekhnovich2005linear} into the algorithm of Bleichenbacher--Kiayias--Yung \cite{bleichenbacher2003decoding}, whereas Eldridge et al. were inspired by the work of Cohn and Heninger~\cite{cohn2013approximate}. Moreover, as an application-focused paper, \cite{EBGHJ24} did not provide a theoretical analysis on the success probability of their algorithms, and all their algorithmic results are heuristic and experimental. After the initial posting of our work, we also learned that the use of Alekhnovich~\cite{alekhnovich2005linear} and similar techniques for fast polynomial interpolation have appeared in related works~\cite{zappatore2020simultaneous,guerrini2021polynomial}.%
\end{remark}

\paragraph{Running time analysis of Algorithm \ref{algo:BKY+ALEC}} We can use divide-and-conquer trick plus fast polynomial multiplication to compute $Q_0(X)$ in time $O(\mathsf{M}(n)\log{n})$. Then, for each $h\in[s]$ we need to compute the Lagrange interpolation polynomial $Q_h(X)$. Each of them can be computed in time $O(\mathsf{M}(n)\log{n})$ by \cref{thm:fast-lagrange}. Then, from \cref{thm:alek} we can compute $E,C_1,\dots,C_s$ in time $O(s^{O(1)}\mathsf{M}(n)\log{n})$. For each of $h\in[s]$, $B_h(X)$ and $f_h(X)$ can be computed in time $O(\mathsf{M}(n))$. Finally, we need to use fast multipoint evaluation as \cref{thm:evaluation} running in $O(s\mathsf{M}(n)\log{n})$ time to get the codeword, and compare it to the received word to see whether they are close to each other. The total running time can be bounded by $O(s^{O(1)}\mathsf{M}(n)\log{n})\leq O(s^{O(1)}n\log^2{n}\log{\log{n}})$.

\subsection{Improved Analysis for the Semi-Adversarial Errors}

In this section, our aim is to prove the following main theorem of our analysis on Algorithm~\ref{algo:BKY+ALEC} for the semi-adversarial errors.

\begin{theorem}\label{thm:algoirs_analysis}
Consider integers $s,n,k\ge1,e\leq \frac{s}{s+1}(n-k)$.
   For any $s$-interleaved Reed--Solomon code $\mathsf{IRS}_{n,k,q,s}(\alpha_1, \dots, \alpha_n)$ with distinct $\alpha_1, \dots, \alpha_n\in\F_q\setminus \{0\},$ Algorithm~\ref{algo:BKY+ALEC} with input parameters $(n,k,q,s,e)$ can {uniquely decode} any message polynomials $f_1(X),\dots,f_s(X)$ of degree at most $k-1$ from {$(e_0,e)$-semi-adversarial} errors, with probability at least $1-{e}/{q}$, for any $ e_0\leq\min(e,n-k-e).$ The running time of Algorithm~\ref{algo:BKY+ALEC} is $O(s^{O(1)}n\log^2{n}\log{\log{n}})$.
\end{theorem}

In the following subsections, we will fix $s,n,k,e_0,e$ satisfying the conditions in \cref{thm:algoirs_analysis}. We will also fix polynomials $f_1(X),\dots,f_s(X)$, each with degree at most $k-1$ as our message polynomials from the $s$-interleaved Reed--Solomon code $\mathsf{IRS}_{n,k,q,s}(\alpha_1, \dots, \alpha_n)$ with distinct $\alpha_1, \dots, \alpha_n\in\F_q\setminus\{0\}$. Most importantly, the received word $\vec{y}:=(\vec{y}_1,\vec{y}_2,\hdots,\vec{y}_n)\in(\mathbb{F}^s_q)^n$, where $\vec{y}_i:=({y}_{i,1},{y}_{i,2},\hdots,{y}_{i,s})\in\F_q^s$, will always be sampled from the semi-adversarial error models considered in \cref{def:semi-adv} with parameters $e_0,e$.

\subsubsection{Original BKY Interpolation of Fixed Degree}
A key optimization which Algorithm~\ref{algo:BKY+ALEC} makes over the algorithm of Bleichenbacher--Kiayias--Yung~\cite{bleichenbacher2003decoding} is that the latter has to loop over the possible degrees of the interpolation polynomial whereas our use of Alekhnovich~\cite{alekhnovich2005linear} ensures the best polynomial degree is found automatically. However, in order to prove the correctness of Algorithm~\ref{algo:BKY+ALEC}, we first need to reintroduce the ``degree-fixed'' Berlekamp--Welch type interpolation of \cite{bleichenbacher2003decoding}.

\begin{definition}[BKY Interpolation of degree $\be$]\label{def:bky_inter}
    Given any $\be\ge 1$, we say that a polynomial $E(X)$ of degree at most $\be$ together with $s$ polynomials $\{A_i(X)\}_{i=1}^s$, each of degree at most $k-1+\be$ form a solution to  BKY interpolation of degree $\be$ if the following equations
\begin{equation}\label{eq:bky_interpolation}
A_{h}(\alpha_i)=y_{i,h}E\left(\alpha_i\right)\quad i\in\{1,2,\dots,n\}
\end{equation}
hold for $h \in \{1,2,\hdots,s\}$. 
\end{definition}

Since the equations (\ref{eq:bky_interpolation}) defining BKY interpolation are linear, we have that the set of BKY interpolations forms a vector space over $\F_q$. Next, we show that for any message received in the $(e_0, e)$ semi-adversarial model, there exists a suitable degree $\be$ such that the space of BKY interpolations of degree $\be$ has dimension at least one.

From \cref{def:semi-adv}, there exists $I\subseteq[n]$ with $|I|=n-e_0$ such that all adversarial errors occur in some subset $K\subseteq [n]\setminus I$ with $|K|=e_A\leq e_0$. Formally, we define $K$ such that for any position $j\in[n]\setminus I$, $j\in K$ iff there exists some $h\in[s]$ such that $y_{j,h}\neq f_h(\alpha_j)$.
Furthermore, by \cref{def:semi-adv}, there exists a set \(K' = I \setminus J\) with \(\lvert K'\rvert=e_R \leq e - e_0\), such that all positions in \(K'\) experience uniform random errors. 
Let $\be :=e_A+e_R\leq e$, we know for all positions $h\in[s]$ and $ j\in[n]\setminus(K\cup K')$, there must be $y_{j,h}=f_h(\alpha_j)$.
Consider the error locator 
\[E(X):=\prod_{i\in K\cup K'}\left(1-\frac{X}{\alpha_i}\right).\] 
Note that $E(X)$ has degree $\be$. We define $s$ polynomials $A_{1}(x),A_{2}(x),\cdots,A_{s}(x)$ with a degree at most $k-1+\be$ by

\[A_{h}\left(X\right):=f_h\left(X\right)E\left(X\right), \quad  h\in\{1,2,\cdots,s\}.\]
\begin{lemma}\label{lem:BKY-semi-adv-comp}
The set of solutions to BKY interpolations of degree $\be$ has dimension at least one.
\end{lemma}

\begin{proof}We first verify that all the interpolation equations in (\ref{eq:bky_interpolation}) are satisfied by the constructed solution $E(X),A_1(X),$ $A_2(X),\dots,A_s(X)$ introduced above. Let us fix any $h\in[s], i\in[n]$. There are two cases to consider, as discussed below.
\begin{itemize}
    \item[Case 1.] If $i\notin K\cup K'$, from the previous discussions we know $y_{i,h}=f_h(\alpha_i)$ so the corresponding interpolation equation holds.
    \item[Case 2.] If $i\in K\cup K'$, by the construction we know $\alpha_i$ is a root of $E(X)$, and $\alpha_i$ is also a root of  $A_h\left(X\right)$ by its definition. Therefore, both sides of the interpolation equation are zero so it still holds.
\end{itemize}
Note that for any $\beta\in\mathbb{F}_q$, $\{\beta A_1(X),\dots,\beta A_s(X), \beta E(X)\}$ is still a valid solution. Thus, the set of solutions to BKY interpolations with degree $\be$ has dimension at least one.
\end{proof}

Moving forward, for a fixed $\vec{y}$, we let $\mathrm{Sol_{BKY}}:=\{(\beta A_1(X),\dots,\beta A_s(X), \beta E(X)):\beta\in\mathbb{F}_q\}$ be the subspace of \change{intended} solutions to BKY interpolations of degree $\be$ \change{corresponding to the construction solution in \cref{lem:BKY-semi-adv-comp}}.

\subsubsection{Reducing Algorithm~\ref{algo:BKY+ALEC} to BKY Interpolation}

Before proving \cref{thm:algoirs_analysis}, we need an additional lemma, which shows that BKY interpolation of degree $\be$ is equivalent to the interpolation step in Algorithm~\ref{algo:BKY+ALEC}.
\begin{lemma}\label{lem:irsequivalency}
Assume $k>1$. Each non-zero solution $A_1(X),\dots, A_s(X), E(X)$ to BKY interpolation with degree $\be$ satisfies
\begin{compactitem}
\item[(1)] There exists unique $C'_1(X),\dots,C'_s(X)$ such that $A_h(X)=Q_h(X)E(X)+Q_0(X)C'_h(X)$ holds for all $h\in[s]$.

\item[(2)] $0<\max\left(\deg A_1(X), \dots,\deg A_s(X),\deg X^{k-1}E(X)\right)\leq k-1+\be$.
\end{compactitem}
The backward direction also holds. Namely, any possible solution $B_1(X),\dots,B_s(X),E(X)$ of Algorithm~\ref{algo:BKY+ALEC} satisfying $0<\max\left(\deg B_1(X),\dots,\deg B_s(X),\deg X^{k-1}E(X)\right)\leq k-1+\be$ is a non-zero solution to BKY interpolation with degree $\be$.
\end{lemma}
\begin{proof}
We first prove the forward direction: Fix any solution $A_1(X),\dots,A_s(X),E(X)$ to BKY interpolation with degree $\be$. For any $h\in[s]$, it follows that
\[
A_h(X)-YE(X)=(Q_h(X)-Y)E(X)+A_h(X)-Q_h(X)E(X).
\]
Since we know all of $\left\{(X=\alpha_i, Y=y_{i,h})\right\}_{i\in[n]}$ are roots of both  $A_h(X)-YE(X)$ and $Q_h(X)-Y\in\mathbb{F}_q[X,Y]$, there must be $Q_0(X)\mid A_h(X)-Q_h(X)E(X)$. Therefore, let $G_0(X)=-E(X)$,
there must be unique $G_h(X)\in\mathbb{F}_q[X]$ such that
\[
A_h(X)-YE(X)=(Y-Q_h(X))G_0(X)+G_h(X)Q_0(X).
\]
so it implies
\[
A_h(X)=-Q_h(X)G_0(X)+G_h(X)Q_0(X)=Q_h(X)E(X)+Q_0(X)C'_h(X)
\]
If we define $C'_h(X)=G_h(X)$, it is easy to check the max-degree requirement given degree requirement of $A_1(X),\dots,A_s(X),$ $X^{k-1}E(X)$ since they are a solution to BKY interpolation with degree $\be$. We finish the proof for forward direction.

For the backward direction, for any possible solution $\big\{B_h(X)=Q_h(X)E(X)+Q_0(X)C_h(X)\big\}_{h\in[s]},$ $ E(X)$ of Algorithm~\ref{algo:BKY+ALEC} satisfying $0<\max\left(\deg B_1(X),\dots,\deg B_s(X),\deg X^{k-1}E(X)\right)\leq k-1+\be$. It is easy to verify that
\[
B_h(\alpha_i)=y_{i,h}E(\alpha_i)
\]
holds for any $h\in[s],i\in[n]$. By definition, $B_1(X),\dots,B_s(X),E(X)$ must also be a non-zero solution of BKY interpolation with degree $\be$.
\end{proof}

\subsubsection{Reducing \cref{thm:algoirs_analysis} to a Matrix Rank Condition}\label{sec:reducing_irs}

Since we have shown that Algorithm~\ref{algo:BKY+ALEC} is equivalent to BKY interpolation of degree $\be$, we now seek to show analyze the success probability of BKY interpolation in the semi-adversarial model. In accordance with the notation in Section~\ref{sec:notation}, we define the following matrices 
\begin{align*}
    \bM &:= \Vand_{[0,k+\be-1]}(\vec{\alpha}) \in \F_{q}^{n \times (k+\be)}\\
    \bN &:= \Vand_{[1,\be]}(\vec{\alpha}) \in \F_q^{n \times \be}\\
    \forall h \in [s], \bD_h &:= \Diag(\{y_{i,h}\}_{i=1}^n) \in \F_q^{n \times n}\\ 
    \forall h \in [s], \bN_h &:= \bD_h \bN \in \F_q^{n \times \be}.
\end{align*}
Following the original analysis of Bleichenbacher, Kiayias, and Yung \cite{bleichenbacher2003decoding}, we consider the block matrix
\[
    \bB := \begin{pmatrix}
    \bM & & & & -\bN_1\\
    & \bM & & & -\bN_2\\
    & & \ddots & & \vdots\\
    & & & \bM & -\bN_s
    \end{pmatrix}.
\]
We will reduce the correctness of Algorithm~\ref{algo:BKY+ALEC} to proving the full column rankness of this block matrix 
$\bB$. We formalize this ``matrix rank condition'' for 
$\bB$ in the following theorem.
\begin{theorem}\label{thm:110}
     With probability at least $1-{\be}/{q}$, the block matrix $\bB\in\F_q^{sn\times (sk+(s+1)\be)}$ has full column rank.
\end{theorem}
We postpone the proof of the above \cref{thm:110} to \cref{subsec:reduction_matrix}. For the rest of this subsection, we will prove \cref{thm:algoirs_analysis} conditioned on the following theorem, which is a corollary of \cref{thm:110} by solving the BKY interpolation as described in (\ref{eq:bky_interpolation}).
\begin{theorem}\label{thm:bky_uniqueness}
With probability at least $1-{\be}/{q}$, all solutions to BKY interpolation defined in \cref{def:bky_inter} with degree $\be$ are contained in $\mathrm{Sol_{BKY}}$.
\end{theorem}
\begin{proof}
Recall the BKY interpolation in \cref{def:bky_inter}, for any $h\in[s]$, we define the coefficients of these polynomials as
\[A_h(X):=\sum_{i=0}^{k-1+\be}a^{(h)}_iX^i,\]
and coefficients of $E(X)$ as
\[E(X):=b_0+\sum_{i=1}^{\be}b_iX^i.\]
We define vectors $\vec{b}\in\F_q^{\be+1},\vec{a}_h\in\F_q^{k+\be},\vec{\change{y'}}_h\in\F_q^{n}$ by
\begin{align*}
\vec{b} &:=(b_1,\dots, b_{\be}, b_0)^T,\\
\vec{a}_h &:=(a^{(h)}_0,a^{(h)}_1,\dots,a^{(h)}_{k-1+\be})^T\text{, and}\\
\vec{y'}_h &:=(y_{1,h},\dots,y_{n,h})^T.
\end{align*}
Meanwhile, we define the matrix ${\bf Y}\in\F_q^{sn\times 1}$ by
\[{\bf Y}:=(\vec{y'}_1,\vec{y'}_2,\hdots,\vec{y'}_s),\]
and the solution vector $\vec{\bf s}\in\F_q^{s(k+\be)+\be+1}$ by
\[\vec{\bf s}:=(\vec{a}_1,\vec{a}_2,\hdots,\vec{a}_s,\vec{b}).\]
By the basic calculations, we can observe that solving these BKY interpolation equations (i.e. output $A_1(X),\hdots,A_s(X),E(X)\in\F_q[X]$) in (\ref{eq:bky_interpolation}) is equivalent to solving the solution vector $\vec{\bf s}$ of the following linear system that
\[
({\bB}|{\bf Y})\cdot \vec{{\bf s}}=0
\]
If the block matrix $\bB$ has full column rank, then the original interpolation matrix $({\bB}|{\bf Y})$ has a kernel
with dimension at most 1. Since our constructed one-dimensional solutions $\mathrm{Sol_{BKY}}$ are valid, they will be precisely the complete set of valid solutions.
By \cref{thm:110}, we know that,
with probability at least $1-\be/q$, the block matrix $\bB$ has full column rank. Therefore, with probability at least $1-\be/q$, all solutions to the BKY interpolation with degree $\be$ are contained in $\mathrm{Sol_{BKY}}$.
\end{proof}

Now, we are ready to prove \cref{thm:algoirs_analysis}, which is the main result of this section.
\paragraph{Proof of \cref{thm:algoirs_analysis} conditioned on \cref{thm:bky_uniqueness}}
The solution space $\mathrm{Sol_{BKY}}$ contains a non-zero solution to BKY interpolation with degree $\be$. By \cref{lem:irsequivalency}, the minimum max-degree solution $B_1(X),\dots,B_s(X),X^{k-1}E(X)$ found in Algorithm~\ref{algo:BKY+ALEC} must satisfy 
\[0<\max\left(\deg B_1(X),\dots,\deg B_s(X),\deg X^{k-1}E(X)\right)\leq k-1+\be.\]
Moreover, by the backward direction of \cref{lem:irsequivalency}, all possible solutions found in Algorithm~\ref{algo:BKY+ALEC} with max-degree at most $k-1+\be$ are also solutions to BKY interpolation with degree $\be$. Therefore, by \cref{thm:bky_uniqueness}, with probability at least $1-\frac{\be}{q}\ge 1-\frac{e}{q}$, the non-zero solution found in Algorithm~\ref{algo:BKY+ALEC} is in $\mathrm{Sol_{BKY}}$. Since all non-zero solutions in $\mathrm{Sol_{BKY}}$ generate $f_1(X),\dots,f_s(X)$ in Algorithm~\ref{algo:BKY+ALEC} by its definition, it implies that Algorithm~\ref{algo:BKY+ALEC} will output the correct answer with probability at least $1-\frac{e}{q}$.

\subsubsection{Bounding Matrix Rank: Proof of \cref{thm:110}}\label{subsec:reduction_matrix}
This subsection presents our major technical contributions within \cref{sec:BKY} and serves as the key reason why we can improve the analysis of the Bleichenbacher--Kiayias--Yung algorithm for interleaved RS codes.

As a first step, we apply some column operations to $\bB$ to better reveal the underlying errors in the rightmost block column. Recall that $f_1, \hdots, f_s \in \F_q[X]$ of degree at most $k-1$ are the polynomials we encoded. Let $\vec{c} \in (\F_q^s)^n$ be the encoding of these polynomials. That is, for all $i \in [n]$ and $h \in [s]$, we have that $c_{i,h} = f_h(\alpha_i)$. Now, define for each $h \in [s]$ the expansion of $f_h(X)$ to be
\[
    f_h(X) = \sum_{j=0}^{k-1} f_{h,j} X^j.
\]
Thus, for any $i \in [n]$, $h \in [s]$, and $g \in [\be]$ we have that
\begin{align}
    c_{i,h} \alpha_i^g  = \sum_{j=0}^{k-1} f_{h,j} \alpha_i^{j+g}.\label{eq:f-c}
\end{align}
For each $h \in [s]$, define the matrix $\bF_h \in \F_q^{(k+\be) \times \be}$ to be
\[
    (\bF_h)_{a,b} = \begin{cases}
    f_{h,a-b-1} & a-b \in \{1,\hdots, k\}\\
    0 & \text{otherwise}.
    \end{cases}
\]
By (\ref{eq:f-c}), we have that for an $h \in [s]$, we have that
\[
    \bM \bF_h = \Diag(\{c_{i,h}\}_{i=1}^n)\bN.
\]
Therefore, for each $h \in [s]$, if we multiply the $h$th column block of $\bB$ by $\bF_h$ and add it to the final column block, we get a matrix $\bB'$ of the same rank as $\bB$ such that
\[
    \bB' := \begin{pmatrix}
    \bM & & & & \bN'_1\\
    & \bM & & & \bN'_2\\
    & & \ddots & & \vdots\\
    & & & \bM & \bN'_s
    \end{pmatrix},
\]
where for each $h \in [s]$,
\[
    \bN'_{h} := \Diag(\{c_{i,h} - y_{i,h}\}_{i=1}^n)\bN.
\]
In other words, the rightmost column of $\bB'$ only cares about the errors introduced in the transmission of $\vec{c}$ over the $(e_0, e)$ semi-adversarial channel.

Henceforth, assume without loss of generality that indices $\{1, \hdots, e_A\}$ experiences adversarial errors and indices $\{e_A + 1, \hdots, e_A + e_R\}$ experienced random errors. In particular, this means that $c_{i,h} - y_{i,h} = 0$ for all $i \in \{\be+1, \hdots, n\}$ and $h \in [s]$.

Since $\be \le n$, and $\bN$ is a Vandermonde matrix, there exists an invertible matrix $\bP \in \F^{\be \times \be}$ such that $(\bN\bP)_{[\be], [\be]} = \bI_{\be}$. Thus, by applying $\bP$ to the final column block of $\bB'$, we get a matrix $\bB''$ of the same rank with the shape
\[
    \bB'' :=  \begin{pmatrix}
    \bM^{\row}_{[\be]} & & & & \bD'_1\\
    \bM^{\row}_{[\be+1,n]} & & & &\\
    & \bM^{\row}_{[\be]} & & & \bD'_2\\
    & \bM^{\row}_{[\be+1,n]} & & & \\
    & & \ddots & & \vdots\\
    & & & \bM^{\row}_{[\be]} & \bD'_s\\
    & & & \bM^{\row}_{[\be+1,n]} &
    \end{pmatrix},
\]
where for $\bD'_h := \Diag(\{c_{i,h} - y_{i,h}\}_{i=1}^{\be})$ for each $h \in [s]$. Note that the rank of $\bB''$ can depend on the random errors $y_{i,h}$ for $i \in \{e_A+1, \hdots, \be\}$. To analyze these random errors, we first convert $\bB''$ into a symbolic matrix $\bB^{\SYM}$, analyze the rank of this symbolic matrix, and then apply the Schwartz-Zippel Lemma to complete the proof of \cref{thm:110}.

To construct $\bB^{\SYM}$, we have symbolic variables $Y_{i,h}$ for $i \in \{e_A+1, \hdots, \be\}$ and $h \in [s]$. For each $h \in [s]$, we let
\begin{align*}
    \bZ_h &:= \Diag(\{c_{i,h} - y_{i,h}\}_{i=1}^{e_A}),\\
    \bY_h &:= \Diag(\{Y_{i,h}\}_{i=e_A+1}^{\be}).
\end{align*}
We then let
\[
    \bB^{\SYM} := \begin{pmatrix}
    \bM^{\row}_{[e_A]} & & & & \bZ_1 &\\
    \bM^{\row}_{[e_A+1,\be]} & & & & & \bY_1\\
    \bM^{\row}_{[\be+1,n]} & & & & &\\
    &\bM^{\row}_{[e_A]}  & & & \bZ_2 &\\
    &\bM^{\row}_{[e_A+1,\be]}  & & & & \bY_2\\
    &\bM^{\row}_{[\be+1,n]}  & & & &\\
    & & \ddots & & \vdots\\
    &&&\bM^{\row}_{[e_A]}  & \bZ_s &\\
    &&&\bM^{\row}_{[e_A+1,\be]}   & & \bY_s\\
    &&&\bM^{\row}_{[\be+1,n]}   & &\\
    \end{pmatrix}.
\]
To prove that $\bB^{\SYM}$ has full column rank over $\F_q[\{Y_{i,h}\}_{i,h}]$, it suffices to identify a subset of rows of $\bB^{\SYM}$ such that restricting $\bB^{\SYM}$ to these rows results in a square matrix with nonzero determinant.

For each $h \in [s]$, let $T_h \subseteq [e_A]$ be the set of indices $i \in [e_A]$ such that $c_{i,h} - y_{i,h} \neq 0$. Note that by definition of $e_A$, we have that each adversarial error changed at least one value so $T_1 \cup \cdots \cup T_s = [e_A]$. For each $h \in [s]$, let $T'_h := T_h \setminus (T_1 \cup \cdots \cup T_{h-1})$. 

Recall that $\bar{e} \le e \le \frac{s}{s+1}(n-k)$ and that $e_A \le e_0 \le n-k-e \le n-k-\be$. That is, $k+\be+e_A \le n$. Thus, we can pick $n_1, \hdots, n_s \in [n]$ subject to the following constraints.
\begin{align*}
n_h &\ge |T'_h| + k + \be \quad \forall h \in [s],\\
n_1 + \cdots + n_s &= sk + (s+1)\be.
\end{align*}
Then, the matrix
\[
    \widetilde{\bB}^{\SYM} := \begin{pmatrix}
    \bM^{\row}_{[e_A]} & & & & \bZ_1 &\\
    \bM^{\row}_{[e_A+1,\be]} & & & & & \bY_1\\
    \bM^{\row}_{[\be+1,n_1]} & & & & &\\
    &\bM^{\row}_{[e_A]}  & & & \bZ_2 &\\
    &\bM^{\row}_{[e_A+1,\be]}  & & & & \bY_2\\
    &\bM^{\row}_{[\be+1,n_2]}  & & & &\\
    & & \ddots & & \vdots\\
    &&&\bM^{\row}_{[e_A]}  & \bZ_s &\\
    &&&\bM^{\row}_{[e_A+1,\be]}   & & \bY_s\\
    &&&\bM^{\row}_{[\be+1,n_s]}   & &\\
    \end{pmatrix}.
\]
is a square matrix. We seek to prove that $\det \widetilde{\bB}^{\SYM} \neq 0$.
To do this, it suffices to find a monomial in the determinant expansion of $\det \widetilde{\bB}^{\SYM}$ which appears with a nonzero coefficient. To understand  $\det \widetilde{\bB}^{\SYM}$, we expand along the last $\be$ columns of $\widetilde{\bB}^{\SYM}$. In particular, for each $i \in [\be]$, there are (up to) $s$ choices of which term to pick, based on which of the $s$ blocks we choose. To parameterize these choices, we consider a partition $S_1 \sqcup \cdots \sqcup S_s = [\be]$ such that $S_h$ corresponds to the columns picked in the $h$th row block. After expanding along these $\be$ rightmost columns, the remaining matrix is a block-diagonal matrix where each block is of the form $\bM^{\row}_{[n_h] \setminus S_h}$ for some $h \in [s]$. If some of these matrices are not square, then the determinant of the block diagonal matrix will be zero. Otherwise, if $n_h - |S_h| = \be + k$ for all $h \in [s]$, we have that the choice of partition $S_1 \sqcup \cdots \sqcup S_s$ contributes (up to sign) the monomial
\begin{align}
    \prod_{h=1}^s \left[\det \bM^{\row}_{[n_h] \setminus S_h} \prod_{i \in S_h \cap [e_A]} (c_{i,h} - y_{i,h}) \prod_{i \in S_h \cap [e_A+1, \be]} Y_{i,h}\right].\label{eq:part}
\end{align}
For each $h \in [s]$, let $m_h := n_h - \be - k$. We thus have the requirement that $|S_h| = m_h$ for all $h \in [s]$ in order for the contribution of the partition to be nonzero.

To finish, it suffices to find a monomial (i.e., a product of $Y_{i,h}$'s) which appears with a nonzero coefficient for exactly one choice of $S_1, \hdots, S_s$. We construct such a monomial as follows. Recall we selected each $n_h$ such that $n_h \ge |T'_h| + k + \be$. Thus, $m_h \ge |T'_h|$. Let $\ell_h := m_h - |T'_h|$. Since $n_1 + \cdots + n_s = sk + (s+1)\be$ and $|T'_1| + \cdots +|T'_s| = e_A$, we have that $\ell_1 + \cdots + \ell_s = \be - e_A = e_R$. Let $U_1, \hdots, U_s$ be an arbitrary partition of $[e_A + 1, \be]$ such that $|U_h| = \ell_h$ for all $h \in [s]$. Then, the partition $\{S_h := T'_h \cup U_h : h \in [s]\}$ of $[\be]$ corresponds to the monomial
\[
    \prod_{h=1}^s \left[\det \bM^{\row}_{[n_h] \setminus S_h} \prod_{i \in T'_h} (c_{i,h} - y_{i,h}) \prod_{i \in U_h} Y_{i,h}\right].
\]
Note that this monomial has nonzero coefficient since $c_{i,h} - y_{i,h} \neq 0$ for all $i \in T'_h \subseteq T_h$. We also note that this partition $\{U_h\}^s_{h=1}$ restricting to $[e_A+1,\be]$ is the only one which produces the monomial $\prod_{h=1}^s \prod_{i \in U_h} Y_{i,h}$. Consider a partition $S'_1 \sqcup \cdots \sqcup S'_s = [\be]$ which produces the same monomial with non-zero contribution term as in (\ref{eq:part}).  Since the $Y_{i,h}$ are distinct variables for all $i \in [e_A+1, \be]$ and $h \in [s]$, we must have that $S'_h \cap [e_A+1, \be] = U_h$ for all $h \in [s]$. Now, since we must have that $|S'_h| = m_h$ for all $h \in [s]$, we must have that $|S'_h \setminus U_h| = |T'_h|$ for all $h \in [s]$. Furthermore, in order for (\ref{eq:part}) to be nonzero, we must have that $S'_h\setminus U_h \subseteq T_h$ for all $h \in [s]$. \change{At this point, we prove the following claim.}

\change{
\begin{claim}
$S'_h\setminus U_h = T'_h$ for all $h \in [s]$
\end{claim}
\begin{proof}
We prove this by induction on $h$. The base case that $S'_1\setminus U_1 = T_1=T'_1$ is immediate since the cardinalities match. For any $h \ge 2$, note that 
\begin{align*}
S'_h\setminus U_h &\subseteq T_h \setminus ((S'_1 \setminus U_1) \cup \cdots \cup (S'_{h-1} \setminus U_{h-1})) &\text{($S'_h \setminus U_h \subseteq T_h$ and $S'_1, \hdots, S'_{h-1}$ are disjoint from $S'_h$.)}\\
&= T_h \setminus (T'_1 \cup \cdots \cup T'_{h-1}) &\text{(induction hypothesis)}\\
&= T_h \setminus (T_1 \cup \cdots \cup T_{h-1}) &\text{($T'_i = T_i \setminus (T_1 \cup \cdots \cup T_{i-1})$)}\\
&= T'_h,
\end{align*}
so $S'_h \setminus U_h= T'_h$ holds for all $h\in[s]$ since the cardinalities match.
\end{proof}
}

\change{Therefore, we must have that}  $S'_h=(S'_h\setminus U_h)\cup U_h=T'_h\cup U_h=S_h$ for all $h\in[s]$. Thus, the contribution regarding monomial $\prod_{h=1}^s \prod_{i \in U_h} Y_{i,h}$ in $\det{\widetilde{\bB}^{\SYM}}$ is indeed unique with a nonzero coefficient, so $\widetilde{\bB}^{\SYM}$ has nonzero determinant. Therefore, $\bB^{\SYM}$ has full column rank.

To finish, note that $\bB''$ corresponds to $\bB^{\SYM}$ except uniformly random values from $\F_q$ are substituted for each $Y_{i,h}$ with $i \in [e_A+1,\be]$ and $h \in [s]$. Let $\widetilde{\bB}''$ be the restriction of $\bB''$ to the same rows as $\widetilde{\bB}^{\SYM}$ was restricted to relative to $\bB^{\SYM}$. Since $\det \widetilde{\bB}^{\SYM}$ has degree $e_R \le \be$, by the Schwartz-Zippel Lemma \cite{Zip79,Sch80}, we have that $\det \widetilde{\bB}'' \neq 0$ with probability at least $1 - \be/q$. Thus, $\bB''$ (and so $\bB$) has full column rank with probability at least $1 - \be/q$, as desired. This complete the proof of Theorem~\ref{thm:110}.

\section{\change{A} Decoding Algorithm for Folded RS Codes}\label{sec:FRS}
Before this work, decoding algorithms for folded RS codes are primarily based on \cite{Guru06} and \cite{guruswami2013linear}. These list-decoding algorithms can handle $(1-R-\eps)n$ fully adversarial errors if the folding parameter $s$ is large enough, which is almost optimal in terms of decoding radius. However, their framework requires computing a multivariate interpolation polynomial $Q$ and finding all solutions of it, which is computational heavy. Although some later works \cite{kopparty2018improved,goyal2024fast} developed new techniques to speed up this algorithm, the current best-known running time is still $\widetilde{O}(2^{\poly{(1/\eps)}}n)$ due to \cite{goyal2024fast}. It's an interesting open problem how to decode $(1-R-\eps)n$ fully adversarial errors for folded RS codes in $\widetilde{O}(\poly(1/\eps)n)$ or even $\widetilde{O}(n/\eps)$ time.

In this section, we provide a new decoding algorithm for folded RS codes and then prove our main results by bounding the failure probability of the proposed algorithm for semi-adversarial errors (see Theorem~\ref{thm:FRS}). Although our algorithm only works for semi-adversarial errors rather than full adversarial errors, it achieves the running time $\widetilde{O}(\poly(1/\eps)n)$. More specifically, our algorithm depends on \change{a} linear-algebraic interpolation procedure and replaces the traditional root-finding step with a simple polynomial long division.

Throughout this section, we will fix $s,n,k>0$, a finite field $\mathbb{F}_q$, where $|\mathbb{F}_q|>sn>k$, a generator $\gamma$ of $\mathbb{F}^{\times}$, and any sequence of $n$ elements $\alpha_1,\dots,\alpha_n\in\mathbb{F}_q \change{\setminus \{0\}}$ such that $\{\gamma^i\alpha_j\}_{i\in\{0,1,\dots,s-1\},j\in[n]}$ are distinct. We consider the corresponding folded RS code $\mathcal{C}=\mathsf{FRS}^{s,\gamma}_{n,k,q}(\alpha_1,\alpha_2,\dots,\alpha_n)\subseteq(\mathbb{F}^s_q)^n$. For any $f\in\mathbb{F}_q[X]$ with $\deg{f}<k$, we use $\mathcal{C}(f)$ to denote its encoded codeword defined as:
\[
\mathcal{C}(f):=(\mathcal{C}^f_1,\dots,\mathcal{C}^f_n)\in\left(\mathbb{F}^s_q\right)^n,\quad\text{where } \mathcal{C}^f_i:=\Bigl(f(\alpha_i),f(\gamma\alpha_i),\dots,f(\gamma^{s-1}\alpha_i)\Bigl)\in\mathbb{F}^s_q.
\]

\subsection{A New  Interpolation for Folded RS codes}\label{interpo:FRS}
In this subsection, we introduce a novel interpolation procedure that captures the core idea behind our algorithm. Given a received word $\vec{y} := (\vec{y}_1,\dots,\vec{y}_n) \in (\F_q^s)^n$, where $\vec{y}_i:=(y_{i,1},\dots,y_{i,s})\in\mathbb{F}^s_q$, we wish to find all codewords in $\mathcal{C}$ that are close enough with $\vec{y}$. Our interpolation is inspired by \cref{def:bky_inter} and differs from the framework built from \cite{guruswami2013linear}\change{.} However, in order to obtain the ideal radius $(1-R-\eps)n$, we use a similar inner-block shifting parameter to that of \cite{guruswami2013linear}. In order to compare with the previous work and provide some intuition, we will first briefly recap the interpolation of Guruswami--Wang \cite{guruswami2013linear}.

\paragraph{GW Interpolation for FRS Codes}
Fix some decoding parameter $1\leq L\leq s$, the first step of GW interpolation is to find some multivariate polynomial $Q(X,Y_1,\dots,Y_L)\in\mathbb{F}_q[X,Y_1,\dots,Y_L]$ such that
\begin{equation}\label{eq:gw-inter}
Q(\gamma^{i-1}\alpha_j,y_{j,i},\dots,y_{j,i+L-1})=0\quad\quad i\in\{1,2,\dots,s-L+1\}, \text{ }j\in\{1,2,\dots,n\}.
\end{equation}
\change{We remark that \cite{guruswami2013linear} further enforces (\ref{eq:gw-inter}) with multiplicity, but our adaptation does not need to account for this.}
Then, \cite{guruswami2013linear} shows that for any polynomial  $f(X)\in\mathbb{F}_q[X]$ with $\deg{f}<k$ such that the distance  $\dist(\mathcal{C}(f),\vec{y})\leq \frac{L}{L+1}(1-\frac{sR}{s-L+1})n$, The following equality holds. \[T_f(X):=Q(X,f(X),f(\gamma X),\dots,f(\gamma^{L-1}X))=0\in\mathbb{F}_q[X].\] Note that when $L=\Theta(1/\eps)$ and $s=\Theta(1/\eps^2)$, this radius is at least $(1-R-\eps)n$.

The essential reason that GW interpolation obtains the radius $\frac{L}{L+1}(1-\frac{sR}{s-L+1})n$ is that for each $j\in[n]$ such that $\vec{y}_j=\mathcal{C}^f_j$, the interpolation rules (\ref{eq:gw-inter}) guarantee $s-L+1$ roots of $T_f(X)$. Namely, in this case, there must be $T_f(\alpha_j)=T_f(\gamma\alpha_j)=\cdots=T_f(\gamma^{s-L}\alpha_j)=0$. 

\paragraph{Our New FRS Interpolation} GW interpolation suffers from complicated implementations of multivariate interpolation and root-finding procedures. To address this, motivated by the classic BKY interpolation (See \cref{def:bky_inter}), we introduce a new lightweight interpolation rule, which is defined as follows.

\begin{definition}[$L$-length FRS interpolation of degree $\be$]\label{def:frs-inter}
Given any $\be\ge 1$, and decoding parameter $1\leq L\leq s$, we say that a polynomial $E(X)$ of degree at most $\be$ together with $L$ polynomials $\{A_i(X)\}_{i=1}^L$ each of degree at most $k-1+\be$ form a solution to $L$-length FRS interpolation of degree $\be$ if the following equations
\begin{equation}\label{eq:frs_interpolation}
A_{h}(\gamma^{i-1}\alpha_j)=y_{j,i+h-1}E\left(\gamma^{i-1}\alpha_j\right)\quad i\in\{1,2,\dots,s-L+1\},j\in\{1,2,\dots,n\}
\end{equation}
hold for $h \in \{1,2,\hdots,L\}$. 
\end{definition}

 Our new interpolation only consists of lightweight Berlekamp--Welch like equations, and the solution can be obtained by simply a single long division as we will illustrate later. However, in order to achieve the target radius $(1-R-\eps)n$, we choose a specific decoding parameter $1\leq L\leq s$. Moreover, in (\ref{eq:frs_interpolation}) we choose $(s-L+1)n$ ``proposed roots'' $\{\gamma^{i-1}\alpha_j\}_{i\in[s-L+1],j\in[n]}$ as in (\ref{eq:gw-inter}). \cref{def:frs-inter} \change{is similar to} our main FRS decoding algorithm \change{except we fix ahead of time the degree $\bar{e}$}, and plays the similar role as BKY interpolation \cref{def:bky_inter} does in Algorithm \ref{algo:BKY+ALEC}. In the next subsection we first describe our main FRS decoding algorithm and will explain its relation with \cref{def:frs-inter} after that.

\subsection{Description of Our Algorithm}
The following pseudo-code (See Algorithm \ref{algo:main}) describes our fast folded RS codes decoding algorithm that integrates our new ``degree-fixed'' FRS interpolation (\cref{def:frs-inter}) with \cite{alekhnovich2005linear} (see \cref{thm:alek}).

\begin{algorithm}\label{algo:main}
    \caption{Near-linear time folded RS decoding algorithm}
    \SetAlgoLined
    \KwIn{Parameters $n,k,q,s, L, e$\;
    Evaluation points $(\alpha_1, \hdots, \alpha_n) \in \F_q^n$\;
    Received \change{word} $\vec{\bf{y}} := ((y_{1,1}, \hdots, y_{1,s}), \hdots, (y_{n,1}, \hdots, y_{n,s})) \in (\F_q^s)^n$. }
    \KwOut{A polynomials $f(X)$ of degree at most $k-1$ or $\Fail$}
    {
    $Q_{0}(X):=\prod_{i=1}^{s-L+1}\prod_{j=1}^{n}\left(X-\gamma^{i-1}\alpha_j\right)$\;
    \ForEach{$h=1,\dots,L$}{
        $Q_h(X):= \Lagrange\Bigl(X; \{(\gamma^{i-1}\alpha_1, y_{1,i+h-1})\}_{i\in[s-L+1]}, \hdots, \{(\gamma^{i-1}\alpha_n, y_{n,i+h-1})\}_{i\in[s-L+1]}\Bigl)$\;
    }
    $E,C_1,\hdots, C_L := \mathsf{Interpolate}(Q_0, Q_1, \dots,Q_L)$\;
    \change{
    $B_1(X):=Q_1(X)E(X)+Q_0(X)C_1(X)$\;
    $f(X):=B_1(X)/E(X)$\;
    }
    \If{$f \not\in \F_q[X]$ \textbf{\emph{or}} $\deg{f}\ge k$}{
        \Return{$\Fail$}\;
        }
    \If{$\Delta(\vec{\mathbf{y}},\{(f(\alpha_i),f(\gamma\alpha_i),\hdots,f(\gamma^{s-1}\alpha_i))\}_{i=1}^n) > e$}{
        \Return{$\Fail$}\;
    }
    \Return{f}\;
    }
\end{algorithm}
\paragraph{Running Time of Our Algorithm \ref{algo:main}} We can use divide-and-conquer trick plus fast polynomial multiplication to compute $Q_0(X)$ in time $O(\mathsf{M}(sn)\log{sn})$. Then, for each $h\in[L]$ we need to compute the Lagrange interpolation polynomial $Q_h(X)$. Each of them can be computed in time $O(\mathsf{M}(sn)\log{sn})$ by \cref{thm:fast-lagrange}. Then, from \cref{thm:alek} we can compute $E,C_1,\dots,C_L$ in time $O(L^{O(1)}\mathsf{M}(sn)\log{sn})$, and so in time $O(s^{O(1)}\mathsf{M}(n)\log{n})$ since $L\leq s$. $B(X)$ and $f(X)$ can be computed in time $O(\mathsf{M}(sn))$. Finally, we need to use fast multipoint evaluation as \cref{thm:evaluation} running in $O(\mathsf{M}(sn)\log{sn})$ time to get the codeword, and compare it to the received word to see whether they are close to each other. The total running time can be bounded by $O(s^{O(1)}\mathsf{M}(n)\log{n})\leq O(s^{O(1)}n\log^2{n}\log{\log{n}})$.
\subsection{Analysis of Algorithm \ref{algo:main} for the Semi-Adversarial Errors}

In this section, our aim is to prove the following main theorem of our analysis on Algorithm \ref{algo:main} for
the semi-adversarial errors.

\begin{theorem}\label{thm:algofrs_analysis}
Consider integers $s,L,n,k,e\ge 1$, where $L\leq s, e \leq \frac{L}{L+1}\left(n-\frac{k}{s-L+1}\right)$, and a generator $\gamma$ of $\mathbb{F}^{\times}_q$. For any $(s,\gamma)$-folded Reed-Solomon code $\mathsf{FRS}_{n,k}^{s,\gamma}(\alpha_1, \dots, \alpha_n)$ with appropriate evaluation points $\alpha_1,\hdots,\alpha_n\in\mathbb{F}_q\setminus\{0\}$, Algorithm~\ref{algo:main} with input parameters $(n,k,q,s,L,e)$ can uniquely decode any message polynomial $f$ of degree at most $k-1$ from {$(e_0,e)$-semi-adversarial} errors, with probability at least $1-e(s-L+1)/q$, for any $e_0\leq\min\left(e,n-e-\frac{k}{s-L+1}\right)$. The running time of Algorithm~\ref{algo:main} is at most $O(s^{O(1)}n\log^2{n}\log{\log{n}})$.
\end{theorem}

The remaining part of this section serves for proving \cref{thm:algofrs_analysis}. We will fix $s,L,n,k,e_0,e,\gamma$ satisfying the conditions in \cref{thm:algofrs_analysis}. We will also fix the message polynomial $f(X)\in\mathbb{F}_q[X]$ with degree at most $k-1$ from the $(s,\gamma)$-folded RS code $\mathsf{FRS}_{n,k}^{s,\gamma}(\alpha_1, \dots, \alpha_n)$ with appropriate evaluation points $\alpha_1,\hdots,\alpha_n\in\mathbb{F}_q\setminus\{0\}$. Most importantly, the received word $\vec{y}:=(\vec{y}_1,\vec{y}_2,\hdots,\vec{y}_n)\in(\mathbb{F}^s_q)^n$, where $\vec{y}_i:=(y_{i,1},y_{i,2},\hdots,y_{i,s})\in\mathbb{F}^s_q$ will always be sampled from the semi-adversarial error models considered in \cref{def:semi-adv} with parameters $e_0,e$.
\subsubsection{Reducing Algorithm \ref{algo:main} to $L$-length FRS Interpolation of Fixed Degree}
Recall that Algorithm \ref{algo:BKY+ALEC} is actually an efficient and automatic way to find the exact degree for fixed-degree BKY interpolation (See \cref{def:bky_inter}). Here, the mission of Algorithm \ref{algo:main} is to find an appropriate degree $\be$ so that $L$-length FRS interpolation of degree $\be$ (See \cref{def:frs-inter}) succeeds with high probability. Therefore, we will first identify a specific $\be$ and see how it works well.

\paragraph{Choose Interpolation Degree and Error Set} First let's identify the set of $L$-length blocks that are truly corrupted adversarially. From \cref{def:semi-adv}, there exists $I\subseteq [n]$ with $|I|=n-e_0$ such that all adversarial errors occur in $[n]\setminus I$. Therefore, we can identify the set of adversarial $L$-length ``inner-block'' errors. Formally, the elements in this set indexed by their ``head indices'' are defined as follows. 
\begin{equation}\label{eq:frs-adv-set}
P:=\Bigl\{(i,j)\in[s-L+1]\times([n]\setminus I):(y_{j,i},\dots,y_{j,i+L-1})\neq (f(\gamma^{i-1}\alpha_j),\dots,f(\gamma^{i+L-2}\alpha_j))\Bigl\}.
\end{equation}
We define $e_A:=|P|$ as the number of such $L$-length ``inner-block'' adversarial errors.

We need to also include $L$-length random errors. By \cref{def:semi-adv}, there must be some $K'=I\setminus J$ with $|K'|=e'\leq e-e_0$ such that all position in $K'$ experience random errors. Let $P':=[s-L+1]\times K'$ and $e_R=|P'|=e'(s-L+1)$, we know for all $(i,j)\notin P\cup P'$, there must be
\begin{equation}\label{eq:frs-agree}
(y_{j,i},\dots,y_{j,i+L-1})= (f(\gamma^{i-1}\alpha_j),\dots,f(\gamma^{i+L-2}\alpha_j)).
\end{equation}
This tells us all corrupted $L$-length errors occur in $L$-length segments with indices in $P\cup P'$. Therefore, let $\be:=e_A+e_R$, we can consider the error locator 
\[
E(X):=\prod_{(i,j)\in P\cup P'}\left(1-\frac{X}{\gamma^{i-1}\alpha_j}\right)
.\]

Note that $E(X)$ has degree $\be$. We define $L$ polynomials $A_{1}(x),A_{2}(x),\cdots,A_{L}(x)$ with degree at most $k-1+\be$ as

\begin{equation}\label{eq:frs-defsol}
A_{h}\left(\frac{X}{\gamma^{h-1}}\right)=f\left(X\right)E\left(\frac{X}{\gamma^{h-1}}\right) \quad h\in\{1,2,\cdots,L\}
\end{equation}

\paragraph{\change{Intended} Solutions to FRS Interpolation} As we have identified the \change{target} interpolation degree $\be$, next we aim to identify a specific one-dimensional subspace of \change{intended} solutions to FRS interpolations. This resorts to the following lemma.
\begin{lemma}
The set of solutions to $L$-length FRS interpolations of degree $\be=e_A+e_R$ (See \cref{def:frs-inter}) has dimension at least one.
\end{lemma}
\begin{proof}
We now verify that all the interpolation equations in (\ref{eq:frs_interpolation}) are satisfied by the constructed solution $E(X),A_1(X),$ $A_2(X),\dots,A_L(X)$ introduced above. Let us fix any $h\in[L], i\in[s-L+1],j\in[n]$. There are two cases to consider, as discussed below.
\begin{itemize}
    \item[Case 1.] If $(i,j)\notin P\cup P'$, from (\ref{eq:frs-agree}) we know $y_{j,i+h-1}=f(\gamma^{i+h-2}\alpha_j)$. Then, there must be $A_h(\gamma^{i-1}\alpha_j)=y_{j,i+h-1}E(\gamma^{i-1}\alpha_j)$ by (\ref{eq:frs-defsol}). Therefore, the corresponding interpolation equation holds.
    \item[Case 2.] If $(i,j)\in P\cup P'$, by the construction we know $\gamma^{i-1}\alpha_j$ is a root of $E(X)$, and $\gamma^{i-1}\alpha_j$ is also a root of  $A_h\left(X\right)$ by its definition. Therefore, both sides of the interpolation equation are zero so it still holds.
\end{itemize}
Note that for any $\beta\in\mathbb{F}_q$, $\{\beta A_1(X),\dots,\beta A_s(X), \beta E(X)\}$ is still a valid solution. Thus, the set of solutions to $L$-length FRS interpolations with degree $\be$ has dimension at least one.
\end{proof}

Moving forward, for a fixed $\vec{y}$, we let $\mathrm{Sol_{FRS}}:=\{(\beta A_1(X),\dots,\beta A_s(X), \beta E(X)):\beta\in\mathbb{F}_q\}$ be the subspace of \change{intended} solutions to $L$-length FRS interpolations of degree $\be$.

\paragraph{Equivalency Between Solution Spaces} We have shown an \change{intended} set of solutions $\mathrm{Sol_{FRS}}$ in the solution space of $L$-length FRS interpolation of degree $\be$. In order to reduce the solution space of Algorithm \ref{algo:main} to it, we need the following equivalency lemma in degree-bounded scenarios, which is similar to \cref{lem:irsequivalency}.
\begin{lemma}\label{lem:frsequivalency}
Assume $k>1$. Each non-zero solution $A_1(X),\dots, A_L(X), E(X)$ to $L$-length FRS interpolation of degree $\be$ satisfies
\begin{compactitem}
\item[(1)] There exists unique $C'_1(X),\dots,C'_L(X)$ such that $A_h(X)=Q_h(X)E(X)+Q_0(X)C'_h(X)$ holds for all $h\in[L]$.

\item[(2)] $0<\max\left(\deg A_1(X),\dots,\deg A_L(X),\deg X^{k-1}E(X)\right)\leq k-1+\be$.
\end{compactitem}
The backward direction also holds. Namely, any possible solution $B_1(X),\dots,B_L(X),E(X)$ of Algorithm~\ref{algo:main} satisfying $0<\max\left(\deg B_1(X),\dots,\deg B_L(X),\deg X^{k-1}E(X)\right)\leq k-1+\be$ is a non-zero solution to $L$-length FRS interpolation of degree $\be$.
\end{lemma}
\begin{proof}
We first prove the forward direction: Fix any solution $A_1(X),\dots,A_L(X),E(X)$ to $L$-length FRS interpolation of degree $\be$. For any $h\in[L]$, we know all of $\left\{(X=\gamma^{i-1}\alpha_j, Y=y_{j,i+h-1})\right\}_{i\in[s-L+1],j\in[n]}$ are roots of  $A_h(X)-YE(X)\in\mathbb{F}_q[X,Y]$. From the similar argument to the proof of \cref{lem:irsequivalency}, there must exist unique $G_0(X)=-E(X), G_h(X)\in\mathbb{F}_q[X]$ such that
\[
A_h(X)-YE(X)=(Y-Q_h(X))G_0(X)+G_h(X)Q_0(X)
\]
so it implies
\[
A_h(X)=-Q_h(X)G_0(X)+G_h(X)Q_0(X)=Q_h(X)E(X)+Q_0(X)C'_h(X)
\]
If we define $C'_h(X)=G_h(X)$, it is easy to check the max-degree requirement given degree requirement of $A_1(X),\dots,A_L(X),$ $X^{k-1}E(X)$ since they are a solution to $L$-length FRS interpolation of degree $\be$. We finish the proof for forward direction.

For the backward direction, for any possible solution $\big\{B_h(X)=Q_h(X)E(X)+Q_0(X)C_h(X)\big\}_{h\in[L]}$, $ E(X)$ of Algorithm~\ref{algo:main} satisfying $0<\max\left(\deg B_1(X),\dots,\deg B_L(X),\deg X^{k-1}E(X)\right)\leq k-1+\be$, it is easy to verify that
\[
B_h(\gamma^{i-1}\alpha_j)=y_{j,i+h-1}E(\gamma^{i-1}\alpha_j)
\]
holds for any $h\in[L],i\in[s-L+1],j\in[n]$. By definition, $B_1(X),\dots,B_L(X),E(X)$ must also be a non-zero solution of $L$-length FRS interpolation of degree $\be$.
\end{proof}

\subsubsection{Reducing \cref{thm:algofrs_analysis} to a Matrix Rank Condition}

Since we have shown that the solutions spaces of Algorithm~\ref{algo:main} and $L$-length FRS interpolation of degree $\be$ are equivalent, we now formally show how analyzing the success probability of fixed-degree FRS interpolation lead to \cref{thm:algofrs_analysis} in the semi-adversarial model. In accordance with the notation in Section~\ref{sec:notation}, we define the following matrices 
\begin{align*}
    \bM &:= \Vand_{[0,k+\be-1]}(\{\gamma^{i-1}\alpha_j\}_{i\in[s-L+1],j\in[n]}) \in \F_{q}^{(s-L+1)n \times (k+\be)}\\
    \bN &:= \Vand_{[1,\be]}(\{\gamma^{i-1}\alpha_j\}_{i\in[s-L+1],j\in[n]}) \in \F_q^{(s-L+1)n \times \be}\\
    \forall h \in [L], \bD_h &:= \Diag(\{y_{j,i+h-1}\}_{i\in[s-L+1],j\in[n]}) \in \F_q^{(s-L+1)n \times (s-L+1)n}\\ 
    \forall h \in [L], \bN_h &:= \bD_h \bN \in \F_q^{(s-L+1)n \times \be}.
\end{align*}
Similar to \cref{sec:reducing_irs}, we consider the block matrix
\[
    \bB := \begin{pmatrix}
    \bM & & & & -\bN_1\\
    & \bM & & & -\bN_2\\
    & & \ddots & & \vdots\\
    & & & \bM & -\bN_L
    \end{pmatrix}.
\]
We will reduce the correctness of Algorithm~\ref{algo:main} to proving the full column rankness of this block matrix 
$\bB$. We formalize this ``matrix rank condition'' for 
$\bB$ in the following theorem.
\begin{theorem}\label{thm:frs110}
     With probability at least $1-{\be}/{q}$, the block matrix $\bB\in\F_q^{(s-L+1)Ln\times (Lk+(L+1)\be)}$ has full column rank.
\end{theorem}
We postpone the proof of the above \cref{thm:frs110} to \cref{subsec:frs-reduction_matrix}. For the rest of this subsection, we will prove \cref{thm:algofrs_analysis} conditioned on the following theorem, which is a corollary of \cref{thm:frs110} by solving the $L$-length FRS interpolation of degree $\be$ as described in (\ref{eq:frs_interpolation}).
\begin{theorem}\label{thm:frs_uniqueness}
With probability at least $1-{\be}/{q}$, all solutions to $L$-length FRS interpolation defined in \cref{def:frs-inter} of degree $\be$ are contained in $\mathrm{Sol_{FRS}}$.
\end{theorem}
\begin{proof}
Recall the $L$-length FRS interpolation in \cref{def:frs-inter}, for any $h\in[L]$, we define the coefficients of these polynomials as
\[A_h(X):=\sum_{i=0}^{k-1+\be}a^{(h)}_iX^i,\]
and coefficients of $E(X)$ as
\[E(X):=b_0+\sum_{i=1}^{\be}b_iX^i.\]
We define vectors $\vec{b}\in\F_q^{\be+1},\vec{a}_h\in\F_q^{k+\be},\vec{y'}_h\in\F_q^{(s-L+1)n}$ by
\begin{align*}
\vec{b} &:=(b_1,\dots, b_{\be}, b_0)^T,\\
\vec{a}_h &:=(a^{(h)}_0,a^{(h)}_1,\dots,a^{(h)}_{k-1+\be})^T\text{, and}\\
\vec{y'}_h &:=(\{y_{1,i+h-1}\}^{s-L+1}_{i=1},\dots,\{y_{n,i+h-1}\}^{s-L+1}_{i=1})^T.
\end{align*}
Meanwhile, we define the matrix ${\bf Y}\in\F_q^{(s-L+1)Ln\times 1}$ by
\[{\bf Y}:=(\vec{y'}_1,\vec{y'}_2,\hdots,\vec{y'}_L),\]
and the solution vector $\vec{\bf s}\in\F_q^{L(k+\be)+\be+1}$ by
\[\vec{\bf s}:=(\vec{a}_1,\vec{a}_2,\hdots,\vec{a}_L,\vec{b}).\]
By the basic calculations, we can observe that solving these $L$-length FRS interpolation equations (i.e. output $A_1(X),\hdots,$ $A_L(X),$ $E(X)\in\F_q[X]$) in (\ref{eq:frs_interpolation}) is equivalent to solving the solution vector $\vec{\bf s}$ of the following linear system that
\[
({\bB}|{\bf Y})\cdot \vec{{\bf s}}=0
\]
If the block matrix $\bB$ has full column rank, then the original interpolation matrix $({\bB}|{\bf Y})$ has a kernel
with dimension at most 1. Since our constructed one-dimensional solutions $\mathrm{Sol_{FRS}}$ are valid, they will be precisely the complete set of valid solutions.
By \cref{thm:frs110}, we know that,
with probability at least $1-\be/q$, the block matrix $\bB$ has full column rank. Therefore, with probability at least $1-\be/q$, all solutions to the $L$-length FRS interpolation of degree $\be$ are contained in $\mathrm{Sol_{FRS}}$.
\end{proof}

Now, we are ready to prove \cref{thm:algofrs_analysis}, which is the main result of this section.
\paragraph{Proof of \cref{thm:algofrs_analysis} conditioned on \cref{thm:frs_uniqueness}}
The solution space $\mathrm{Sol_{FRS}}$ contains a non-zero solution to $L$-length FRS interpolation of degree $\be$. By \cref{lem:frsequivalency}, the minimum max-degree solution $B_1(X),\dots,B_L(X),X^{k-1}E(X)$ found in Algorithm~\ref{algo:main} must satisfy 
\[0<\max\left(\deg B_1(X),\dots,\deg B_L(X),\deg X^{k-1}E(X)\right)\leq k-1+\be.\]
Moreover, by the backward direction of \cref{lem:frsequivalency}, all possible solutions found in Algorithm~\ref{algo:main} with max-degree at most $k-1+\be$ are also solutions to $L$-length FRS interpolation of degree $\be$. Therefore, by \cref{thm:frs_uniqueness}, with probability at least $1-\be/q\ge 1-e(s-L+1)/q$, the non-zero solution found in Algorithm~\ref{algo:main} is in $\mathrm{Sol_{FRS}}$. Since all non-zero solutions in $\mathrm{Sol_{FRS}}$ generate $f(X)$ in Algorithm~\ref{algo:main} by its definition, it implies that Algorithm~\ref{algo:main} will output the correct answer with probability at least $1-e(s-L+1)/q$.
\subsubsection{Bounding Matrix Rank: Proof of \cref{thm:frs110}} \label{subsec:frs-reduction_matrix}
The only remaining task is to prove \cref{thm:frs110}. The basic framework of the proof is similar to \cref{subsec:reduction_matrix}. However, in folded RS codes case there is one more technical difficulty. Concretely, in our previous proof for Interleaved RS codes there are no repeated formal variables in the interpolation matrix, which facilitates our arguments on the monomial uniqueness. However, there are repeated formal variables when we are addressing folded RS codes, so we need to choose the monomial in a more clever way. We will illustrate it in detail later.   

We first perform a series of column operations on $\bB$ to better reveal the underlying errors in the rightmost block column. Recall that $f \in \F_q[X]$ of degree at most $k-1$ is the message polynomials we would like to decode. Let $\vec{c}=(\vec{c}_1,\hdots,\vec{c}_n)=\mathcal{C}(f) \in (\F_q^s)^n$ be the corresponding codeword where $\vec{c}_i=(f(\alpha_i),f(\gamma\alpha_i),\hdots,f(\gamma^{s-1}\alpha_i))=(c_{i,1},\dots,c_{i,s})\in\mathbb{F}^s_q$. Now, we explicitly write the coefficients of $f(X)$ as
\[
    f(X) = \sum_{j=0}^{k-1} f_{j} X^j.
\]
Thus, for any $i \in [s-L+1],j\in[n]$, $h \in [L]$, and $g \in [\be]$ we have that
\begin{align}
    c_{j,i+h-1} (\gamma^{i-1}\alpha_j)^g  = \sum_{t=0}^{k-1} f_{t}(\gamma^{i+h-2}\alpha_j)^t(\gamma^{i-1} \alpha_j)^{g}=\sum_{t=0}^{k-1}f_t\gamma^{(h-1)t}(\gamma^{i-1}\alpha_j)^{t+g}.\label{eq:frs-f-c}
\end{align}
For each $h \in [L]$, define the matrix $\bF_h \in \F_q^{(k+\be) \times \be}$ to be
\[
    (\bF_h)_{a,b} = \begin{cases}
    f_{a-b-1}\gamma^{(h-1)(a-b)} & a-b \in \{1,\hdots, k\}\\
    0 & \text{otherwise}.
    \end{cases}
\]
By (\ref{eq:frs-f-c}), it follows that for an $h \in [L]$, we have that
\[
    \bM \bF_h = \Diag(\{c_{j,i+h-1}\}_{i\in[s-L+1],j\in[n]})\bN.
\]
Therefore, for each $h \in [L]$, if we multiply the $h$-th column block of $\bB$ by $\bF_h$ and add it to the final column block, we get a matrix $\bB'$ of the same rank as $\bB$ such that
\[
    \bB' := \begin{pmatrix}
    \bM & & & & \bN'_1\\
    & \bM & & & \bN'_2\\
    & & \ddots & & \vdots\\
    & & & \bM & \bN'_L
    \end{pmatrix},
\]
where for each $h \in [L]$,
\[
    \bN'_{h} := \Diag(\{c_{j,i+h-1} - y_{j,i+h-1}\}_{i\in[s-L+1],j\in[n]})\bN.
\]
In other words, the rightmost column of $\bB'$ only cares about the errors introduced in the transmission of $\vec{c}$ over the $(e_0, e)$ semi-adversarial channel.

Recall that $P$ defined in (\ref{eq:frs-adv-set}) are the set of $L$-length ``inner-block'' adversarial errors with size $|P|=e_A$. Let $(i_1,j_1),\hdots,(i_{e_A},j_{e_A})$ denote the all  elements in $P$ (listed in arbitrary order). Also, without loss of generality we can assume the random error set $K'=[e']$. Then, for each $h\in[L]$, since the rows of $({\bM}|\bN'_h)$ can be naturally indexed by $[s-L+1]\times[n]$, we can permute these rows such that the rows are presented in the following order from top to bottom.
\begin{align*}
&(i_1,j_1),\hdots,(i_{e_A},j_{e_A}),(1,1),\hdots,(s-L+1,1),\hdots,(1,e'),\hdots,(s-L+1,e'),\\&(i_{\be+1},j_{\be+1})\hdots,(i_{(s-L+1)n},j_{(s-L+1)n}).
\end{align*}
Namely, the first $e_A$ rows are indices representing the adversarial error locations in $P$, followed by $e_R=(s-L+1)e'$ indices representing the random error locations in $P'=[s-L+1]\times K'$. The remaining indices $(i_{\be+1},j_{\be+1}),\hdots,(i_{(s-L+1)n},j_{(s-L+1)n})$ in the end are agreement locations $([s-L+1]\times[n])\setminus (P\cup P')$ listed in arbitrary order. In the following, we will use $\{(i_t,j_t)\}^{(s-L+1)n}_{t=1}$ to index the rows as the order above.

Henceforth, by (\ref{eq:frs-agree}) this implies that $c_{j_t,i_t+h-1} - y_{j_t,i_t+h-1} = 0$ for all $t \in \{\be+1, \hdots, (s-L+1)n\}$ and $h \in [L]$.

Since $\be \le (s-L+1)n$, and $\bN$ is a Vandermonde matrix, there exists an invertible matrix $\bP \in \F^{\be \times \be}$ such that $(\bN\bP)_{[\be], [\be]} = \bI_{\be}$. Thus, by applying $\bP$ to the final column block of $\bB'$, we get a matrix $\bB''$ of the same rank with the shape
\[
    \bB'' :=  \begin{pmatrix}
    \bM^{\row}_{[\be]} & & & & \bD'_1\\
    \bM^{\row}_{[\be+1,n]} & & & &\\
    & \bM^{\row}_{[\be]} & & & \bD'_2\\
    & \bM^{\row}_{[\be+1,n]} & & & \\
    & & \ddots & & \vdots\\
    & & & \bM^{\row}_{[\be]} & \bD'_L\\
    & & & \bM^{\row}_{[\be+1,n]} &
    \end{pmatrix},
\]
where for $\bD'_h := \Diag(\{c_{j_t,i_t+h-1} - y_{j_t,i_t+h-1}\}^{\be}_{t=1})$ for each $h \in [L]$. Note that the rank of $\bB''$ can depend on the random errors $y_{j_t,i_t+h-1}$ for $t \in \{e_A+1, \hdots, \be\}$. To analyze these random errors, we first convert $\bB''$ into a symbolic matrix $\bB^{\SYM}$, analyze the rank of this symbolic matrix, and then apply the Schwartz-Zippel Lemma to complete the proof of \cref{thm:frs110}.

To construct $\bB^{\SYM}$, we have algebraic independent symbolic variables $Y_{j,u}$ for all $j \in [e']$ and $u \in [s]$. For each $h \in [L]$, we let
\begin{align*}
    \bZ_h &:= \Diag(\{c_{j_t,i_t+h-1} - y_{j_t,i_t+h-1}\}_{t=1}^{e_A}),\\
    \bY_h &:= \Diag(\{Y_{j_t,i_t+h-1}\}_{t=e_A+1}^{\be}).
\end{align*}
We then let
\[
    \bB^{\SYM} := \begin{pmatrix}
    \bM^{\row}_{[e_A]} & & & & \bZ_1 &\\
    \bM^{\row}_{[e_A+1,\be]} & & & & & \bY_1\\
    \bM^{\row}_{[\be+1,n]} & & & & &\\
    &\bM^{\row}_{[e_A]}  & & & \bZ_2 &\\
    &\bM^{\row}_{[e_A+1,\be]}  & & & & \bY_2\\
    &\bM^{\row}_{[\be+1,n]}  & & & &\\
    & & \ddots & & \vdots\\
    &&&\bM^{\row}_{[e_A]}  & \bZ_L &\\
    &&&\bM^{\row}_{[e_A+1,\be]}   & & \bY_L\\
    &&&\bM^{\row}_{[\be+1,n]}   & &\\
    \end{pmatrix}.
\]
To prove that $\bB^{\SYM}$ has full column rank over $\F_q[\{Y_{j,u}\}_{j,u}]$, it suffices to identify a subset of rows of $\bB^{\SYM}$ such that restricting $\bB^{\SYM}$ to these rows results in a square matrix with nonzero determinant.

For each $h \in [L]$, let $T_h \subseteq [e_A]$ be the set of indices $t \in [e_A]$ such that $c_{j_t,i_t+h-1} - y_{j_t,i_t+h-1} \neq 0$. Note that by the definition of $P=\{i_t,j_t\}^{e_A}_{t=1}$ as in (\ref{eq:frs-adv-set}), we have that each $L$-length adversarial error location $(i,j)\in P$ changed at least one value $h\in[L]$ so $T_1 \cup \cdots \cup T_L = [e_A]$. For each $h \in [L]$, let $T'_h := T_h \setminus (T_1 \cup \cdots \cup T_{h-1})$. 

Recall that $\bar{e} \le (s-L+1)e \le \frac{L}{L+1}\left((s-L+1)n-k\right)$. Moreover, from the constraint on $e_0$ there must be  $e_A \le (s-L+1)e_0 \le (s-L+1)(n-e)-k \le (s-L+1)n-\be-k$. That is, $k+\be+e_A \le (s-L+1)n$. Thus, we can pick $n_1, \hdots, n_L \in [(s-L+1)n]$ subject to the following constraints.
\begin{align*}
n_h &\ge |T'_h| + k + \be \quad \forall h \in [L],\\
n_1 + \cdots + n_L &= Lk + (L+1)\be.
\end{align*}
Then, the matrix
\[
    \widetilde{\bB}^{\SYM} := \begin{pmatrix}
    \bM^{\row}_{[e_A]} & & & & \bZ_1 &\\
    \bM^{\row}_{[e_A+1,\be]} & & & & & \bY_1\\
    \bM^{\row}_{[\be+1,n_1]} & & & & &\\
    &\bM^{\row}_{[e_A]}  & & & \bZ_2 &\\
    &\bM^{\row}_{[e_A+1,\be]}  & & & & \bY_2\\
    &\bM^{\row}_{[\be+1,n_2]}  & & & &\\
    & & \ddots & & \vdots\\
    &&&\bM^{\row}_{[e_A]}  & \bZ_L &\\
    &&&\bM^{\row}_{[e_A+1,\be]}   & & \bY_L\\
    &&&\bM^{\row}_{[\be+1,n_L]}   & &\\
    \end{pmatrix}.
\]
is a square matrix. We seek to prove that $\det \widetilde{\bB}^{\SYM} \neq 0$.
To do this, it suffices to find a monomial in the determinant expansion of $\det \widetilde{\bB}^{\SYM}$ which appears with a nonzero coefficient. To understand  $\det \widetilde{\bB}^{\SYM}$, we expand along the last $\be$ columns of $\widetilde{\bB}^{\SYM}$. In particular, for each $i \in [\be]$, there are (up to) $L$ choices of which term to pick, based on which of the $L$ blocks we choose. To parameterize these choices, we consider a partition $S_1 \sqcup \cdots \sqcup S_L = [\be]$ such that $S_h$ corresponds to the columns picked in the $h$-th row block. After expanding along these $\be$ rightmost columns, the remaining matrix is a block-diagonal matrix where each block is of the form $\bM^{\row}_{[n_h] \setminus S_h}$ for some $h \in [L]$. If some of these matrices are not square, then the determinant of the block diagonal matrix will be zero. Otherwise, if $n_h - |S_h| = \be + k$ for all $h \in [L]$, we have that the choice of partition $S_1 \sqcup \cdots \sqcup S_L$ contributes (up to sign) the monomial
\begin{align}
    \prod_{h=1}^L \left[\det \bM^{\row}_{[n_h] \setminus S_h} \prod_{t \in S_h \cap [e_A]} (c_{j_t,i_t+h-1} - y_{j_t,i_t+h-1}) \prod_{t \in S_h \cap [e_A+1, \be]} Y_{j_t,i_t+h-1}\right].\label{eq:frs-part}
\end{align}
For each $h \in [L]$, let $m_h := n_h - \be - k$. We thus have the requirement that $|S_h| = m_h$ for all $h \in [s]$ in order for the contribution of the partition to be nonzero.

There is another necessary condition that the partition $\{S_h\}^L_{h=1}$ has to follow so that its coefficient of the contribution is non-zero. We state it as follows.
\begin{claim}\label{clm:frs-prefix}
If a partition $S_1\sqcup\cdots\sqcup S_L=[\be]$ contributes non-zero coefficient in (\ref{eq:frs-part}), \change{then} $\sum^L_{t=h}|S_t\cap [e_A+1,\be]|\leq \sum^L_{t=h}(m_t-|T'_t|)$ for all $h\in[L]$.
\end{claim}
\begin{proof}
Since $\sum^L_{t=1}|S_t\cap [e_A+1,\be]|=\be-e_A=\sum^L_{t=1}(m_t-|T'_t|)$, the statement is equivalent to the following.
\begin{equation}\label{eq:frs-reverse}
\sum^h_{t=1}|S_t\cap [e_A+1,\be]|\ge \sum^h_{t=1}(m_t-|T'_t|)\quad\forall h\in[L].
\end{equation}
\change{Therefore (\ref{eq:frs-reverse}) is equivalent to showing that}
\change{
\[
  \sum_{t=1}^h |T'_t| \le \sum^h_{t=1}\left(m_t-|S_t\cap [e_A+1,\be]|\right).
\]
}
\change{Since $|S_h|=m_h$  and $T'_1 \cup \cdots \cup T'_h = T_1 \cup \cdots \cup T_h$ for all $h \in [L]$, the above inequality is equivalent to showing that}
\change{
 \begin{align}
|T_1 \cup \cdots \cup T_h| \ge \sum^h_{t=1}|S_t\cap [e_A]|.\label{eq:key}
\end{align}
}
\change{To see why this inequality holds, note that by assumption $S_1 \cup \cdots \cup S_L$ contributes a non-zero coefficient in (\ref{eq:frs-part}),} and there are at most $|T_1\cup\cdots \cup T_h|$ non-zero columns in $[e_A]$ from the first $h$ row blocks. \change{Therefore, (\ref{eq:key}) must hold, implying (\ref{eq:frs-reverse}).} 
\end{proof}
\begin{remark}
\cref{clm:frs-prefix} is also true for the previous Interleaved RS codes case. However, Interleaved RS codes case are relatively simpler than folded RS codes case and we didn't explicitly use \cref{clm:frs-prefix} in the proof, so we omitted it in \cref{subsec:reduction_matrix}.

Motivated by the above observations, we call a partition $S_1\sqcup \cdots\sqcup S_L=[\be]$ valid iff $|S_h|=m_h$ and the condition stated in \cref{clm:frs-prefix} holds for all $h\in[L]$. Only valid partition can possibly contribute non-zero coefficient in (\ref{eq:frs-part}).

\end{remark}

To finish, it suffices to find a monomial (i.e., a product of $Y_{j,u}$'s) which appears with a nonzero coefficient for exactly one choice of valid partition $S_1, \hdots, S_L$. In \cref{subsec:reduction_matrix} the construction of such a monomial is quite straight-forward. However, since the same symbolic variable $Y_{j,u}$ may appear across different row blocks, we need more work to find such a monomial. Here we directly assume the existence of such a monomial.
\begin{lemma}\label{lem:frs-monomial}
There exists a monic monomial $M\in\mathbb{F}_q[\{Y_{j,u}\}_{j,u}]$, such that there is exactly one unique partition $U_1 \sqcup \cdots \sqcup U_L = [e_A+1,\be]$ satisfies the following two conditions simultaneously.
\begin{equation}\label{eq:frs-prefix-small}
\sum^L_{t=h}|U_t|\leq \sum^L_{t=h}(m_t-|T'_t|)\quad\forall h\in[L].
\end{equation}
\begin{equation}\label{eq:frs-evaluation}
\prod^L_{h=1}\prod_{t \in U_h} Y_{j_t,i_t+h-1}=M.
\end{equation}
Moreover, this special partition satisfies:
\begin{equation}\label{eq:frs-capacity}
|U_h|=m_h-|T'_h|\quad \forall h\in[L].
\end{equation}
\end{lemma}

We postpone the proof of \cref{lem:frs-monomial} with concrete construction to \cref{subsec:frs-monomial}. Here we first show how to finish the proof of \cref{thm:frs110} conditioned on it.

\paragraph{Proof of \cref{thm:frs110} conditioned on \cref{lem:frs-monomial}} For each 
$h\in[L]$, recall we selected each $n_h$ such that $n_h \ge |T'_h| + k + \be$. Thus, $m_h \ge |T'_h|$. Let $\ell_h := m_h - |T'_h|$. Since $n_1 + \cdots + n_L = Lk + (L+1)\be$ and $|T'_1| + \cdots +|T'_L| = e_A$, we have that $\ell_1 + \cdots + \ell_L = \be - e_A = e_R$. Let $M$ and $U_1\sqcup\cdots\sqcup U_L=[e_A+1,\be]$ be the monomial and the unique partition guaranteed by \cref{lem:frs-monomial} such that $|U_h|=\ell_h$ for all $h\in[L]$. Then, the partition $\{S_h:=T'_h\cup U_h
:h\in[L]\}$ of $[\be]$ corresponds to the monomial
\[
    M\prod_{h=1}^L \left[\det \bM^{\row}_{[n_h] \setminus S_h} \prod_{t \in T'_h} (c_{j_t,i_t+h-1} - y_{j_t,i_t+h-1})\right].
\]
Note that this monomial has nonzero coefficient since $c_{j_t,i_t+h-1} - y_{j_t,i_t+h-1} \neq 0$ for all $t \in T'_h \subseteq T_h$. We also note that by the uniqueness of $\{U_h\}^L_{h=1}$ guaranteed by \cref{lem:frs-monomial} this partition restricting to $[e_A+1,\be]$ is the only one which produces the monomial $M$ with non-zero coefficient as contribution in (\ref{eq:frs-part}). Consider a partition $S'_1 \sqcup \cdots \sqcup S'_L = [\be]$ which produces the same monomial with non-zero coefficient in the corresponding contribution term, we must have that $S'_h \cap [e_A+1, \be] = U_h$ for all $h \in [L]$. Now, since the partition $\{S'_h\}^L_{h=1}$ has non-zero coefficient in its contribution term, from the previous discussion we must have that $|S'_h| = m_h$ for all $h \in [L]$. This implies that $|S'_h \setminus U_h| = |T'_h|$ for all $h \in [L]$. Furthermore, in order for (\ref{eq:frs-part}) to be nonzero, we must have that $S'_h\setminus U_h \subseteq T_h$ for all $h \in [L]$. At this point, we can easily prove by induction that $S'_h\setminus U_h = T'_h$ for all $h \in [L]$. The base case that $S'_1\setminus U_1 = T_1=T'_1$ is immediate since the cardinalities match. For any $h \ge 2$, note that 
\[
S'_h\setminus U_h \subseteq T_h \setminus (S'_1 \cup \cdots \cup S'_{h-1}) = T_h \setminus (T'_1 \cup \cdots \cup T'_{h-1}) = T_h \setminus (T_1 \cup \cdots \cup T_{h-1}) = T'_h,
\]
so $S'_h \setminus U_h= T'_h$ holds for all $h\in[L]$ since the cardinalities match. As a corollary, there must be $S'_h=(S'_h\setminus U_h)\cup U_h=T'_h\cup U_h=S_h$ for all $h\in[L]$. Thus, the contribution regarding monomial $\prod_{h=1}^L \prod_{i \in U_h} Y_{i,h}$ in $\det{\widetilde{\bB}^{\SYM}}$ is indeed unique with a nonzero coefficient, so $\widetilde{\bB}^{\SYM}$ has nonzero determinant. Therefore, $\bB^{\SYM}$ has full column rank.

To finish, note that $\bB''$ corresponds to $\bB^{\SYM}$ except uniformly random values from $\F_q$ are substituted for each $Y_{j,u}$ with $j \in [e']$ and $u \in [s]$. Let $\widetilde{\bB}''$ be the restriction of $\bB''$ to the same rows as $\widetilde{\bB}^{\SYM}$ was restricted to relative to $\bB^{\SYM}$. Since $\det \widetilde{\bB}^{\SYM}$ has degree $e_R \le \be$, by the Schwartz-Zippel Lemma \cite{Zip79,Sch80}, we have that $\det \widetilde{\bB}'' \neq 0$ with probability at least $1 - \be/q$. Thus, $\bB''$ (and so $\bB$) has full column rank with probability at least $1 - \be/q$, as desired. This completes the proof of Theorem~\ref{thm:frs110}.
\subsubsection{Constructing the Intended Monomial: Proof of \cref{lem:frs-monomial}}\label{subsec:frs-monomial}
In this last subsection we prove \cref{lem:frs-monomial}, which fills in the final piece and completes the whole proof. \cref{lem:frs-monomial} is actually one of the main technical differences for folded RS case from the interleaved RS codes case. 

Let's define the weight function $\mathrm{wt}(i,j)=i+(s+1)j$ for $(i,j)\in[s]\times[n]$. It's obvious that $\mathrm{wt}(i,j)=\mathrm{wt}(i',j')$ iff $(i,j)=(i',j')$. Then, we can compare two paired indices according to this weight function as $(i,j)>(i',j')\Leftrightarrow \mathrm{wt}(i,j)>\mathrm{wt}(i',j')$. Then, we will use the following greedy Algorithm~\ref{alg:generate_frs} to produce a ``lexicographically largest'' monic monomial $M$ of degree $e_R$ and a corresponding partition $U_1\sqcup\cdots\sqcup U_L=[e_A+1,\be]$ that generates $M$. We will show the output of Algorithm~\ref{alg:generate_frs} satisfies (\ref{eq:frs-evaluation},\ref{eq:frs-capacity}) and it is the unique valid partition that generates $M$, which completes the proof of \cref{lem:frs-monomial}. Recall that for all $h\in[L]$ we define $\ell_h=m_h-|T'_h|$, and $\{(i_t,j_t)^{\be}_{t=e_A+1}\}=[s-L+1]\times[e']$ representing the $L$-length random error locations. In the following, for the ease of our analysis we assume $\{U_h\}^L_{h=1}$ is a partition of $[s-L+1]\times[e']$, which is equivalent to a partition of $[e_A+1,\be]$.

\begin{algorithm}[h]\label{alg:generate_frs}
    \caption{Generate Lexicographically Largest Monomial}
    \SetAlgoLined
    \KwIn{$e_A,\be,s,L,\ell_1,\dots,\ell_L$}
    \KwOut{A monomial $M$ and a partition $U_1\sqcup\cdots\sqcup U_L=[s-L+1]\times[e']$}
    {
    $M:=1$\;
    \ForEach{$h=1,\dots,L$}{
    $\mathsf{cap}_h:=\ell_h$\;
    $U_h:=\emptyset$\;
    }
    \ForEach{$j=e',\dots,1$}{
    \ForEach{$i=s-L+1,s-L,\dots,1$}{
    $h^*$:=largest $h\in[L]$ such that $\mathsf{cap}_h>0$\;
    $M:=MY_{j,i+h^*-1}$\;
    $U_{h^*}:=U_{h^*}\cup\{(i,j)\}$\;
    $\mathsf{cap}_h\leftarrow \mathsf{cap}_h-1$\;
    }
    }
    \Return{$M,U_1,\cdots, U_L$}\;
    }
\end{algorithm}
From the previous discussions $\sum_{h=1}^L\ell_h=\sum_{h=1}^Lm_h-|T'_h|=\be-e_A=e_R=e'(s-L+1)$. From the definition of Algorithm~\ref{alg:generate_frs} we know the output must satisfy (\ref{eq:frs-capacity}) and (\ref{eq:frs-prefix-small}). We can also verify the output satisfies (\ref{eq:frs-evaluation}) from the instructions of Algorithm~\ref{alg:generate_frs}. Therefore, it suffices to prove the uniqueness claim to complete the proof of \cref{lem:frs-monomial}.
\begin{claim}[Uniqueness]\label{clm:uniqueness}
$U_1\sqcup\cdots\sqcup U_L=[s-L+1]\times [e']$ is the unique partition of $[e_A+1,\be]$ that generates $M$ and satisfies (\ref{eq:frs-prefix-small}) at the same time.
\end{claim}
\begin{proof}
Recall $\{(i_t,j_t)\}^{\be}_{t=e_A+1}=[s-L+1]\times[e']$. For each $(i,j)\in[s-L+1]\times [e']$, let $d[i,j]\in[L]$ denote the unique block such that $(i,j)\in U_{d[i,j]}$. Suppose by contrapositive there is another partition $U'_1\sqcup\cdots\sqcup U'_L=[e_A+1,\be]$ that also generates $M$ and satisfies (\ref{eq:frs-prefix-small}). Similarly, we can use $d'[i,j]\in[L]$ to denote the unique block such that $(i,j)\in U'_{d'[i,j]}$. Since $\{U_h\}^L_{h=1}$ satisfies (\ref{eq:frs-capacity}) and both two partitions satisfy (\ref{eq:frs-prefix-small}), we can conclude that
\begin{equation}\label{eq:frs-prefix-capacity}
\sum^L_{t=h}|U'_t|\leq \sum^L_{t=h}|U_t|\quad \forall h\in[L].
\end{equation}
Since $U_1\sqcup\cdots\sqcup U_L$ and $U'_1\sqcup\cdots\sqcup U'_L$ are different partitions, there must be some $(i,j)\in[s-L+1]\times[e']$ such that $d[i,j]\neq d'[i,j]$.  Then, let $(i^*,j^*)\in[s-L+1]\times[e']$ be the largest-weight index where $d[i^*,j^*]\neq d'[i^*,j^*]$. It follows that for all larger indices $(i,j)>(i^*,j^*)$ there must be $d[i,j]=d'[i,j]$. Next we prove that for any $(i,j)\leq (i^*,j^*)$, there is $d'[i,j]\leq d[i^*,j^*]$. If $d[i^*,j^*]=L$, this is obvious. If $d[i^*,j^*]<L$, from 
(\ref{eq:frs-prefix-capacity}) and the mechanism of Algorithm~\ref{alg:generate_frs}, the following holds.
\begin{align*}
\sum^L_{t=d[i^*,j^*]+1}|U'_t|\leq\sum^L_{t=d[i^*,j^*]+1}|U_t|&= \sum^L_{t=d[i^*,j^*]+1}|\{(i,j):(i,j)>(i^*,j^*),d[i,j]=t,(i,j)\in[s-L+1]\times[e']\}|\\
&=\sum^L_{t=d[i^*,j^*]+1}|\{(i,j):(i,j)>(i^*,j^*),d'[i,j]=t,(i,j)\in[s-L+1]\times[e']\}|
\end{align*}

The reason for the first equality is that any $(i,j)\leq (i^*,j^*)$ must have $d[i,j]\leq d[i^*,j^*]$ because $d[i,j]$ is a non-decreasing function with respect to $<$ by Algorithm~\ref{alg:generate_frs}. Therefore, for any $t>d[i^*,j^*]$, $U'_t$ has the same size as its subset $|\{(i,j):(i,j)>(i^*,j^*),d'[i,j]=t,(i,j)\in[s-L+1]\times[e']\}|$, which implies that the above is in fact an equality.

From above, we know for all $(i,j)\leq (i^*,j^*)$ there must be $d'[i,j]\leq d[i^*,j^*]$.  Specifically, there is  $d'[i^*,j^*]<d[i^*,j^*]$ since $d'[i^*,j^*]\neq d[i^*,j^*]$. Therefore, for any $(i,j)\leq (i^*,j^*)$, recall that $\{U'_{h}\}^L_{h=1}$ generates $Y_{j,i+d'[i,j]-1}$ at index $(i,j)$. We can show that $Y_{j,i+d'[i,j]-1}\neq Y_{j^*,i^*+d[i^*,j^*]-1}$ by observing $\mathrm{wt}(i+d'[i,j]-1,j)<\mathrm{wt}(i^*+d[i^*,j^*]-1,j^*)$ as follows.
\begin{itemize}
\item[Case 1:] When $(i,j)=(i^*,j^*)$, it follows that
\[\mathrm{wt}(i+d'[i,j]-1,j)-\mathrm{wt}(i^*+d[i^*,j^*]-1,j^*)=d'[i^*,j^*]-d[i^*,j^*]<0\]
\item[Case 2:] Otherwise we have $(i,j)<(i^*,j^*)$, we observe that
\[\mathrm{wt}(i+d'[i,j]-1,j)-\mathrm{wt}(i^*+d[i^*,j^*]-1,j^*)=d'[i,j]-d[i^*,j^*]+\mathrm{wt}(i,j)-\mathrm{wt}(i^*,j^*)<0\]
\end{itemize}
This implies this $Y_{j^*,i^*+d[i^*,j^*]-1}$ generated by $\{U_h\}^L_{h=1}$ at index $(i^*,j^*)$ cannot be generated by $\{U'_h\}^L_{h=1}$ in any index $(i,j)\leq(i^*,j^*)$. Moreover, since $d[i,j]=d'[i,j]$ for all $(i,j)>(i^*,j^*)$ and $\{U_h\}^L_{h=1}$ generates exactly the same partial monomial as $\{U'_h\}^L_{h=1}$ does in these indices, it means that $\{U_h\}^L_{h=1}$ actually generates at least one more $Y_{j^*,i^*+d[i^*,j^*]-1}$ than $\{U'_h\}^L_{h=1}$ does. Therefore, $\{U'_h\}^L_{h=1}$ cannot generate $M$ which is generated by $\{U_h\}^L_{h=1}$. This is a contradiction.
\end{proof}

\section{\change{A} Decoding Algorithm for Multiplicity Codes} \label{sec:MULT}
In this section, with a similar interpolation framework of FRS codes provided in \cref{sec:FRS}, we provide a new decoding algorithm for multiplicity codes and then prove our
main results by bounding the failure probability of the proposed algorithm for semi-adversarial
errors (see \cref{thm:MULT}). The algorithm here (see Algorithm \ref{algo:Multmain}) also achieves the running time $\widetilde{O}(\poly(1/\eps)n)$. 

Throughout this section, we will fix $s,n,k>0$, a finite field $\mathbb{F}_q$, where $|\mathbb{F}_q|>n>k$, and any sequence of $n$ \change{distinct} elements $\alpha_1,\dots,\alpha_n\in\mathbb{F}_q\change{\setminus \{0\}}$. We consider the corresponding folded RS code $\mathcal{C}=\change{\mathsf{MULT}}^{s,\gamma}_{n,k,q}(\alpha_1,\alpha_2,\dots,\alpha_n)\subseteq(\mathbb{F}^s_q)^n$. For any $f\in\mathbb{F}_q[X]$ with $\deg{f}<k$, we use $\mathcal{C}(f)$ to denote its encoded codeword defined as:
\[
\mathcal{C}(f):=(\mathcal{C}^f_1,\dots,\mathcal{C}^f_n)\in\left(\mathbb{F}^s_q\right)^n,\quad\text{where } \mathcal{C}^f_i:=\Bigl(f(\alpha_i),f^{(1)}(\alpha_i),\dots,f^{(s-1)}(\alpha_i)\Bigl)\in\mathbb{F}^s_q,
\]
where we recall that $f^{(i)}$ denotes the $i$th Hasse derivative of $f$.

\subsection{A New Berlekamp–Welch Typle Interpolation}
In this subsection, much like in \cref{interpo:FRS} for folded RS codes, we present a new interpolation procedure below that succinctly captures the core insight driving our multiplicity codes algorithm. 

\paragraph{Our New Multiplicity Interpolation} A similar new lightweight interpolation rule is
defined as follows for multiplicity codes.

\begin{definition}[$L$-length multiplicity interpolation of degree $(D_1,D_2)$]\label{def:mult-inter}
Given any $D_1,D_2\ge 1$, and decoding parameter $1\leq L\leq s$, we say that a polynomial $E(X)$ of degree at most $D_1$ together with $s$ polynomials $\{A_i(X)\}_{i=1}^L$ each of degree at most $k-1+D_2$ form a solution to $L$-length multiplicity interpolation of degree $(D_1,D_2)$ if the following equations
\begin{equation}\label{eq:Int-mult}
A_{h}^{(i-1)}(\alpha_j)=\sum\limits_{\ell=0}^{i-1} \binom{h+i-2-\ell}{h-1} y_{j,i+h-1-\ell}E^{(\ell)}\left(\alpha_j\right)\quad j\in[n]\text{ and } i\in\{1,\cdots,s-L+1\}
\end{equation}
hold for $h \in \{1,2,\hdots,L\}$. If $D_1=D_2=D$, we just say they form a solution to $L$-length multiplicity interpolation of degree $D$. 
\end{definition}

\subsection{Description of Our Algorithm}
Similar to the FRS case in \cref{sec:FRS}, the following pseudo-code (See Algorithm \ref{algo:Multmain}) describes our fast multiplicity codes decoding algorithm that integrates our new ``degree-fixed'' multiplicity interpolation (\cref{def:mult-inter}) with \cite{alekhnovich2005linear} (see \cref{thm:alek}).
\paragraph{Running Time of Our Algorithm \ref{algo:Multmain}} The analysis is the same as that of Algorithm~\ref{algo:main}. There are only two differences. Fist, is that in Algorithm~\ref{algo:Multmain} we additionally need to compute the derivative of $f(X)$ for $s$ times which only requires $O(sn)$ computation. We also need to use Theorem~\ref{thm-fast-Hermite} to compute $Q_0$ in $O(M(sn)\log(sn))$. Therefore, the final running time is still $O(s^{O(1)}\mathsf{M}(sn)\log{n})\leq O(s^{O(1)}n\log^2{n}\log{\log{n}})$.
\begin{algorithm}\label{algo:Multmain}
    \caption{Near-linear time multiplicity code decoding algorithm}
    \SetAlgoLined
    \KwIn{Parameters $n,k,q,s, L, e$\;
    Evaluation points $(\alpha_1, \hdots, \alpha_n) \in \F_q^n$\;
    Received \change{word} $\vec{\bf{y}} := ((y_{1,1}, \hdots, y_{1,s}), \hdots, (y_{n,1}, \hdots, y_{n,s})) \in (\F_q^s)^n$. }
    \KwOut{A polynomials $f(X)$ of degree at most $k-1$ or $\Fail$}
    {
    $Q_{0}(X):=\prod_{j=1}^{n}\left(X-\alpha_j\right)^{s-L+1}$\;
    \ForEach{$h=1,\dots,L$}{
        $Q_h(X):= \Hermite\Bigl(X; \left\{\left(\alpha_i,  \binom{h-1}{h-1}y_{i,h},\hdots,\binom{j-1}{h-1}y_{i,j},\hdots, \binom{s-L+h-1}{h-1}y_{i,s-L+h}\right)\right\}_{i\in[n]}\Bigl)$\;
    }
    $E,C_1,\hdots, C_L := \mathsf{Interpolate}(Q_0, Q_1, \dots,Q_L)$\;
    $B(X):=Q_1(X)E(X)+Q_0(X)C_1(X)$\;
    $f(X):=B(X)/E(X)$\;
    \If{$f \not\in \F_q[X]$ \textbf{\emph{or}} $\deg{f}\ge k$}{
        \Return{$\Fail$}\;
        }
    \If{$\Delta(\vec{\mathbf{y}},\{(f(\alpha_i),f^{(1)}(\alpha_i),\hdots,f^{(s-1)}(\alpha_i))\}_{i=1}^n) > e$}{
        \Return{$\Fail$}\;
    }
    \Return{f}\;
    }
\end{algorithm}

\subsection{Analysis of Algorithm~\ref{algo:Multmain} for the Semi-Adversarial Errors}
In this section, our aim is to prove the following main theorem of our analysis on Algorithm \ref{algo:Multmain} for the semi-adversarial errors.
\begin{theorem}\label{thm-multAlgAna}
Consider integers $s,L,n,k\ge 1$ with $s<\operatorname{Char}(\F_q)$. For any multiplicity code $\mathsf{MULT}_{n,k}^{s}(\alpha_1, \dots, \alpha_n)$ over the alphabet $\F_{q}^s$, algorithm~\ref{algo:Multmain} can { uniquely decode} any message polynomial $f$ of degree at most $k-1$ from { $(e_0,e)$-semi-adversarial} errors, with probability at least $1-\frac{e(s-L+1)}{q}$, for any $e \le \frac{L}{L+1}\left(n-\frac{k}{s-L+1}-1\right)$ and $e_0\leq\min\left(e,n-e-\frac{k}{s-L+1}\right)$. The running time of algorithm~\ref{algo:Multmain} is $O(s^{O(1)}n\log^2{n}\log{\log{n}})$.
\end{theorem}

\paragraph{Solutions for $L$-length multiplicity interpolation of a fixed degree}

We will first construct solutions to our interpolation family. Let us assume $f$ was the message polynomial. Let the index set $I_R$ contain all random errors and the index set $I_A$ contain all the adversarial errors. We note $|I_A|\leq e_0$ and $|I_R|\leq e-e_0$. We can adapt $e,e_0$ such that $|I_A|=e_0$ and $|I_R|=e-e_0$. Concretely, if adversarial errors are fewer than expected we arbitrarily increase $I_A$ to be of size $e_0$. If random errors are fewer than expected we can decrease $e$ and it won't break our requirement on $e,e_0$. Moreover, we will finally prove that with high probability the output of algorithm~\ref{algo:Multmain} is exactly $f$. Therefore, such an adaption will not influence the behavior of algorithm~\ref{algo:Multmain} when it indeed finds $f$.

The adversarial errors can contain coordinates where no errors have been made. For each $i\in I_A$ we let $w'_i=\min\{j| f^{(j-1)}(\alpha_i)\ne y_{i,j}\text{ or } j=s+1\}$ and $w_i=\max(w'_i-L,0)$. $w'_i$ tells us the first derivative in which no error has occurred and $w_i$ tells us how many derivatives over $L-1$ have no error. We let $d_i=s-L+1$ for $i\in I_R$ and $d_i=s-L+1-w_i$ for $i\in I_A$. As before, we can assume $I_A=\{1,\dots,e_0\}$ and $I_R=\{e_0+1,\dots,e\}$ without loss of generality.

Consider the error locator 
\begin{equation}\label{eq-MultErrInt}
E(X):=\prod_{j=1}^{e}\left(X-\alpha_{j}\right)^{d_i}.
\end{equation}
$E(X)$ has degree $D=\sum_{i=1}^{e} d_i = e(s-L+1)-\sum\limits_{i\in I_A} w_i$. 

We define $L$ polynomials $A_{1}(x),A_{2}(x),\cdots,A_{L}(x)$ with degree at most $k-1+D$ as 
\begin{equation}\label{eq-MultInterSol}
A_{h}\left(X\right)=f^{(h-1)}\left(X\right)E\left(X\right) \quad h\in\{1,2,\cdots,L\}
\end{equation}

\begin{lemma}\label{lem-multConstSol}
$P E(X)$ and $(PA_1(X),\hdots,PA_L(X))$ where $E$ and $A_i,i\in \{1,\hdots,L\}$ is as defined in \eqref{eq-MultErrInt} and \eqref{eq-MultInterSol} respectively, and $P$ is any polynomial of degree at most $e (s-L+1)-D=\sum_{j\in I_A} w_j$ is a solution to the $L$-length Multiplicity interpolation of degree $e(s-L+1)$.

In other words, $X^i E(X)$ and $(X^iA_1(X),\hdots,X^iA_L(X))$ where $i\in \{0,\hdots,e (s-L+1)-D\}$ span a dimension $e (s-L+1)-D+1$ space of solutions for the $L$-length multiplicity interpolation of degree $e(s-L+1)$. 
\end{lemma}
\begin{proof}
    For any coordinate outside of $I_A$ and $I_R$, we know that $y_{j,i+1}=f^{(i)}(\alpha_j)$. The equations for the multiplicity interpolation for these coordinates are shown to be satisfied then using 2 and 3 in Lemma~\ref{lem:HasseProp}.
    \begin{align*}\sum\limits_{\ell=0}^{i-1} \binom{h+i-2-\ell}{h-1} y_{j,i+h-1-\ell}E^{(\ell)}\left(\alpha_j\right)&=\sum\limits_{\ell=0}^{i-1} \binom{h+i-2-\ell}{h-1} f^{(i+h-2-\ell)}(\alpha_j)E^{(\ell)}\left(\alpha_j\right)\\
    &=\sum\limits_{\ell=0}^{i-1} (f^{(h-1)})^{(i-1-\ell)}(\alpha_j) E^{(\ell)}\left(\alpha_j\right)\\
    &= (f^{(h-1)}E)^{(i-1)}(\alpha_j)=A_h^{(i-1)}(\alpha_j).
    \end{align*}
    
    For $j\in I_R$, we see that $E^{(i)}(\alpha_j)=0$ for all $i\in \{0,\hdots,s-L\}$. We see that both sides of interpolation equations \eqref{eq:Int-mult} are $0$ in this case.

    For $j\in I_A$, if $w_j=0$ then the argument is the same as for coordinates in $I_R$. If $w_j>0$, we see that $E^{(i)}(\alpha_j)=0$ for all $i\in \{0,\hdots,s-L-w_j\}$. We also have $y_{j,i}=f^{(i-1)}(\alpha_j)$ for $i\in \{1,\hdots,w'_j-1\}$.
    We note then,
    \begin{align*}\sum_{\ell=0}^{i-1} \binom{h+i-2-\ell}{h-1} 
 y_{j,i+h-1-\ell}  E^{(\ell)}(\alpha_j)&= \sum_{\ell=s-L+1-w_j}^{i-1}  \binom{h+i-2-\ell}{h-1} f^{(i+h-2-\ell)}(\alpha_j) E^{(\ell)}(\alpha_j),\\
 &= \sum_{\ell=0}^{i-1}  \binom{h+i-2-\ell}{h-1} f^{(i+h-2-\ell)}(\alpha_j) E^{(\ell)}(\alpha_j)= A_h^{(i-1)}(\alpha_j),
 \end{align*}
    as $\ell \ge s-L+1 - w_j $ implies $i-1+h-\ell \le s-\ell \le w_j+L-1=w'_j-1$.

    The same argument holds for the multiple of these solutions.
\end{proof}

\paragraph{Equivalency of solutions} Similar to \cref{lem:frsequivalency}, we then show that the solutions to $L$-length multiplicity interpolation actually yield a correspondence to the solutions found in Algorithm~\ref{algo:Multmain}. 
\begin{lemma}\label{lem:MULTequivalency}
Assume $k>1$. Each non-zero solution $A_1(X),\dots, A_L(X), E(X)$ to $L$-length multiplicity interpolation of degree $D$ satisfies
\begin{compactitem}
\item[(1)] There exists unique $C'_1(X),\dots,C'_L(X)$ such that $A_h(X)=Q_h(X)E(X)+Q_0(X)C'_h(X)$ holds for all $h\in[L]$.

\item[(2)] $0<\max\left(A_1(X),\dots,A_L(X),X^{k-1}E(X)\right)\leq k-1+D$.
\end{compactitem}
The backward direction also holds. Namely, any possible solution $B_1(X),\dots,B_L(X),E(X)$ of Algorithm~\ref{algo:Multmain} satisfying $0<\max\left(B_1(X),\dots,B_L(X),X^{k-1}E(X)\right)\leq k-1+D$ is a non-zero solution to $L$-length multiplicity interpolation of degree $D$.
\end{lemma}
\begin{proof}
We first prove the forward direction: Fix any solution $A_1(X),\dots,A_L(X),E(X)$ to $L$-length multiplicity interpolation of degree $D$. Using Lemma~\ref{lem:HasseProp}, for any $h\in[L], j \in [s-L+1]$, $X=\alpha_1,\hdots,\alpha_n$ will be a root of
$$(A_h-Q_hE)^{(j-1)}(X).$$
This means $A_h(X)-Q_h(X)E(X)$ is divisible by $Q_0$. This means there exist a unique $G_h\in \F_q[X]$ such that $A_h(X)=Q_h(X)E(X)+Q_0(X)G_h(X).$ If we set $C'_h=G_h$ we see that both condition (1) and (2) in our lemma are satisfied.

For the backward direction, for any possible solution $\big\{B_h(X)=Q_h(X)E(X)+Q_0(X)C_h(X)\big\}_{h\in[L]}$,  $E(X)$ of Algorithm~\ref{algo:Multmain} satisfying $0<\max\left(B_1(X),\dots,B_L(X),X^{k-1}E(X)\right)\leq k-1+D$, we want to show that
\[
B_h^{(i-1)}(\alpha_j)=\sum\limits_{\ell=0}^{i-1}\binom{h+i-2-\ell}{h-1}y_{j,i+h-1-\ell}E^{(\ell)}(\alpha_j)
\]
holds for any $h\in[L],i\in[s-L+1],j\in[n]$. This will imply that $B_1(X),\dots,B_L(X),E(X)$ is a non-zero solution of $L$-length multiplicity interpolation of degree $D$.

As $Q_0^{(i-1)}(\alpha_j)=0$ for all $j\in [n]$ and $i\in [s-L+1]$, we have 
$$B_h^{(i-1)}(\alpha_j)= \sum\limits_{\ell=0}^{i-1} Q_h^{i-1-\ell}(\alpha_j) E^{(\ell)}(\alpha_j).$$
As $Q_h^{(i-1)}(\alpha_j)=\binom{i+h-2}{h-1}y_{j,i+h-1}$ we are done.
\end{proof}

\subsubsection{Reducing \cref{thm-multAlgAna} to a Matrix Rank Condition}

We first define the following matrices which will be useful in describing the interpolation equations. We denote $\vec{\alpha}:= (\alpha_1,\dots,\alpha_n)$.
\begin{align*}
\bM&:= \Eval_{[0,k-1+e(s-L+1)]}(\vec{\alpha},(s-L+1,\hdots,s-L+1))\in \mathbb{F}_q^{(s-L+1)n\times (k+e(s-L+1))},\\
\bL_h&:= \BlockDiag\left(\Lower\left(\binom{i+h-2}{h-1}y_{j,i+h-1}\right)_{i=1}^{s-L+1}\right)_{j=1}^n \in \F_q^{n(s-L+1)\times n(s-L+1)} \\
\bN &:= \Eval_{[1,D]}(\vec{\alpha}, (s-L+1,\hdots,s-L+1)) \in \F_q^{ n(s-L+1)\times D}\\
\bN_h&:= \bL_h \bN \in \F_q^{n(s-L+1)\times D}
\end{align*}

Similar to the previous two results we consider the matrix,
\[
\bB:=\begin{pmatrix}
\bM & & & & -\bN_1\\
 & \bM & & & -\bN_2\\
 & & \ddots & &\vdots\\
 & & & \bM & -\bN_L
\end{pmatrix},
\]
it has $L(k+e(s-L+1))+D$ columns and $L(s-L+1)n$ rows.

\begin{theorem}\label{thm:mult-non-singular}
   The matrix $\bB$ \change{has full column rank} over $\F_q$ for any set of non-zero distinct evaluation points with probability at least $1-\frac{(e-e_0)(s-L+1)}{q}$.
\end{theorem}

We postpone the proof of \cref{thm:mult-non-singular} to \cref{sec:mult-bound}. As before, we first show that whenever $\bB$ \change{has full column rank,} the solutions described in Lemma~\ref{lem-multConstSol} are in fact all valid solutions to the $L$-length multiplicity interpolation.

\begin{corollary}\label{cor:multKer}
     The constructed solutions in Lemma~\ref{lem-multConstSol} are in fact all valid solutions to the $L$-length multiplicity interpolation with probability at least $1-\frac{(e-e_0)(s-L+1)}{q}$.
\end{corollary}
\begin{proof}
Recall the $L$-length multiplicity interpolation in \cref{def:mult-inter} of degree $(D,e(s-L+1))$, for any $h\in[L]$, we define the coefficients of these polynomials as
\[A_h(X):=\sum_{i=0}^{k-1+e(s-L+1)}a^{(h)}_iX^i.\]
We restrict $E$ to be of degree at most $D$ and represent its coefficients as
\[E(X):=b_0+\sum_{i=1}^{D}b_iX^i.\]
We define vectors $\vec{b}\in\F_q^{D+1},\vec{a}_h\in\F_q^{k+e(s-L+1)},\vec{y'}_h\in\F_q^{(s-L+1)n}$ by
\begin{align*}
\vec{b} &:=(b_1,\dots, b_{D}, b_0)^T,\\
\vec{a}_h &:=(a^{(h)}_0,a^{(h)}_1,\dots,a^{(h)}_{k-1+e(s-L+1)})^T\text{, and}\\
\vec{y'}_h &:=\left(\left\{\binom{i+h-2}{h-1}y_{1,i+h-1}\right\}^{s-L+1}_{i=1},\dots,\left\{\binom{i+h-2}{h-1}y_{n,i+h-1}\right\}^{s-L+1}_{i=1}\right)^T.
\end{align*}
Meanwhile, we define the matrix ${\bf Y}\in\F_q^{(s-L+1)Ln\times 1}$ by
\[{\bf Y}:=(\vec{y'}_1,\vec{y'}_2,\hdots,\vec{y'}_L),\]
and the solution vector $\vec{\bf s}\in\F_q^{L(k+\be)+\be+1}$ by
\[\vec{\bf s}:=(\vec{a}_1,\vec{a}_2,\hdots,\vec{a}_L,\vec{b}).\]

We see that the solutions to
\[
({\bB}|{\bf Y})\cdot \vec{{\bf s}}=0
\]
correspond exactly to the solutions of the interpolation of the multiplicity of length $L$ of degree $(D,e(s-L+1))$. If the block matrix $\bB$ has full column rank, then the original interpolation matrix $({\bB}|{\bf Y})$ \change{has rank at least the number of columns of $\bB$}. \change{Therefore, the kernel of $({\bB}|{\bf Y})$ has dimension at most the number of columns of $\bf Y$ which is $1$}.

\change{Let $d$ be the dimension of the solution space of the $L$-length multiplicity interpolation of degree $e(s-L+1)$. For every polynomial solution $(A_1, \hdots, A_h, E)$ in this space, we have that $\deg E \le e(s-L+1)$. As such, the the space of solutions with $\deg E \le D$ is at least $d - [e(s-L+1) - D]$. However, by the above argument, this quantity is at most $1$. Therefore,
\[
  d - [e(s-L+1) - D] \le 1 \implies d \le e(s-L+1)-D+1.
\]
}
\change{Hence,} the $L$-length multiplicity interpolation of degree $e(s-L+1)$ has a solution space of dimension at most $e(s-L+1)-D+1$. As the solutions generated in Lemma~\ref{lem-multConstSol} have exactly this dimension we are done.
\end{proof}

We see the above along with Lemma~\ref{lem:MULTequivalency} imply Theorem~\ref{thm-multAlgAna}. To complete the argument we now proceed to prove Theorem~\ref{thm:mult-non-singular}.

\subsubsection{Bounding Matrix Rank: Proof of Theorem~\ref{thm:mult-non-singular}}\label{sec:mult-bound}
Recall, we let $m_r(x)=x^r$. We may assume the random errors are on $I_R=\{e_0+1,\hdots,e\}$ and the adversarial errors are on $I_A=\{1,\hdots,e_0\}$ for the original message polynomial $f(x)=\sum_{t=0}^{k-1}f_tx^t$. We note that for any $h\in[L],r\in\{1,\dots,D\}$ adding $\binom{t}{h-1}f_t$ times the $t+1+r-(h-1)$-th column of $\bM$ for all $t\in\{0,1,\dots,s-1\}$ gives us a column which equals $(m_rf^{(h-1)})^{(i)}(\alpha_j)$ for all $j\in [n],i\in \{0,\hdots,s-L\}$.

By the definition, we know for any $j\not\in I_R\cup I_A, i\in[s]$, the entry in the received word $y_{j,i}=f^{(i-1)}(\alpha_j)$. Using the observation in the previous paragraph we can use column operations to clear the bottom $(n-e)(s-L+1)$ rows of $-\bN_h$ which will become all zero after these operations. The rank will be unchanged, and $\bB$ will be transformed to
\[\bB':=\begin{pmatrix}
\bM & & & & -\bN'_1\\
 & \bM & & & -\bN'_2\\
 & & \ddots & &\vdots\\
 & & & \bM & -\bN'_L
\end{pmatrix},\]
where 

\begin{align*}
\bL'_h&:=\BlockDiag\left(\Lower\left(\binom{i+h-2}{h-1}z_{j,i+h-1}\right)_{i=1}^{s-L+1}\right)_{j=1}^n \\
\bN'_h&:= \bL'_h \bN
\end{align*}
and
\[z_{j,i}:=-y_{j,i}+f^{(i-1)}(\alpha_j).\]
As $z_{j,i}=0$ for $j\ge e+1$ we see that the bottom $(n-e)(s-L+1)$ rows of $\bN'_h$ are $0$.
For each $j\in I_R, i\in\{1,\dots,s\}$, consider the formal variable $Y_{i,j}$. Let $\bB''$ be the matrix over $\mathbb{F}_q[\{Y_{j,i}\}_{j\in I_R,i\in[s]\}}]$ where each term $z_{j,i}$ for $j\in I_R$ of $\bB'$ is replaced with $Y_{j,i}$ for $j\in I_R$ (we similarly define $\bN''_h$). Note that $\{z_{j,i}\}_{j,i}$ are distributed uniformly if and only if $\{Y_{j,i}\}_{j,i}$ are distributed uniformly. The value $z_{j,i}$ for $j\in I_A$ are exactly the adversarial errors which may or may not be zero.

The goal in the remainder of the proof is to pick a square submatrix of $\bB''$ and find a unique monomial in its determinant expansion with non-zero coefficient. Applying Schwartz-Zipple for this sub-determinant will give us our final result.

We define $R_i,B_i,U_i$ for $i\in [L]$ here $R_i+B_i$ represent the number of columns we will pick from $\bN'_i$. $R_i$ will be coming from the `random part' and $B_i$ will be from the `adversarial part'. $U_i$ are the extra rows that we will remove from $\bN'_i$.

$B_L$ will be picked to the largest possible and we would reverse inductively pick, $B_{L-1},\hdots,B_1$. We let 
$$d_{j,h}=\max\left(\left\{i| z_{j,s-L+1+h-i}\ne 0, i\in [s-L+1]\right\},0\right)$$
for, $j\in I_A,h\in [L]$. We let $b_{j,L}=d_{j,L}$ for $j\in I_A$. We then inductively define,
$$b_{j,L-i}=\max\left(d_{j,L-i}-\sum\limits_{\ell=0}^{i-1} b_{j,L-\ell},0\right),$$
for $i\in [L-1]$. Finally, we set $B_i=\sum_{j\in I_A} b_{j,i}, i\in [L]$. We note $b_{j,1}+\hdots+b_{j,L}=\max(d_{j,1},\hdots,d_{j,L})=s-L+1-w_j$ and $B_1+\hdots+B_L=e_0(s-L+1)-\sum_{j\in I_A} w_j$.

Next we want to ensure $R_i+U_i=(n-e)(s-L+1)-k-B_i$. As $|I_A|=e_0\le n-e-k/(s-L+1)$ we have $(n-e)(s-L+1)-k-B_i\ge 0$. We also must have $\sum_{i=1}^L U_i=L(s-L+1)(n-e)-Lk-D$ and $\max(0,-k+(n-2e)(s-L+1))\le U_i\le (n-e)(s-L+1)-k-B_i$ (the need for the lower bound will be clear later). The constraint on $U_i$, implies $R_1+\hdots+R_L=(e-e_0)(s-L+1)$.

\begin{claim}\label{clm-chooseRi}
It is possible to choose $R_i,i\in [L]$ such that $R_i$ is divisible by $s-L+1$.
\end{claim}
\begin{proof}
We take two cases, if $-k+(n-2e)(s-L+1)\ge 0$ then we set $R_1=(e-e_0)(s-L+1)$ and $R_i=0$ for $L\ge i>1$. As $(n-e)(s-L+1)-k-B_1-R_1\ge (n-2e)(s-L+1)-k-\sum_{i=1}^L B_i+e_0(s-L+1)\ge -k+(n-2e)(s-L+1)$ we can solve for $U_i$.

If $-k+(n-2e)(s-L+1)<0$, we note $U_i\in [0, (n-e)(s-L+1)-k-B_i]$.
We set  $R_i=(s-L+1)c_i,i\in [L]$ so that $c_1+\hdots+c_{L}=e-e_0$. We have the constraint that $c_i\le \lfloor n-e-(k+B_i)/(s-L+1)\rfloor $.
We first note,
\begin{align*}\label{eq-Uibound}\sum\limits_{i=1}^L \frac{U_i}{L(s-L+1)} &= (n-e)-\frac{(Lk+D)}{L(s-L+1)}\\
&\ge n-\frac{k}{s-L+1} -\frac{L+1}{L}e\\
&\ge 1.\numberthis
\end{align*}
In the last step we used the fact that $e\le \frac{L}{L+1}(n-\frac{k}{s-L+1}-1)$. We next note,
\begin{align*}
\sum\limits_{i=1}^L(s-L+1)\left\lfloor n-e-\frac{(k+B_i)}{s-L+1}\right\rfloor &\ge L(s-L+1)(n-e)-Lk-\sum\limits_{i=1}^L B_i-L(s-L+1)\\
&\ge L(s-L+1)(n-e)-Lk-\sum\limits_{i=1}^L B_i - \sum\limits_{i=1}^L U_i\\
&= \sum\limits_{i=1}^L R_i= (e-e_0)(s-L+1).
\end{align*}
In the first step above we used \eqref{eq-Uibound}. Given the above inequality we see that we can solve for $c_1,\hdots,c_L$ under the given constraints proving the claim.
\end{proof}

For each $R_i$ we let $R_i=\sum\limits_{j\in I_R} r_{j,i}$. $r_{j,i}$ correspond to the columns picked from $j\in I_R$. We need to have $r_{j,1}+\hdots+r_{j,L}=s-L+1$. As each $R_i$ is divisible by $s-L+1$, we let $r_{e_0+1,1}=\hdots=r_{e_0+R_1/(s-L+1),1}=r_{e_0+R_1/(s-L+1)+1,2}=\hdots=r_{e_0+(R_1+R_2)/(s-L+1),2}=\hdots=r_{e_0+(R_1+\hdots+R_{L-1})/(s-L+1)+1,L}=\hdots=r_{e_0+(R_1+\hdots+R_{L})/(s-L+1),L}=s-L+1$ and the remaining $r_{j,i}$ we set to $0$.

For each $h\in[L]$, we remove arbitrary $U_h$ rows from the last $(n-e)(s-L+1)$ rows in the $h$th block of rows in $\bB''$ (in other words those rows which contain the last $(n-e)(s-L+1)$ rows of $\bN''_h$) to obtain the matrix $\bB_0$. 

$\bB_0$ is a square matrix of size $L(k+e(s-L+1))+D$. If $\det {\bB_0}\neq 0$, then we have that $\det {\bB_0}$ is a polynomial with total degree at most $(e-e_0)(s-L+1)$. By Schwartz-Zippel lemma, we know by independently and uniformly sampling the assignments of $\{Y_{j,i}\}$ over $\mathbb{F}_q$, the evaluation of $\det {\bB_0}$ will be non-zero with probability at least $1-\frac{(e-e_0)(s-L+1)}{q}$, which means $\bB$ will have full column rank. 

To finish the proof we just need to show that $\bB_0$ has non-zero determinant.

Since $\alpha_j$ are non-zero distinct elements, the following matrix 
$$\bP:=\Eval_{[1,D]}((\alpha_1,\hdots,\alpha_n),(s-L+1-w_1,\hdots,s-L+1-w_{e_0},s-L+1,\hdots,s-L+1))\in \F_q^{D\times D}$$ is non-singular over $\mathbb{F}_q$ (as no non-zero degree at most degree $D$ polynomial can have at most $D$ roots counted with multiplicity).

We construct $\bB^\SYM$ as follows. Note that the rank of $\bB^\SYM$ is the same as $\bB_0$.
\[
\bB^\SYM:=\bB_0
\begin{pmatrix}
\bI_{L(k+e(s-L+1))}& \\
 & \bP^{-1}
\end{pmatrix}.
\]
For convenience, we set $(e-e_0)(s-L+1)=e_R$ and $e_0(s-L+1)=e_A$.
Then, the matrix $\bB^\SYM$ will have the following form
\[
\bB^\SYM=
\begin{pmatrix}
\bM^{\row}_{[e_A]}& & & & \bZ_1 & \textbf{0}\\
\bM^{\row}_{[e_A+1,e_R+e_A]}& & & &\textbf{0} & \bY_1\\
\bM^{\row}_{[e_R+e_A+1,n(s-L+1)-U_1]}& & & & \textbf{0} & \textbf{0} \\
 & \bM^{\row}_{[e_A]}& & & \bZ_2 & \textbf{0}\\
 &\bM^{\row}_{[e_A+1,e_R+e_A]}& & & \textbf{0} & \bY_2\\
  & \bM^{\row}_{[e_R+e_A+1,n(s-L+1)-U_2]}& & & \textbf{0} & \textbf{0}\\
 & & \ddots & &\vdots\\
 & & & \bM^{\row}_{[e_A]} & \bZ_L &\textbf{0}\\
 & & & \bM^{\row}_{[e_A+1,e_R+e_A]} & \textbf{0} & \bY_L\\
 & & & \bM^{\row}_{[e_R+e_A+1,n(s-L+1)-U_L]} & \textbf{0} & \textbf{0}
\end{pmatrix}
\]
where
\[
\bY_h= \BlockDiag\left(\Lower\left(\binom{i+h-2}{h-1}Y_{j,i+h-1}\right)_{i=1}^{s-L+1}\right)_{j=e_0+1}^e,
\]

\[\bZ_h=\BlockDiag\left(\Lower^{\col}_{[s-L+1-w_{j}]}\left(\binom{i+h-2}{h-1}z_{j,i+h-1}\right)_{i=1}^{s-L+1}\right)_{j=1}^{e_0}.\]
We crucially used the fact that $z_{j,i}=0$ for $j\in I_A$ and $i\le w'_j-1$ as that implies the last $w_j$ columns of $\Lower(z_{j,i+h-1})_{i=1}^{s-L+1}$ are $0$ for all $h\in [L]$.
By definition,
$$\det \bB^\SYM=\sum\limits_{\sigma \in \textsf{Perm}(D+L(k+e(s-L+1)))} \textsf{sign}(\sigma)\prod\limits_{i=1}^{D+L(k+e(s-L+1)}\bB^\SYM(i,\sigma(i)),$$
where $\textsf{Perm}(n)$ is the set of all permutations over $[n]$. 
We will show one of these terms will give us a unique monomial with non-zero coefficients. $(i,\sigma(i))$ gives us a set of elements in $\bB^\SYM$ such that each row and column has exactly one element picked. For convenience we use,

\begin{equation*}
\bM_h:=\begin{pmatrix}
\bM^{\row}_{e_A}\\
\bM^\row_{e_A+1,e_R+e_A}\\
\bM^\row_{e_A+e_R+1,n(s-L+1)-U_h}
\end{pmatrix},
\end{equation*}
to denote the $h$-th copy of  sub-matrix $\bM$ in $\bB^\SYM$.

\begin{claim}\label{clm:count_term_mult}
For each non-zero $\prod_{i=1}^{D+L(k+e(s-L+1))}\bB^\SYM(i,\sigma(i))$, it must contain exactly $(n-e)(s-L+1)-k-U_h$ terms from $\bY_h$ and $\bZ_h$ for every $h\in[L]$.
\end{claim}
\begin{proof}
The sub-matrix $\bM^{\row}_{[e_R+e_A+1,n(s-L+1)-U_h]}$ in $\bB^\SYM$ has $(n-e)(s-L+1)-U_h$ rows and the rest of the rows in which it lies are zero. After selecting terms in $\bM^{\row}_{[e_R+e_A+1,n(s-L+1)-U_L]}$ and removing the selected columns, the remaining sub-matrix in $\bM_{h}$ will still have $k+(2e-n)(s-L+1)+U_h\ge 0$ columns (this where we need the lower bound on $U_h$). This will lead to the choice of the same number of rows from $\bM^{\row}_{e_A}$ and $\bM^\row_{e_A+1,e_R+e_A}$ in total. There still will be $(n-e)(s-L+1)-k-U_h$ rows in excess and terms for which will have to be chosen from $\bY_h$ and $\bZ_h$. 
\end{proof}
For convenience we let,
$$\bZ_{j,h}=\Lower^{\col}_{[s-L+1-w_{j}]}(z_{j,i+h-1})_{i=1}^{s-L+1}$$ for $j\in \{1,\hdots,e_0\}$ and 
$$\bY_{j,h}=\Lower\left(\binom{i+h-2}{h-1}Y_{j,i+h-1}\right)_{i=1}^{s-L+1}$$ for $e_0+1\le j \le e$.

Recall that, $B_h+R_h=(n-e)(s-L+1)-k-U_h$. In the end, we will have that $R_h$ rows will be chosen from $\bY_h$ and $B_h$ rows will be chosen from $\bZ_h$. We first state a claim that any non-zero expansion has to follow. This constraint is similar to (\ref{eq:frs-prefix-small}). The following claim is a counterpart of \cref{clm:frs-prefix} for the multiplicity code.

\begin{claim}\label{cl-advrows-prefix}
For each non-zero $\prod_{i=1}^{D+L(k+e(s-L+1))}\bB^\SYM(i,\sigma(i))$, the following holds.
\begin{itemize}
\item For all $h\in[L]$, we must choose at most $\sum^h_{i=1}R_i$ elements in $\bY_1,\dots,\bY_h$. 
\item For all $h\in[L]$, we must choose at least $\sum^h_{i=1}R_{L-h+i}$ elements in $\bY_{L-h+1},\dots,\bY_L$
\end{itemize}

\end{claim}
\begin{proof}
As $R_1+\hdots+R_L=(e-e_0)(s-L+1)$ and a determinant expansion must pick one element from each column, we see the two properties in this claim are equivalent. We will now prove the second one.

Fix any $h\in[L]$, we note that the matrix 
$$\begin{pmatrix}
\bZ_{j,L-h+1}\\
\bZ_{j,L-h+2}\\
\vdots\\
\bZ_{j,L}
\end{pmatrix}$$
is zero outside of its first $\sum\limits_{i=1}^h b_{j,L-i+1}\ge \max(d_{j,L-h+1},\hdots,d_{j,L})$ columns. This means we can pick at most $\sum\limits_{i=1}^h b_{j,L-i+1}$ many elements from $\bZ_{j,L-h+1},\hdots,\bZ_{j,L}$. Summing over $j\in I_A$ implies that we can pick at most $B_{L-h+1}+\hdots+B_{L}$ elements from $\bZ_{L-h+1},\hdots,\bZ_L$.

Since we have shown that we have to pick exactly $\sum^h_{i=1}(R_{L-h+i}+B_{L-h+i})$ elements in $\bZ_{L-h+1},\hdots,\bZ_L$ and $\bY_{L-h+1},\hdots,\bY_L$ in total, it follows that we must choose at least $\sum^h_{i=1}R_{L-h+i}$ elements in $\bY_{L-h+1},\dots,\bY_L$. 

\end{proof}

We let $\mathcal M$ be the set of monomials in $\det \bB^\SYM$ which can be obtained under the constraints of Claim~\ref{cl-advrows-prefix}. 

We will give a procedure to construct a monomial in $\mathcal M$ which is obtained by a unique choice of entries in $\bY_h$, takes exactly $R_h$ terms in $\bY_h$, and has a special property. We note we need $s<\operatorname{Char}(\F_q)$ so that the binomial coefficients $\binom{i+h-2}{h-1}$ are non-zero when $s\ge i+h-2\ge h-1$.

\begin{claim}\label{clm-multMono}
Conditioned on \cref{cl-advrows-prefix}, there is a monomial in $\mathcal M$ which is constructed by a unique selection of elements in $\bY_h$ such that
\begin{itemize}
\item Exactly $R_h$ rows are selected in $\bY_h$,
\item the rows used in $\bY_{j,h}$ for $j\in \{e_0+1,\hdots,e\}$ is either empty or a contiguous block of rows starting from the bottom of $\bY_{j,h}$.
\end{itemize}

\end{claim}
\begin{proof}
We will prove the monomial $\prod\limits_{j\in I_R,i\in [L]}Y_{j,i}^{r_{j,i}}$ satisfies both the properties.
We first note that $T_i=\prod\limits_{j\in I_R}Y_{j,i}^{r_{j,i}}$ has degree $R_i$. Using the values of $r_{j,i}$ we choose we have, 
$$T_i=Y_{e_0+(R_1+\hdots+R_{i-1})/(s-L+1)+1,i}^{s-L+1}\hdots Y_{e_0+(R_1+\hdots+R_i)/(s-L+1),i}^{s-L+1}.$$ 
We see $T_1$ is of degree $R_1$ and as $Y_{j,1}$ terms only appear in the top-most block we have to pick all the terms corresponding to $T_1$. By Claim~\ref{cl-advrows-prefix}, we cannot pick more than $R_1$ terms from $\bY_1$ and as $T_1$ uses $R_1$ terms no more terms can be picked.

Iteratively, $T_i$ uses $R_i$ terms and by this point we are forced to have picked $R_1+\hdots+R_{i-1}$ terms from $\bY_1,\hdots,\bY_{i-1}$ which means we can pick at most $R_i$ terms from $\bY_i$ and no more terms from $\bY_1,\hdots,\bY_{i-1}$. As $Y_{j,i}$ only appears in $\bY_1,\hdots,\bY_i$ we are now forced to pick all the terms for $T_i$ from $\bY_i$.

We see that we either pick a complete block of rows from $\bY_{j,h}$ or an empty block so the second property in the claim is immediately satisfied.
\end{proof}

Let $m_C$ be the monomial in Claim~\ref{clm-multMono}. We need one final claim about the $B_h$ rows which will be chosen from $\bZ_h$.

\begin{claim}\label{cl-advrows}
There is only one way of choosing elements with a non-zero product while picking $B_h$ rows in $\bZ_h$ when calculating the coefficient of $m_C$ and the rows chosen from $\bZ_{j,h}$ are either empty or a contiguous block of rows starting from the bottom of $\bZ_{j,h}$.
\end{claim}
\begin{proof}

We prove this by noting that the square sub-matrix of $\bZ_{j,L}$ using the bottom $b_{j,L}$ rows and the first $b_{j,L}$ columns is lower triangular with a non-zero diagonal and all terms outside of this sub-matrix in $\bZ_{j,L}$ are $0$. As $\sum_{j\in I_A} b_{j,L}=B_L$, for each $\bZ_{j,L}$ we must choose elements from the non-zero diagonal elements from the bottom $b_{j,L}$ rows and the first $b_{j,L}$ columns.

Iteratively, after arguing unique choice for $\bZ_L,\hdots,\bZ_{L-h}$, the first $b_{j,L}+\hdots+b_{j,L-h}$ columns are removed from $\bZ_{j,L-h-1}$ to get $\bZ'_{j,L-h-1}$. Now, the square-sub matrix of $\bZ'_{j,L-h-1}$ with the first $b_{j,L-h-1}$ columns and the bottom $b_{j,L-h-1}$ rows is lower triangular with non-zero diagonal and outside of it $\bZ'_{j,t+1}$ is zero. As $\sum_{j\in I_A} b_{j,L-h-1}=B_{L-h-1}$, we must pick these diagonal elements.
\end{proof}

For any $h\in[L]$, let $\bM'_h$ denote the truncated $\bM_h$ after removing all its rows selected in block $\bY_h$ and $\bZ_h$ given by the unique choices in Claim~\ref{clm-multMono} and Claim~\ref{cl-advrows}. By construction, $\bM'_h$ is a square matrix of size $k+e(s-L+1)\times k+e(s-L+1)$. We recall $\bM_h$ is composed of $n$ blocks of rows corresponding to $\Eval_{[0,k-1+e(s-L+1)]}(\alpha_i,s-L+1)$. As the rows removed in each of these blocks is contiguous with the bottom, $\bM'_h$ will equal $\Eval_{0,k-1+e(s-L+1)}((\alpha_1,\hdots,\alpha_n), (u_{1,h},\hdots,u_{n,h}))$ such that the $u_{1,h}+\hdots+u_{n,h}=k+e(s-L+1)$.
This guarantees that the determinant of $\bM'_h$ is non-zero as a polynomial of degree at most $k-1+e(s-L+1)$ cannot have $k+e(s-L+1)$ roots when counted with multiplicity.
Therefore, the coefficient of $m_C$ is non-zero as, up to signs, it is a product of the determinants of $\bM'_h$ and the non-zero terms picked in Claim~\ref{cl-advrows}.

\section{Conclusions and Future Directions}\label{sec:concl}

In this paper, we proposed a new error model, \emph{semi-adversarial errors}, for studying the efficient decodability of error-correcting codes. Bridging between the adversarial and fully random error models, our interpolation and analysis techniques leveraged ideas from both models to obtain new results in the semi-adversarial model. First, we improved the analysis of Bleichenbacher--Kiayias--Yung's decoding algorithm for interleaved Reed--Solomon codes to show that the algorithms achieves the optimal unique decoding radius for most combinations of adversarial and random errors in the semi-adversarial model. Second, we found a novel application of Alekhnovich's interpolation algorithm to extend BKY's interpolation algorithm to the settings of folded Reed--Solomon codes and multiplicity codes. By extending our analysis methods for interleaved Reed--Solomon codes, we were able to show near information-theoretically optimal decoding for folded Reed--Solomon and multiplicaty codes as well. We conclude the paper by discussing some directions for future research.

\paragraph{Coppersmith--Sudan in the semi-adversarial model} This paper primarily focused on understanding the limits of the Bleichenbacher--Kiayias--Yung algorithm \cite{bleichenbacher2003decoding} in the semi-adversarial model. However, contemporaneous to the BKY algorithm, Coppersmith and Sudan \cite{coppersmith2003reconstructing} independently came up with a different interpolation algorithm --- the CS algorithm --- for uniquely decoding interleaved Reed--Solomon codes. As discussed in the literature overview, the CS algorithm uses the interpolation matrix for a generalization of the Guruswami Sudan algorithm for interleaved Reed-Solomon codes. It has non-linear terms in the received word in this matrix unlike BKY which only has linear terms. This linearity was very helpful in the analysis of our algorithm, and the lack of it poses a challenged for analyzing CS in the semi-adversarial model. Furthermore, CS algorithm only decodes to the distance $1-R^{s/(s+1)}-R$, which is incomparable with BKY (CS beats the BKY rate $\frac{s}{s+1}(1-R)$ for small $R$ for a fixed $s$). \change{We leave open the question of exploring the limits of the CS algorithm in the semi-adversarial model.}

\paragraph{Unique Decoding up to $1-R^{s/(s+1)}$ distance under random errors}
A related problem would be to analyze the Multivariate Interpolation Decoding algorithm of Parvaresh and Vardy~\cite{Parvaresh2004MultivariateID}. They conjecture that it should be able to decode up to a $1-R^{s/(s+1)}$ distance with high probability. Power decoding algorithms~\cite{puchinger2017decoding} are also conjectured to decode with high probability to this radius. The CS algorithm is also conjectured to decode to a $1-R^{s/(s+1)}$ distance. Currently, no algorithm has been theoretically proven to achieve this guarantee. 

\paragraph{List-decoding in the semi-adversarial model}\change{In \cref{prop:list_optimal}, we proved a generalized Singleton bound for semi-adversarial errors (GSSB).} However, to understand the sharpness of the GSSB, one needs to determine whether there exist codes which attain the GSSB, particularly Reed--Solomon codes. In the adversarial case, through a series of papers, \cite{ST23,brakensiek2023generic,guo2023randomly,alrabiah2024randomly}, we now that the GSB can be attained by Reed--Solomon codes with random evaluation points, even when the field size is $O(n)$. For arbitrary subfield evaluation points (i.e., $s$-interleaved RS codes), we only know that $\left(\frac{s}{s+1}(n-k),q^{\gamma sn}\right)$ list decodability is possible \cite{guruswami2022optimal}, where $\gamma$ is a constant in $(0,1)$. However, a recent work  
by Chen and Zhang~\cite{cz24} shows that explicit $s$-folded Reed-Solomon codes are $\left(\frac{L}{L+1}\left(n-\frac{k}{s-L+1}\right),L\right)$ list-decodable. For random evaluation points, some broader families of codes also attain the GSB~\cite{BDG24,BDGZ24}. We leave extending these results to the GSSB as \change{an open question}.

{\bibliography{main}}

\appendix

\subsection{Equivalence Between RS Codes Over a Subfield and Interleaved RS Codes}\label{app:IRS-RS} Let $\F_q$ be a finite field and let $\alpha_1, \hdots, \alpha_n \in \F_q$ be distinct elements. For any $s \ge 1$, we can build two similar-looking codes. 

\begin{itemize}
\item A Reed--Solomon code $\mathcal C_{\RS} := \mathsf{RS}_{n,k,q^s}(\alpha_1, \dots, \alpha_n) \subseteq \F_{q^s}^n$, where we view $\F_q$ as a subfield of $\F_{q^s}$.
\item An interleaved Reed--Solomon code $\mathcal C_{\IRS} := \mathsf{IRS}_{n,k,q}(\alpha_1,\hdots, \alpha_n) \subseteq (\F_q^s)^n$.
\end{itemize}

We observe that these codes are (essentially isomorphic). In particular, let $\gamma$ be a generator of the extension $\F_{q^s} / \F_q$. Then, every $a \in \F_{q^s}$ has a unique representation of the form $a = a_1 + \gamma a_2 + \cdots + \gamma^{s-1} a_{s}$. Thus, we have a bijection $\psi : \F_q^s \to \F_{q^s}$ defined by $\psi(a_1, \hdots, a_s) = a_1 + \gamma a_2 + \cdots \gamma^{s-1} a_{s}$. Furthermore $\psi$ is an isomorphic between the Abelian groups underlying the additive structure of $\F_q^s$ and $\F_{q^s}$. We now show this isomorphism can be lifted between $\cC_{\RS}$ and $\cC_{\IRS}$.

\begin{proposition}\label{prop:irs-iso}
Let $\psi^n : (\F_q^s)^n \to \F_{q^s}^n$ be the coordinate-wise application of $\psi$. We have that $\psi^n$ is an isomorphism between $\cC_{\IRS}$ and $\cC_{\RS}$ in the following sense.
\begin{itemize}
\item[(a)] $\psi^n$ is an isomorphism between $\cC_{\IRS}$ and $\cC_{\RS}$ as Abelian groups.
\item[(b)] $\psi^n$ preserves Hamming distance between codewords.
\end{itemize}
\end{proposition}

\begin{proof}
Consider any polynomials $f_1, \hdots, f_s \in \F_q[x]$ of degree less than $k$. Note that we can view $f_1, \hdots, f_s \in \F_{q^s}[x]$. Thus, we can define $f := f_1 + \gamma f_2 + \cdots + \gamma^{s-1} f_s$, where $\deg f < k$. It is straightforward to verify that for every $i \in [n]$, we have that $\psi(f_1(\alpha_i), \hdots, f_s(\alpha_i)) = f(\alpha_i)$. Therefore since $\psi : \F_q^s \to \F_{q^s}$ is an isomorphism, we have that $\psi^n$ is an injective homomorphism from $\cC_{\IRS}$ to $\cC_{\RS}$. Since $\cC_{\IRS}$ and $\cC_{\RS}$ have the same cardinality of $q^{sk}$, this map is an isomorphism. This proves (a).

To prove (b), note that $\psi(0,\hdots, 0) = 0$ and every nonzero input to $\psi$ maps to a nonzero output. Therefore, $\psi^n$ preserves the Hamming distance between the $0$ string and every codeword (i.e., the Hamming weight). Since the distance between two codewords is the Hamming weight of their difference, the isomorphism preserves the Hamming distance between all pairs of codewords.
\end{proof}

As a result, we can prove for any semi-adversarial channel, which includes as special cases any an adversarial or random channel, that the behavior between $\cC_{\IRS}$ and $\cC_{\RS}$ is isomorphic.

\begin{proposition}
For any $0 \le e_0 \le e \le n$, consider an $(e_0, e)$ semi-adversarial channel $\Psi$ for $\cC_{\IRS}$. Then, there exists a corresponding $(e_0, e)$ semi-adversarial channel $\Psi'$ for $\cC_{\RS}$ such that $\vec{c} \in \cC_{\IRS}$ and $\vec{y} \in (\F_q^s)^n$ we have that
\begin{align}
    \Pr[\Psi(\vec{c}) &= \vec{y}] = \Pr[\Psi'(\psi^n(\vec{c})) = \psi^n(\vec{y})].\label{eq:IRS-RS-equiv}
\end{align}
Furthermore, this correspondence between channels is bijective.
\end{proposition}
Therefore, the decoding properties of $\cC_{\IRS}$ and $\cC_{\RS}$ are isomorphic. In particular, Theorem~\ref{thm:main_RS} applies equally to interleaved Reed--Solomon codes and Reed--Solomon codes with evaluation points over a subfield. 
\begin{proof}
We can view\footnote{Technically, a channel could be a probability distribution of such functions, but by linearity of expectation it suffices to look at one function at a time to prove (\ref{eq:IRS-RS-equiv}).} a $(e_0, e)$ semi-adversarial channel for $\cC_{\IRS}$ as a function $\widetilde{\Psi} : \cC_{\IRS} \to (\F_q^s \cup \{*\})^{n}$, where for every $\vec{c} \in \cC_{\IRS}$, we have that $\widetilde{\Psi}(\vec{c})$ has exactly $e - e_0$ coordinates equal to $*$ and at most $e_0$ other coordinates differ from $\vec{c}$. To turn the function $\widetilde{\Psi}$ into a proper channel $\Psi$, we then sample $e -e_0$ uniformly random elements of $\F_q^{s}$ which we fill into the positions marked $*$.

We can construct a corresponding function $\widetilde{\Psi}' : \cC_{\RS} \to (\F_{q^s} \cup \{*\})^n$ as follows:
\begin{align}
    \widetilde{\Psi}'(\vec{c}) := \psi^n(\widetilde{\Psi}((\psi^n)^{-1}(\vec{c}))),\label{eq:psi-corr}
\end{align}
where we define $\psi(*) = *$. We can turn $\widetilde{\Psi}'$ into a channel $\Psi'$ in an analogous manner by replacing the positions marked $*$ with a uniformly random element of $\F_{q^s}$. The map from $\widetilde{\Psi}$ to $\widetilde{\Psi}'$ is clearly reversible, so the correspondence between $\Psi$ and $\Psi'$ is bijective.

It suffices to prove (\ref{eq:IRS-RS-equiv}). For any $\vec{c} \in \cC_{\IRS}$ and any $\vec{y} \in (\F_q^s)^n$, we can compute that $\Pr[\Psi(\vec{c}) = \vec{y}] = 1/q^{s(e-e_0)}$ if $\widetilde{\Psi}(\vec{c})$ and $\vec{y}$ agree in all coordinates not marked $*$ and $0$ otherwise. Likewise,\footnote{Crucially we use the fact that both $\F_q^s$ and $\F_{q^s}$ have the same cardinality so that uniform distributions over each set are isomorphic.} $\Pr[\Psi'(\psi(\vec{c})) = \psi(\vec{y})] = 1/q^{s(e-e_0)}$ if $\widetilde{\Psi}'(\psi^n(\vec{c}))$ and $\psi^n(\vec{y})$ agree in all coordinates not marked $*$ and $0$ otherwise. By equation (\ref{eq:psi-corr}), we have that $\widetilde{\Psi}'(\psi^n(\vec{c})) = \psi^n(\widetilde{\Psi}(\vec{c}))$. Therefore since $\psi^n$ is a bijection between $(\F_q^s)^n$ and $\F_{q^s}^n$, $\widetilde{\Psi}'(\psi^n(\vec{c}))$ and $\psi^n(\vec{y})$ agree in all coordinates not marked $*$ if and only if $\widetilde{\Psi}(\vec{c})$ and $\vec{y}$ agree in all coordinates not marked $*$. Thus, (\ref{eq:IRS-RS-equiv}) holds.
\end{proof}
\end{document}